\documentclass[aps,prx,floatfix,nofootinbib,twocolumn,superscriptaddress,longbibliography]{revtex4-2}
\usepackage{dsfont}

\usepackage{amsmath}
\usepackage{amsfonts}
\usepackage{physics}
\usepackage{amssymb}
\usepackage{graphicx}
\usepackage[caption=false]{subfig}
\usepackage{color}
\usepackage{accents}
\usepackage{bbold}
\usepackage[colorlinks=true, citecolor=blue]{hyperref}
\usepackage[normalem]{ulem}
\usepackage{url} 
\usepackage{hyperref}
\usepackage{relsize}
\usepackage{multirow}
\usepackage{bm}

\usepackage{tikz}
\usetikzlibrary{quantikz}

\makeatletter %
\newsavebox{\@brx}
\newcommand{\llangle}[1][]{\savebox{\@brx}{\(\m@th{#1\langle}\)}%
  \mathopen{\copy\@brx\kern-0.5\wd\@brx\usebox{\@brx}}}
\newcommand{\rrangle}[1][]{\savebox{\@brx}{\(\m@th{#1\rangle}\)}%
  \mathclose{\copy\@brx\kern-0.5\wd\@brx\usebox{\@brx}}}
\makeatother

\newlength{\dhatheight} %

\newtheorem{theorem}{Theorem}[section]

\newtheorem{corollary}[theorem]{Corollary}

\newcommand{\process}{\mathcal{E}}

\newenvironment{proof}[1][Proof]{\begin{trivlist}
\item[\hskip \labelsep {\bfseries #1}]}{\end{trivlist}}

\newcommand{\qed}{\nobreak \ifvmode \relax \else
      \ifdim\lastskip<1.5em \hskip-\lastskip
      \hskip1.5em plus0em minus0.5em \fi \nobreak
      \vrule height0.5em width0.5em depth0.15em\fi}

\newcommand\varpm{\mathbin{\vcenter{\hbox{%
  \oalign{\hfil$\scriptstyle+$\hfil\cr
          \noalign{\kern-.3ex}
          $\scriptscriptstyle({-})$\cr}%
}}}}
\newcommand\varmp{\mathbin{\vcenter{\hbox{%
  \oalign{$\scriptstyle({+})$\cr
          \noalign{\kern-.3ex}
          \hfil$\scriptscriptstyle-$\hfil\cr}%
}}}}

\long\def\/*#1*/{}

\begin{document}
\title{Classical Shadows for Quantum Process Tomography on Near-term Quantum Computers}
 
\author{Ryan Levy}
\thanks{Co-first authors.}
\affiliation{Institute for Condensed Matter Theory and IQUIST and NCSA Center for Artificial Intelligence Innovation and Department of Physics, University of Illinois at Urbana-Champaign, IL 61801, USA} 
\affiliation{Center for Computational Quantum Physics, Flatiron Institute, New York, NY, 10010, USA}
\author{Di Luo}
\thanks{Co-first authors.}
\affiliation{Institute for Condensed Matter Theory and IQUIST and NCSA Center for Artificial Intelligence Innovation and Department of Physics, University of Illinois at Urbana-Champaign, IL 61801, USA} 
\affiliation{The NSF AI Institute for Artificial Intelligence and Fundamental Interactions}
\affiliation{Center for Theoretical Physics, Massachusetts Institute of Technology, Cambridge, MA 02139, USA}
\affiliation{Department of Physics, Harvard University, Cambridge, MA 02138, USA}
\author{Bryan K.  Clark}
\affiliation{Institute for Condensed Matter Theory and IQUIST and NCSA Center for Artificial Intelligence Innovation and Department of Physics, University of Illinois at Urbana-Champaign, IL 61801, USA}

\date{\today}
\begin{abstract}
Quantum process tomography is a powerful tool for understanding quantum channels and characterizing properties of quantum devices. Inspired by recent advances using classical shadows in quantum state tomography [H.-Y. Huang, R. Kueng, and J. Preskill, Nature Physics \textbf{16}, 1050 (2020)],  we have developed ShadowQPT, a classical shadow method for quantum process tomography.  We introduce two related formulations with and without ancilla qubits.  ShadowQPT stochastically reconstructs the Choi matrix of the device allowing for an a-posteri classical evaluation of the device on arbitrary inputs with respect to arbitrary outputs.  Using shadows we then show how to  compute overlaps, generate all $k$-weight reduced processes,  and perform reconstruction via Hamiltonian learning.  These latter two tasks are efficient for large systems as the number` of quantum measurements needed scales only logarithmically with the number of qubits. 
A number of additional approximations and improvements are developed including the use of a pair-factorized Clifford shadow and a series of post-processing techniques which significantly enhance the accuracy for recovering the quantum channel.  We have implemented ShadowQPT using  both Pauli and Clifford measurements on the IonQ trapped ion quantum computer for quantum processes up to $n=4$ qubits and achieved good performance.

\end{abstract}
\maketitle

\section{Introduction}
Quantum technologies have been developing rapidly in recent years. Quantum devices are used for various applications, such as quantum metrology, quantum teleportation, and quantum simulation~\cite{Bouwmeester1997,Giovannetti_2011,Arute2019,Zhong1460,Ebadi_2021}. One important task in the development of near-term quantum devices is their characterization. 

For Noisy Intermediate-Scale Quantum (NISQ) technology, a first step to characterization is to use quantum state tomography (QST) which attempts to characterize the output state of the quantum circuit.
A next crucial step, quantum process tomography (QPT), characterizes not just one output state but the entire quantum dynamics of the device. 
Earlier attempts at QPT used the linear inversion method~\cite{linear,PhysRevLett.78.390}.  Later various statistical methods were developed including maximum likelihood methods~\cite{PhysRevA.64.052312,Bouchard_2019,Lvovsky_2004,PhysRevLett.105.200504,PhysRevLett.108.070502}, Bayesian methods~\cite{Granade_2017,PhysRevLett.102.020504,Blume_Kohout_2010}, compressed sensing methods~\cite{Rodionov_2014}, tensor network methods~\cite{torlai2020quantum} and other optimization techniques~\cite{Knee_2018,surawystepney2021projected,zhang2021quantum,white2021nonmarkovian,Govia2020,xue2021variational,onorati2021fitting}. Theoretically, quantum process tomography can be related to quantum state tomography through the Jamio\l{}kowski process-state isomorphism~\cite{JAMIOLKOWSKI1972275,CHOI1975285}.
One recent important advancement in QST comes from classical shadow tomography~\cite{Huang2020}, which allows for the prediction of multiple observables with few quantum measurements both theoretically and experimentally~\cite{hadfield2021adaptive,chen2021robust,koh2020classical,hillmich2021decision,struchalin2021experimental,zhang2021experimental,zhao2020fermionic,hu2021hamiltoniandriven,acharya2021informationally,rouze2021learning,hadfield2020measurements,huggins2021unbiasing, hu2021classical,hu2021hamiltoniandriven,huang2021provably}.

In this work, we propose a classical shadows algorithm for quantum process tomography on near-term quantum computers giving both a theoretical analysis and explicit implementations on an IonQ quantum computer. Using Jamio\l{}kowski process-state isomorphism and techniques in process tomography, we develop the first classical shadow quantum process tomography (ShadowQPT). 

In Section II, we prove two theorems related to the effectiveness of ShadowQPT.   
Theorem~\ref{thm:overlap} bounds the amount of ShadowQPT data needed to compute the overlap between any target density
matrix $\sigma$ with the output of the measured quantum process on any input $\rho$.
Theorem~\ref{thm:k-qubit} discusses the complexity of process matrix tomography. 
In Section III, we show a corollary~\ref{thm:ham_learn},  which describes the sample complexity scaling of applying ShadowQPT to Hamiltonian learning, which is logarithmic in system size, following from Theorem~\ref{thm:k-qubit}. 
In Section IV, we develop shadow algorithms with both Pauli and Clifford measurements.  Both an ancilla-based scheme as well as a two-sided scheme which applies unitaries both before and after the channel are introduced.   We proceed to show practical improvements to these algorithms for near term devices by using two qubit Clifford unitaries instead a global Clifford and using multiple repetitions per circuit. 
Furthermore, we develop a series of post-processing techniques such as projecting into the space of physical channels and purifying outputs which significantly improve our results. 
In Section V, we implemented our algorithms on IonQ quantum computers for both unitary and non-unitary process tomography on $n=2,3,4$ qubits systems. 
We benchmark both the process matrix construction as well as test predicting the overlap for pairs of input and output pure states and compare to direct measurement outcomes. 
In Section VI, we numerically simulate Hamiltonian learning and reconstruct a random 1D Ising model using ShadowQPT, showing efficient scalability with system size. 
We conclude in Section VII with discussion on opportunities and future explorations of our ShadowQPT algorithms.

\section{Theoretical Analysis of ShadowQPT}

In this work, we have developed shadow tomography algorithms for quantum process tomography. The key idea is based on the recent development of shadow tomography on quantum state~\cite{Huang2020} and the Jamio\l{}kowski process-state isomorphism~\cite{JAMIOLKOWSKI1972275}.  An important result in classical shadow tomography states that~\cite{Huang2020}: Classical shadows of size $N$ suffice to predict $M$ linear functions $\{tr(O_i \rho)\}_{i=1,\dots,M}$ up to additive error $\epsilon$ given that $N \geq$ (order) log($M$) $\max_{i}||O_i||^2_{shadow}/ \epsilon^2$. 

To connect quantum process tomography with the quantum state tomography, we utilize the Jamio\l{}kowski process-state isomorphism~\cite{JAMIOLKOWSKI1972275}. It provides a positive semi-definite operator representation $\Lambda_{\process}$ for a $n$ qubits process $\process$, which is the Choi matrix representation~\cite{CHOI1975285}

\begin{equation}\label{eq:choi}
   \Lambda_{\process} = (\mathds{I} \otimes \process)(\ket{\phi^{+}}\bra{\phi^{+}}^{\otimes n}) 
\end{equation}
where $\ket{\phi^{+}}=\ket{00}+\ket{11}$. For an input density matrix $\rho$ to the channel $\process$, we have $\process(\rho) = \textup{tr}[(\rho^T \otimes \mathds{I}) \Lambda_\process ]$. 
In the following of the paper, we implicitly denote $\Lambda_{\process}$ as $\Lambda$. Notice that $\textup{tr}(\Lambda_{\epsilon})=2^n$ and we further denote the normalized process matrix $\Lambda$ with trace 1 as $\rho_{\Lambda} = \frac{1}{2^n} \Lambda$. We notate a reduced $k$-qubit ($k \leq n$) process Choi matrix as the Choi matrix representation for the quantum channel on a particular $k$-qubit subsystem where the other $n-k$ qubits are traced out.

Based on the classical shadow tomography property and the Choi matrix representation, we state the following two informal versions of the theorems for classical shadow tomography for quantum process.

\begin{theorem}\label{thm:k-qubit}

For a $n$-qubit process Choi matrix $\rho_{\Lambda}$ and $\epsilon,\delta \in (0,1)$, the number of random global Clifford measurements $N$ that suffices to simultaneously predict any reduced $k$-qubit process Choi matrix $\rho_{\Lambda^{(k)}}$ with

 \begin{itemize}
    \item Frobenius norm error up to $\epsilon$ with probability $1-\delta$ is of order $\frac{4^{n+k}}{\epsilon^2} \textup{log}(2 (8n)^{2k} / \delta)$
\end{itemize}

Meanwhile, the number of Pauli-6 POVM measurements $N$ that suffices to simultaneously predict any reduced $k$-qubit process Choi matrix $\rho_{\Lambda^{(k)}}$ with

\begin{itemize}
    \item Frobenius norm error up to $\epsilon$ with probability $1-\delta$ is of order $\frac{36^k}{\epsilon^2} \textup{log}(2 (8n)^{2k} / \delta)$
    \item Trace norm error up to $\epsilon$ with probability $1-\delta$ is of order $\frac{144^k}{\epsilon^2} \textup{log}((24n)^{2k}/\delta)$
\end{itemize}
\end{theorem}

The Pauli-6 POVM measurement includes $\{\frac{1}{3}\ket{0}\bra{0},\frac{1}{3}\ket{1}\bra{1},\frac{1}{3}\ket{+}\bra{+},\frac{1}{3}\ket{-}\bra{-},\frac{1}{3}\ket{r}\bra{r},\frac{1}{3}\ket{l}\bra{l}\}$, where $(\ket{0},\ket{1}),(\ket{+},\ket{-}),(\ket{r},\ket{l})$ are eigenvectors of $\sigma_z$, $\sigma_x$ and $\sigma_y$. We note that the Pauli-6 POVM measurement and the random single qubit Clifford measurement usually share similar results since the two measurements have the same shadow norm for factorized Pauli observables~\cite{Huang2020,acharya2021informationally}. In this work, the term `random Pauli measurement' or `Pauli' in the figures refer to Pauli-6 POVM measurement, though one can also apply random single qubit Clifford in those contexts.

Notice that for any fixed $k$, the shadow methods allow simultaneous prediction of all ${\binom {n}{k}}$ reduced $k$-qubit processes in time logarithmic in $n$.   This is useful for learning quantum dynamics of local observables, and provides a foundation for us to develop an efficient k-local Hamiltonian learning scheme (see Sec. \ref{sec:HamLearnAlg}).
For full process tomography ($k=n$), Theorem~\ref{thm:k-qubit} states that the quantum measurement complexity of ShadowQPT scales exponentially which is consistent with the known lower bound.
We note that each measurement of ShadowQPT gives a full unbiased (potentially noisy) stochastic representation of the entire Choi matrix; this stands in contrast with standard MLE which requires exponential classical post-processing to generate the Choi matrix even from small amounts of data.

Theorem \ref{thm:k-qubit} focuses on predicting the complete set of reduced $k$-qubit process Choi matrices. Instead, one can ask about computing the result of the quantum channel run on a series of (input $\rho$, output $\sigma$) pairs - i.e. the quantum device run on the density matrix $\rho$ and traced against the observable $\sigma$.

\begin{theorem}\label{thm:overlap}
For a $n$-qubit quantum process $\rho_{\Lambda}$ and $\epsilon,\delta \in (0,1)$, given a set of density matrix pairs $\{(\rho^{in}_1,\sigma_1), \dots,(\rho^{in}_{M},\sigma_M)\}$, the number of measurements $N$ that suffices to predict $\textup{tr}((\rho^{in\,T}_i \otimes \sigma_i)\rho_{\Lambda})$ for any $i$ up to error $\epsilon$ with probability $1-\delta$ is of order

\begin{equation}
    \frac{ \textup{log}(2M/\delta)}{\epsilon^2} \textup{max}_{i} ||O_{i} - \frac{\textup{tr}(O_{i})}{2}\mathds{I}||^2_{shadow}
\end{equation}
where $O_{i}=\rho^{in}_i \otimes \sigma_i$,  $|| \cdot ||_{shadow}$ is the shadow norm (see \hyperref[thm:shadow-formal]{Appendix Theorem A.1} for definition). 

For random global Clifford measurement,  it requires order $\frac{\textup{log}(2M)}{\epsilon^2}  \textup{max}_{i} \textup{tr}(O_{i}^2)$. In particular, if $\rho^{in}_i$ and $\sigma_i$ are all pure states, then $\textup{tr}(O_{i}^2)=1$. For random single qubit Clifford measurement, it requires order $\frac{\textup{log}(2M) }{\epsilon^2} \textup{max}_{i} 4^{k_{i}} ||O_{i}||_{\infty}^2$
,where $O_{i}$ acts nontrivially on $k_i$-qubits, $||\cdot||_{\infty}$ is the spectral norm. If $\rho^{in}_i$ and $\sigma_i$ are all pure states, then $||O_{i}||_{\infty}^2=1$.

\end{theorem}

If $\sigma_i$ and $\rho_i^{in}$ are $k$-qubit reduced density matrices with the same support, then the overlap between $\sigma_i$ and $\rho_i^{in}$ through the channel corresponds to $\textup{tr}(\process(\rho^{in}_i)\sigma_i)=2^k \textup{tr}(\rho_{\Lambda^{(k)}}(\rho^{in\,T}_i \otimes \sigma_i))$. 
Note that the Clifford ShadowQPT allows us to compute the overlap  of pairs of pure states $(\rho^{in}, \sigma)$ using a number of measurements which scales logarithmically with the number of pairs being considered.  If $\rho^{in}$ and $\sigma$ are stabilizer states, the classical resources required to compute are also efficient as one can use the stabilizer algorithm~\cite{Aaronson_2004}. 

Notice to compute $M$ pairs of $\textup{tr}(\process(\rho^{in}_i)\sigma_i)$, ShadowQPT requires time independent of $n$ and scales as $\textup{log}(M)$, when  $\rho^{in}_i$ and $\sigma_i$ have support over a $k$-qubit system.
The detailed proofs for  Theorems~\ref{thm:k-qubit} and ~\ref{thm:overlap} and are given in Appendix~\ref{app:proof}.

\section{Hamiltonian Learning via ShadowQPT}
\label{sec:HamLearnAlg}

A corollary of Theorem~\ref{thm:k-qubit} implies that certain applications of ShadowQPT can also be efficient. One such application is Hamiltonian learning \cite{Li2020,haah2021optimal}, which seeks to determine an unknown Hamiltonian given access to a time evolution operator $U(t)$.
 In this context, we consider applications of ShadowQPT to Hamiltonian learning. For a $k$-local Hamiltonian $H = \sum_i c_i h_i$,
$h_i$ is a $k$-body operator as a Pauli string and $c_i$ is the coefficient to learn.

 One scheme that we develop is to learn $c_i$ through ShadowQPT on $e^{-iHt}$.  
Notice that access to $e^{-iHt}$ also gives us an ability to time evolve to all $t'$ which are integer multiples of $t$.
Consider $t$ sufficiently small such that $U = e^{-iHt} \approx \mathrm{I} - iHt = \mathrm{I}-it\sum_i c_i h_i$, then $U$ can be approximated by a sum of $k$-local operators.  
To characterize $U$,
we first perform ShadowQPT on $U$ and attain a representation of $\rho_{\Lambda^{(k)}}$. To learn $c_i$,we set $\Lambda(\rho) = U \rho U^\dagger$ and we choose a $k'$-local operator $p_i$ where $k'<k$ such that $q_i \equiv [h_i,p_i] \neq 0$, and then classically evaluate

\begin{align}
  \text{tr}(\Lambda(p_i) q_i) &= \text{tr}((p_i - it [H,p_i]) q_i) + O(t^2) \\
  &= \text{tr}((p_i - it \sum_j  c_j q_j) q_i) + O(t^2); \\
 \text{tr}((p_i^T \otimes q_i)\rho_{\Lambda}) &= -it (c_i + \sum_{j\not= i} g_j c_j) + O(t^2)
\end{align}

where $g_j= \text{tr}(q_j q_i)/\text{tr}(q_i q_i) = \text{tr}(q_j q_i)/2^n$. Hence we can attain $c_i$ from evaluating data from ShadowQPT and solving a system of equations. Note that if $\text{tr}(q_j q_i)=\delta_{ij}$ $\forall i,j$ then $\text{tr}((p_i^T \otimes q_i)\rho_{\Lambda}) = -it c_i + O(t^2)$. We focus on this case, which is true for e.g. the transverse field Ising model.
According to Theorem.~\ref{thm:k-qubit}, the sampling complexity $S$ of estimating $c_i$ is given by the following.

\begin{corollary}\label{thm:ham_learn}
Consider an unknown $k$ local Hamiltonian $H=\sum_i c_i h_i$ where $h_i$ is a $k$-body operator and $c_i$ are real coefficients. Given access to  $U(t)=e^{iHt}$, then ShadowQPT can estimate $c_i$ up to error $\epsilon$ with probability $1-\delta$ with sample complexity $S$
\begin{equation}
    S \sim \frac{36^k}{t^2 \epsilon^2} 2k \log(8n) \log (1/\delta),
\end{equation}
where $t$ is sufficiently small that the linear approximation error is much smaller than $\epsilon$. 
\end{corollary}

Thus for fixed $k$, ShadowQPT can efficiently predict $c_i$ in $O(\log(n))$ samples.
Note that one can systematically reduce the (currently linear) approximation error by a higher order Taylor expansion at the cost of additional classical processing.  We also note that our approach is related to the latest version of Ref.~\cite{haah2021optimal} on real-time dynamics of Hamiltonian learning, which is proved in an alternative way by extending the results of finite-temperature Hamiltonian learning.

\section{ShadowQPT Post-processing Algorithms}
\newcommand*{\lamGate}{\mathlarger{\mathlarger{\mathlarger{ \Lambda}}}}
\newcommand*{\UGate}{\mathlarger{\mathlarger{\mathlarger{U}}}}
\begin{figure}[h]
    \centering
    \subfloat[]{
        \includegraphics[width=0.95\columnwidth]{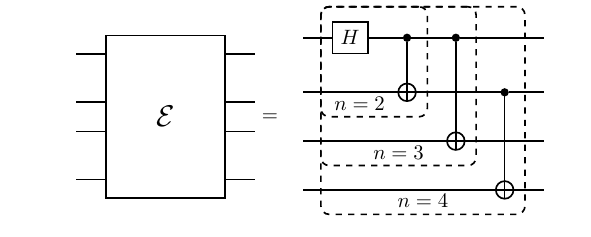}
    \label{fig:process}
    }\\
     \subfloat[]{
    \includegraphics[width=0.95\columnwidth]{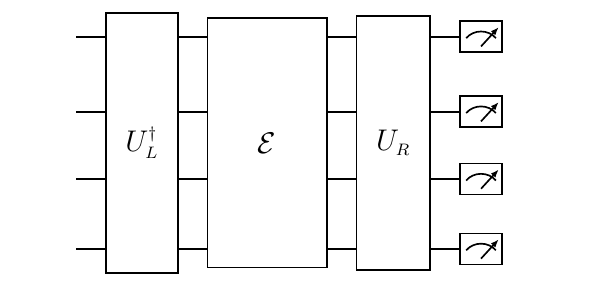}
    \label{fig:two_sided}
    }\\
    \subfloat[]{
    \includegraphics[width=0.9\columnwidth]{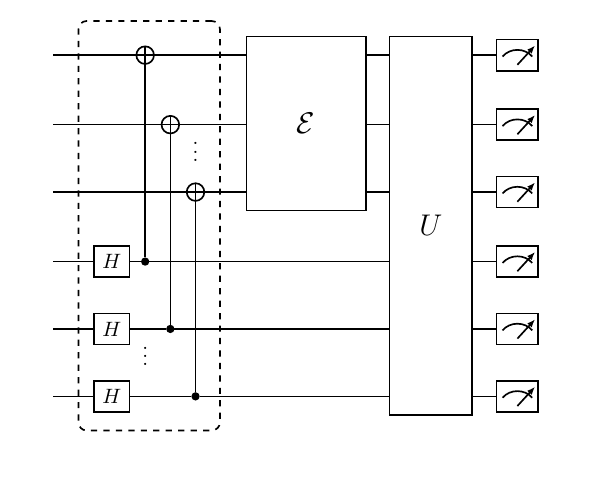}
    \label{fig:choi}
    }

    \caption{(a) Unitary channel $\process$, which is the GHZ process that can create a GHZ state for $n$ qubits from $\ket{0}^{\otimes n}$.
    (b) Circuit for two-sided quantum process tomography to produce a Choi matrix $\Lambda$ of a quantum channel $\process$. The application of two unitary circuits, denoted $U_L^\dagger$ and $U_R$ (e.g. random Pauli/Clifford) can be applied to reconstruct the channel.
    (c) Circuit for ancilla-based quantum process tomography to produce a Choi matrix $\Lambda$ of a quantum channel $\process$. After preparation of a bell state input (dashed box) to $\process$, a unitary rotation $U$ (e.g. random Pauli/Clifford) is applied before measurement. 
   }
    
\end{figure}
In this work, we perform classical shadow process tomography algorithms on a unitary process as well as non-unitary reduced process. We exemplify our approach on the unitary channel $\Lambda$ which is the GHZ process and its reduced process. The GHZ process $\Lambda$ is shown in Fig.~\ref{fig:process} and it generates a $n$-qubit GHZ state \cite{GHZ} $|\textrm{GHZ}_n\rangle = (|0\rangle^{\otimes n} + |1\rangle^{\otimes n})/\sqrt{2}$ from $\ket{0}^{\otimes n}$. The GHZ processes are constructed such that the circuits have a depth of $2,3,3$ for $n=2,3,4$ qubits respectively \cite{Cruz2018}.

\subsection{Unitary full process classical shadows}

In this section, we are going to introduce two schemes for ShadowQPT, the two-sided scheme and the ancilla-based scheme. 
The two sided scheme is shown in Fig.~\ref{fig:two_sided}, where random unitary circuits $U^\dagger_L$ and $U_R$ are applied to the left and right hand side of the channel.  For the ancilla-based scheme, we use the circuit shown in Fig.~\ref{fig:choi}.  
This circuit produces an $2n$ qubit output density matrix which represents the $n$ qubit normalized Choi matrix $\rho_\Lambda$, for the $n$ qubit process $\process$.
This approach is often called ancilla-assisted quantum process tomography~\cite{PhysRevLett.90.193601} (AAPT). 
In our work, for the ancilla-based approach we perform shadow tomography on this $2n$ qubit state. 
It turns out that the two-sided scheme can be viewed as a special case of the ancilla-based scheme, where the random unitary in the ancilla-based is factorized into a product of two unitaries $U_R$ and $U_L$ which respectively act on the top and bottom $n$ wires each. The details for the proof  of equivalence  can be found in Appendix.~\ref{app:equivalence}.

Both schemes can realize a random Pauli measurement scheme. 
For the ancilla-based scheme, the Pauli-6 POVM measurement can be realized by randomly applying a series of basis rotation gates chosen from $G_{R}=\{I, H,SH\}$ as the unitary $U$ after the channel and measuring in the computational basis, where $H$ is the Hadamard gate, $S$ is the phase gate and $I$ is the identity. 
For the two-sided scheme, one can realize the Pauli-6 POVM measurements by randomly selecting $U_L,U_R$ as a tensor product from $G=\{I, H, SH, X, HX, SHX\}$ where $X$ is the Pauli X gate.
In practice, this selection from $G$ only needs to be done for $U_L$ and for $U_R$ we can still randomly draw from $G_R$, just as the ancilla-based scheme, since the application of the last $X$ is irrelevant if a measurement is immediately performed afterwards. 
It is also possible to replace the Pauli-6 POVM measurements with random single qubit Clifford gates ($k=1$). This may be advantageous on certain quantum devices \cite{zhang2021experimental}. We numerically simulate their equivalence in Appendix Fig.~\ref{fig:two_sided_plots_l2}.

Additionally, both schemes can utilize random Clifford measurements, which can have better performance or scaling compared to Pauli measurements. Since most NISQ devices have limited gate depth, a global $n$ or $2n$ qubit random Clifford unitary needs gate depth $O(n^2/\log n)$ in general, which could be challenging to implement accurately; instead, we use a tensor product of 2-qubit Clifford unitaries ($k=2$) to realize the improved scaling of the Clifford group while minimizing the additional gate depth to the circuit. Random Clifford circuits are generated via the Qiskit library \cite{Qiskit}, which uses an optimal number of Control-Not gates for 2 qubits \cite{BravyiClifford}.

\begin{figure*}
    \centering
    \includegraphics[height=4.8cm]{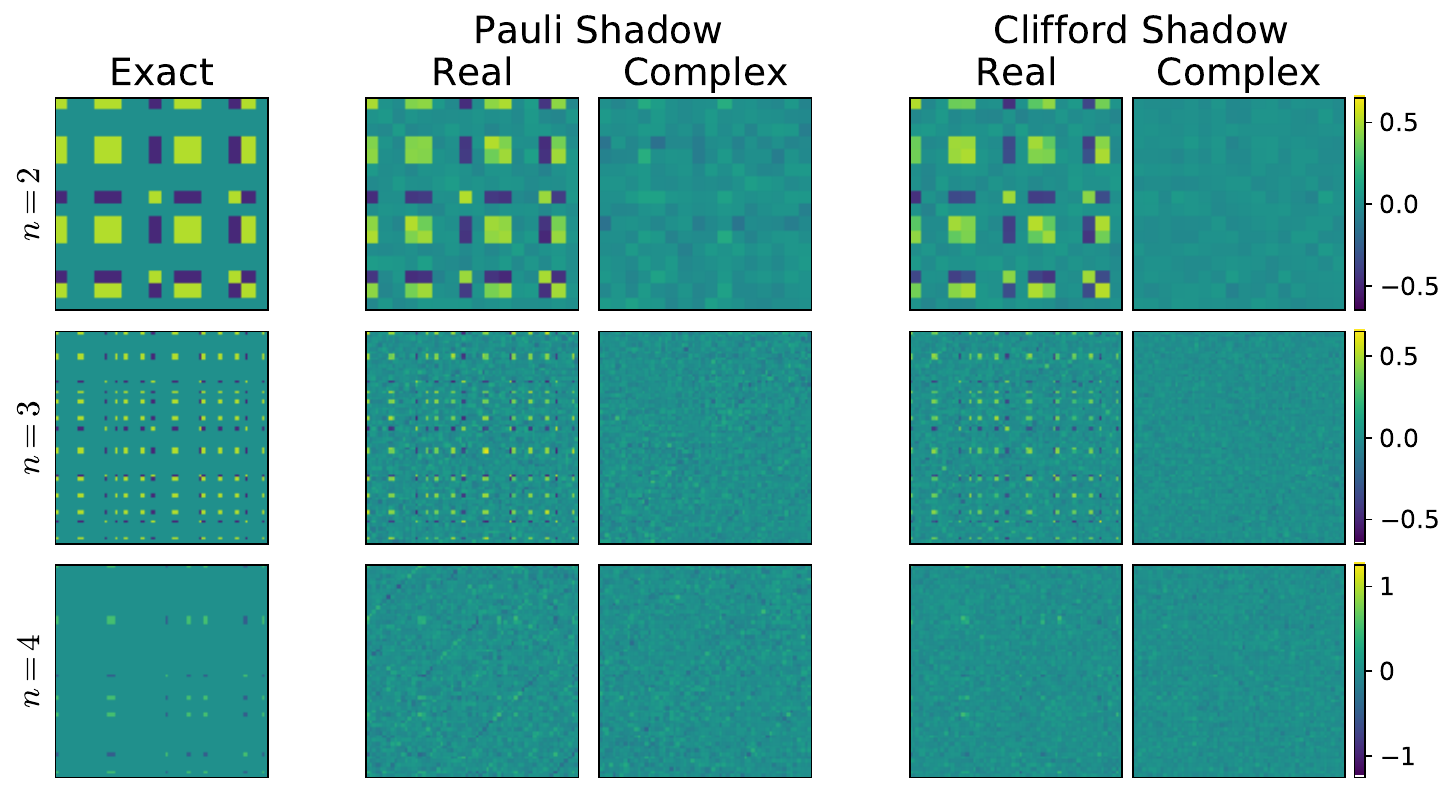}
    \includegraphics[height=4.8cm,trim={5cm 0 0 0},clip]{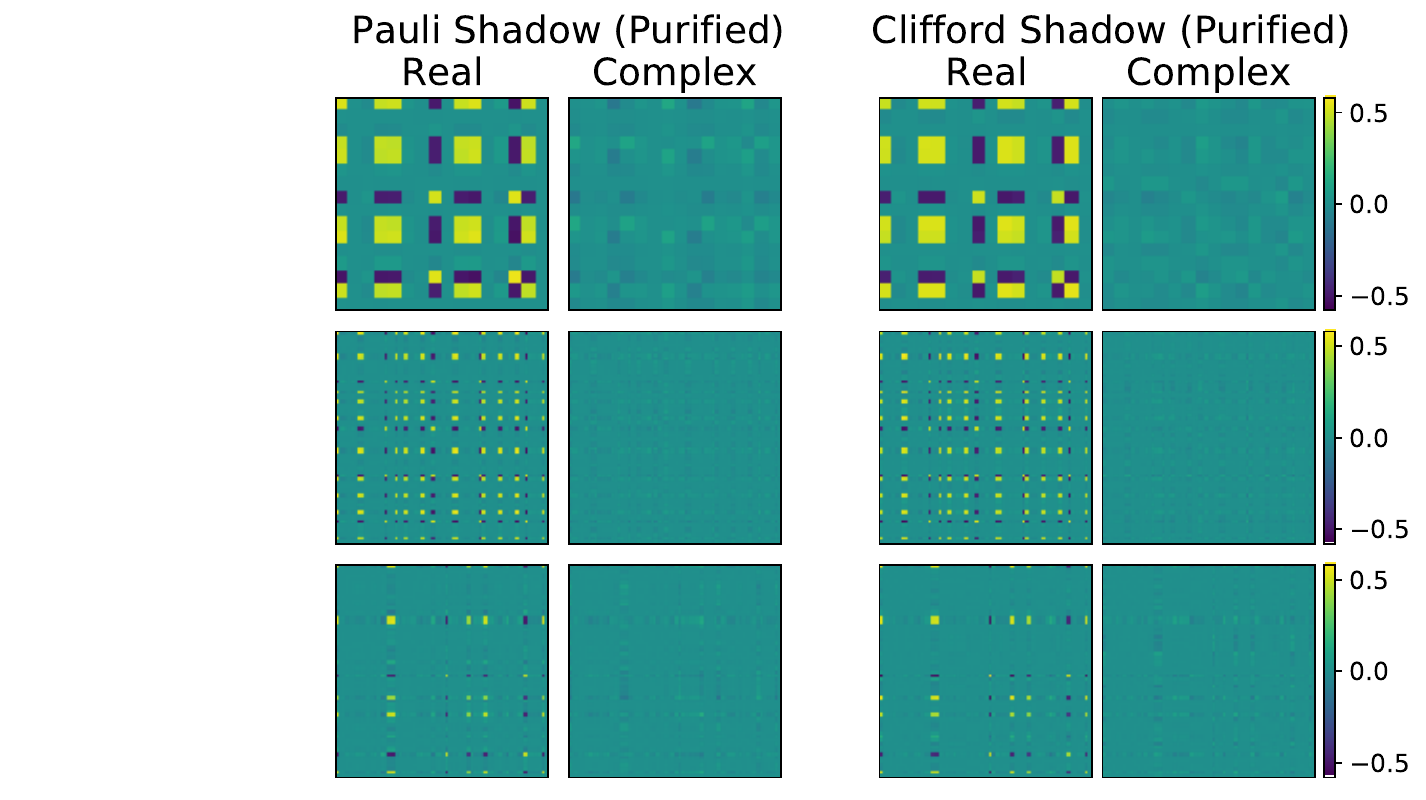}
    \caption{Visualization of the GHZ process Choi matrix, for various numbers of qubits (rows) constructed exactly or via classical shadow quantum process tomography. The first column corresponds to the exact unitary process $\Lambda$, which has no imaginary component. The next columns correspond to real and imaginary parts of the Pauli/Clifford shadow reconstructed Choi matrix $\Lambda^O$ followed by their corresponding purification. The $N=4$ visualization is zoomed into the first $60 \times 60$ elements.  }
    \label{fig:process_visualization}
\end{figure*}

After we attain the random Pauli/Clifford measurements, we can construct a Choi matrix from the classical shadows (denoted \textit{Pauli} or \textit{Clifford} in figures). Precisely, after applying a Unitary $U$ and measuring in the computational basis a bitstring $|b_i\rangle \in \{0,1\}^{\otimes n}$, we form the $i$-th classical shadow $\hat{\Lambda}_i$ as 
\begin{align}
    \hat{\Lambda}_i = 2^n \mathcal{M}^{-1}\left(U_i^\dagger|b_i\rangle\langle b_i|U_i\right), \label{eq:shadow}
\end{align}
where the inverse channel is $\mathcal{M}_n^{-1} (X) = (2^n+1)X-\mathds{I}$ for $n$-qubit Clifford circuits or $\mathcal{M}_n^{-1} (X) = \otimes_i \mathcal{M}_k^{-1}(X)$ for Pauli ($k=1)$ and 2-qubit Clifford ($k=2$) circuits. These objects can be held in memory efficiently, as Clifford circuits and bitstrings have a polynomial representation as well as polynomial evaluation cost~\cite{Aaronson_2004}.  We average over all classical shadows  $\hat{\Lambda}_i$ to obtain the final shadow reconstructed Choi matrix $\Lambda^O$. Note that it is suggested in Ref.~\cite{Huang2020} to utilize a median-of-means rather than direct mean. Our data has little dependence on this difference after projection (see Supplementary Fig.~\ref{fig:median_vs_mean}), which agrees with the observations of Ref.~\cite{struchalin2021experimental}.

We add an additional extension in which a given unitary is repeatedly measured under some number of repetitions. In practice, it is often the case that loading a different unitary on a quantum computer is significantly slower than taking additional measurements over the same unitary. The above theorems do not directly apply in this case, but we show via simulations (see supplementary Fig.~\ref{fig:trace_distance_vs_shots}) that including additional repetitions improves the quality of the result, up to a saturation threshold.

\subsection{Non-unitary reduced process classical shadows}

An interesting additional extension is the ability to characterize a non-unitary quantum process. By separating the $n$ qubit process into subsystems of size 1 and up to size $n-1$ , we can produce a reduced process on a subsystem of up to $n-1$ qubits. All subsets of the GHZ process are then non-unitary. Note that a general process $\Lambda_\process$ need not be unitary, nor have all subsystems be non-unitary. In this work, we focus on a unitary GHZ process $\Lambda$ and simultaneously measuring all corresponding non-unitary reduced processes on each subsystems. 

In a Pauli measurement based tomography scheme, including Pauli shadow tomography, we can directly use the operators that act over the relevant subsystem. Here, the partial trace over the Choi matrix is equivalent to this direct scheme, although for many expensive post-processing schemes it is far easier to optimize in the subsystem space directly.

For a Clifford shadows scheme however, a general $n$-qubit Clifford (or $2n$) will not be separable in the same manner as a Pauli based scheme. Despite this, if a particular subsystem or set of subsystems is known before hand, the Clifford unitaries can be chosen to respect separability. It has not been well studied, however, the trade off between the benefits of using non-separable large $n$-qubit Clifford gates or using separable ones.

\subsection{Post-processing of classical shadows}

A physical quantum process $\Lambda$ is a completely positive trace-preserving ($\mathcal{CPTP}$) map between density matrices. For any positive semi-definite density matrix $\rho^{in}$ with trace one, $\process(\rho^{in})$ should also be positive semi-definite with unit trace. Even though the shadow reconstructed Choi matrix in Eq.~\ref{eq:choi} has the correct trace, it is not a valid $\mathcal{CPTP}$ map in general. We further explore various post-processing projection techniques for improving the classical shadow reconstruction in this section. 

We denote the projection of $\Lambda^O$ as $\Lambda^\prime$. To project $\Lambda^O$ into the space of completely positive preserving ($\mathcal{CP}$) matrices, we use the technique of Ref.~\cite{Smolin2012}, which removes negative eigenvalues and rescales the remaining non-negative eigenvalues. We find this operation has the biggest reduction of trace distance from the target density matrix in our data. 
This does not guarantee that $\Lambda^\prime$ is $\mathcal{CPTP}$, but ensures that the resulting $\Lambda'$ is normalized and positive semi-definite.

In cases where we know the true underlying process should be unitary, we can  remove extraneous statistical noise, by purifying the shadow reconstructions to create a unitary shadow Choi matrix. Purifying in this manner can also be seen as an extreme $\mathcal{CP}$-projection method, in which all but the dominant eigenvector are thrown out. We find this approach significantly improves our result in the case where the target process is indeed unitary.

\subsection{Predicting Overlaps}

Given a set of density matrix pairs $(\rho^{in}_i,\sigma_i)$ as in Theorem~\ref{thm:overlap}, we can utilize shadow reconstructed Choi matrices to obtain predictions of the overlap $\textup{tr}[\process(\rho_i^{in})\sigma_i]$;  we work with the pure state $\sigma_i(\theta) = U(\bm{\theta})^\dagger|0\rangle$  
where $\theta$ is a set of parameters which specify the quantum circuit which generates $U$ and we use $|0\rangle$ as a placeholder for $|0\rangle^{\otimes n}$.

When evolving an initial density matrix $\rho^{in}$ through the process generated by $\Lambda^O$ to produce $\rho^{out}$, we additionally project a $\mathcal{CP}$-projected Choi matrix into the trace preserving ($\mathcal{TP}$) space using the method outlined in Ref.~\cite{knee2018quantum}, followed by projecting $\rho^{out}$ into $\mathcal{CP}$ space. Precisely, we have $\Lambda^\prime = \mathcal{TP}(\mathcal{CP}[\Lambda^O])$ and $ \rho^{out}=\mathcal{CP}\left[ \textup{tr}( ( \rho^{in\, T} \otimes  \mathds{I})) \Lambda^{\prime}\right]$. 

We then take the pure state density matrix $\sigma_i$ to calculate the overlap $\textup{tr}[\rho^{out}_i\sigma_i] = \langle 0|U(\bm{\theta}) \rho^{\prime out}_i U(\bm{\theta})^\dagger|0\rangle$.   For this to be classically efficient (even without postprocessing),  $\sigma_i$ has to be representable either as a short tensor network for Pauli or be within the Clifford group for Clifford unitaries respectively. 

When scaling to large numbers of qubits, the post-processing required to generate $\Lambda^\prime$ becomes infeasible, so we introduce a more efficient routine. To compute the overlap, we consider $\textup{tr}[ (\rho_i^{in\,T}\otimes\sigma_i) \,\Lambda'^O]$ where $\Lambda'^O$ here represents a purified $\Lambda^O$ without any additional projection.

\section{Experimental Results of ShadowQPT on Quantum Hardware}\label{sec:ExpRes}

We exemplify our ShadowQPT methods via quantum circuit measurements on the IonQ trapped ion quantum device. 
The IonQ system consists of 11 trapped ion qubits \cite{wright2019benchmarking}, available via Amazon Braket on Amazon Web Services (AWS). At the time of simulation, the machine was reported\footnote{This information can be found via the AWS status page} to have a 1 qubit gate average fidelity of 0.99717, 2 qubit gate average fidelity of 0.9696, and a state preparation and measurement (SPAM) mean of 0.9961, which agrees with the data presented in Ref.~\cite{wright2019benchmarking}. We use 51200 shadows each for $n=2,3,4$ in the experiments with at most 1024 unitary circuits; see Appendix.~\ref{app:shots}.

Included against our ShadowQPT results is a maximum-likelihood estimation (MLE) for reconstructing the Choi matrix using our random Pauli measurements. We use the iterative method outlined in Ref.~\cite{Lvovsky_2004}. Due to the size of the Clifford POVM, Clifford-MLE reconstructions are done within the space of observed measurements, as is $n=4$ Pauli data. 
Notice that the MLE method is classically exponentially costly in general. 
To perform MLE, we begin with a random complex matrix, taking at least 100 iteration steps such that the iteration has converged, i.e. $\Vert \rho_{k} - \rho_{k-1}\Vert \leq 10^{-3}$ or that the $L_2$-norm of the difference between steps has reached a small value. We include purification of the resulting Choi matrix alongside our shadow reconstructions as well.

We explore the effect of using `fixed' and `non-fixed' schemes by using either a number of random Clifford unitaries or a set of random Clifford unitaries that are separable between the first qubit and the rest of the system. For a `non-fixed' unitary, a two-qubit Clifford unitary may be applied to any two wires of the circuit. For the `fixed' unitary, one Clifford is applied to qubits 1 and $N+1$ (its corresponding ancilla), and the remaining Clifford unitaries are randomly applied in the `non-fixed' scheme.  Our data has 0\%, 50\%, and 43.75\% of the Clifford unitaries  manually fixed for $n=2,3,4$ respectively. 
These various schemes alternatively target respectively a situation where one might ask about any reduced density matrix \textit{a posteriori} (i.e. the non-fixed scheme) or target density matrices on a particular fixed set of qubits (i.e the fixed scheme).

\begin{figure}
    \centering
    \includegraphics[width=\columnwidth]{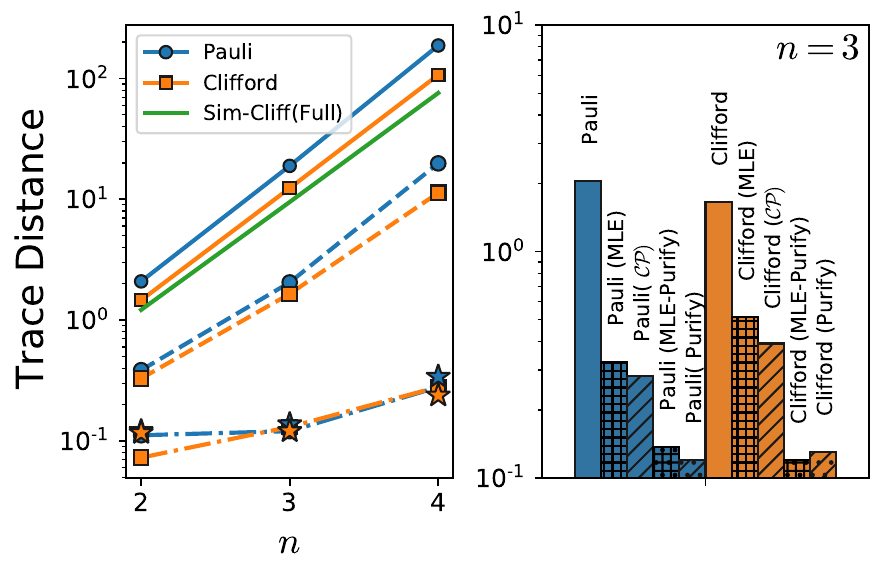}
    \caption{Left: Normalized trace distance $T(\Lambda,\Lambda^O)$ between the unitary Choi matrix $\Lambda$ and a Pauli/Clifford shadow reconstructed Choi matrix $\Lambda^O$. Solid lines represent single repetition shadows with 512 unitaries on IonQ for Pauli/Clifford measurements and in green, simulated $2n$-Clifford circuits; dotted lines includes all data/repetitions collected; dot-dashed lines are purified. The corresponding purified MLE results are shown in blue/orange stars for Pauli/Clifford measurements respectively.
    Right: Trace distance for various post-processing for $n=3$. The projection/post-processing method is shown in parenthesis above the bar. }
    \label{fig:full_trace_distance}
\end{figure}

\begin{figure*}
    \centering
    \includegraphics[width=2\columnwidth]{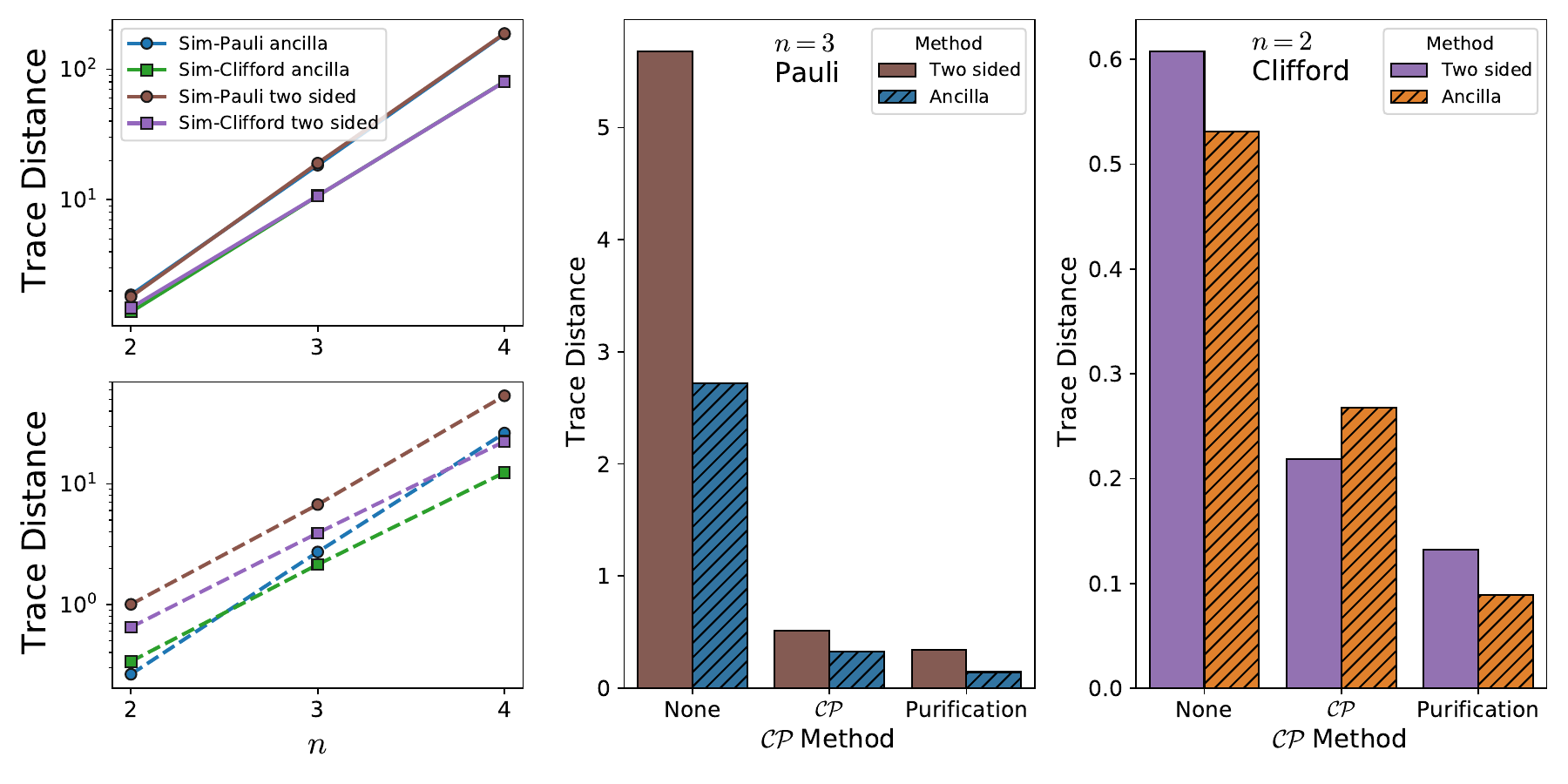}
    \caption{\textit{Left:} Simulated normalized trace distance $T(\Lambda,\Lambda^O)$ scaling for 512 unitaries with a single repetition (top) and 50 repetitions (bottom) between the unitary Choi matrix $\Lambda$ and a reconstruction via Pauli/Clifford ancilla based ShadowQPT or two sided ShadowQPT. Ancilla results are nearly identical to the two sided results with a single repetition, but with multiple repetitions ancilla simulations (squares) have lower trace distance than their two sided counterparts (circles). \textit{Middle:} Normalized trace distance between (IonQ measured) Pauli two-sided ShadowQPT and ancilla ShadowQPT with 512 unitaries and 50 repetitions for a 3 qubit GHZ process. \textit{Right:}  Normalized trace distance between (IonQ measured) Clifford two-sided ShadowQPT and ancilla ShadowQPT with 512 unitaries and 50 repetitions for a 2 qubit GHZ process. The ancilla Cliffords are decomposed into $k=2$ as in Fig.~\ref{fig:full_trace_distance}. }
    \label{fig:two_sided_plots}
\end{figure*}
\subsection{Full Process}

Using classical shadows we directly reconstruct the Choi matrix and compare it using the normalized trace distance in Fig.~\ref{fig:full_trace_distance}. 

\begin{equation}
    T(\Lambda_1, \Lambda_2) = \frac{1}{2^{n+1}} \textup{tr}\left[\sqrt{(\Lambda_1 - \Lambda_2 )^\dagger(\Lambda_1 - \Lambda_2 )}\right]
\end{equation}

First, to illustrate the scaling of ShadowQPT, a comparison of the Pauli and Clifford shadow measurements using a uniform number of shadows is shown in solid lines in Fig.~\ref{fig:full_trace_distance}. While each method is inherently exponential in scaling across process size, $k=2$ Clifford measurements performs better than Pauli measurements, where a full $k=2n$ Clifford measurement scheme should have performed the best. 

In dashed lines, we then show the results when using all of our available measurements (i.e. multiple repetitions of the same unitary). A visualization of each of these Choi matrices is shown in Fig.~\ref{fig:process_visualization} (see Appendix Fig.~\ref{fig:dual_city_plots} for a 3D version of the plot). We see that even in the regime where all possible Pauli strings are measured ($n=2,3$), the Clifford circuits have a smaller trace distance. 

If we consider that our underlying process should be unitary, we can additionally purify our noisy Choi matrix (dash-dotted lines). Purification appears to have an enormous effect, improving the trace distance by up to a factor of 34x. 
For all $n$, purified Clifford shadows have a smaller trace distance compared to the Pauli shadows, but is followed closely by the purified Clifford MLE reconstruction, which has a narrowly smaller distance for $n=3,4$ (see a further comparison in Fig.~\ref{fig:full_trace_distance_bar}).

\begin{figure}
    \centering
\subfloat[\label{fig:reduced_trace_distance_1qubit}]{\includegraphics[width=\columnwidth]{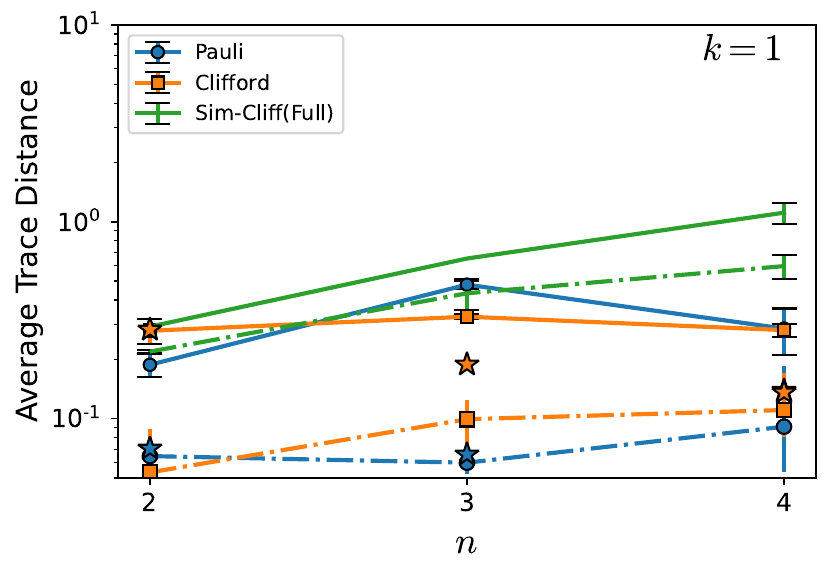} }\\
    \subfloat[\label{fig:reduced_trace_distance_2qubits}]{\includegraphics[width=\columnwidth]{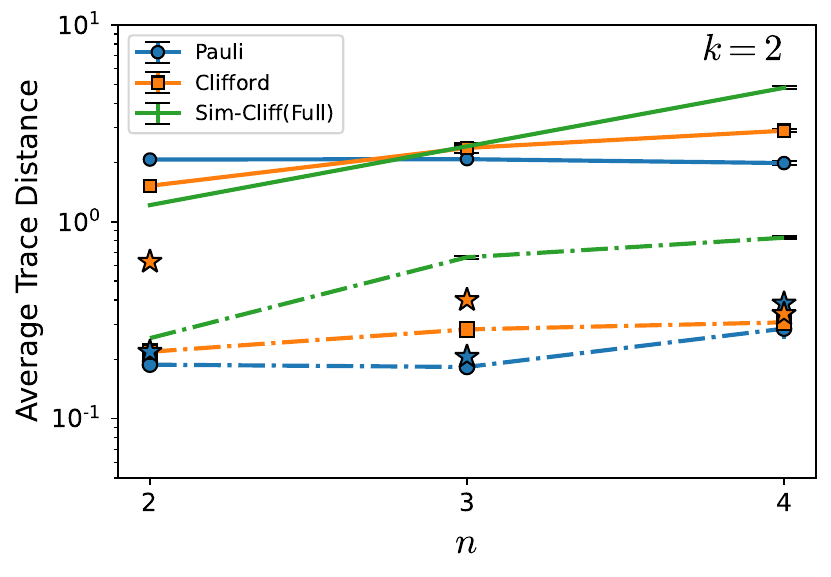}}
    \caption{Normalized trace distance $T(\Lambda_B,\Lambda^O_B)$ of all subsystems of $k=1$ qubit (a) and $k=2$ qubits (b) between the unitary Choi matrix $\Lambda_B$ and a Pauli/Clifford shadow reconstructed Choi matrix $\Lambda^O_B$. 
   Solid lines represent single repetition shadows with 512 unitaries on IonQ for Pauli/Clifford measurements and in green, simulated $2n$-Clifford circuits; dot-dashed lines are $\mathcal{CP}$-projected after the partial trace. The corresponding MLE results are shown in blue/orange stars for Pauli/Clifford measurements respectively. Pauli MLE is done directly in the reduced problem space of $k$ qubit reconstruction. }
    \label{fig:reduced_trace_distance}
\end{figure}

Here we have selected only the purification post-processing step, but there are additional projections, namely $\mathcal{CP}$-projection of Ref.~\cite{Smolin2012}, that one may do. Shown in the right of Fig.~\ref{fig:full_trace_distance} is an example of the effects on the trace distance each post-processing method has for $n=3$. We see a dramatic improvement in the trace distance after purification or projection, with Pauli shadows marginally performing better than the Clifford shadows for this system size. 

As an additional comparison we compare ShadowQPT with both an ancilla based scheme and a two sided scheme, shown in Fig.~\ref{fig:two_sided_plots}. On the left, we show simulations of both Pauli and $k=n$ Clifford measurements for a single and multiple repetitions. The simulations use 512 unitaries and 1(50) repetitions for single(multiple) repetitions. In the case of single repetitions, we find nearly identical exponential scaling but with Clifford measurements with a smaller trace distance. With the addition of multiple repetitions however, the ancilla schemes obtain a smaller trace distance than their two sided counterparts. As each repetition effectively changes both the left and right unitary for the ancilla scheme but only changes the right unitary for the two sided version, this result is unsurprising. We provide more details on the equivalence of the schemes in Appendix \ref{app:equivalence}.

Then on the right of Fig.~\ref{fig:two_sided_plots}, we compare two examples of Pauli and Clifford measurements and the effect of  post processing methods. For the two sided scheme, we use 512 Clifford unitaries at 50 repetitions each run on the IonQ device, and for the ancilla scheme we randomly resample our IonQ data to match. Note that our ancilla measurements are with randomized $k=2$ Clifford unitaries. Pauli measurements are done with 25600 random measured Pauli strings. We find that for both Pauli and Clifford measurements the ancilla based scheme has a lower normalized trace distance without post-processing and under purification, and a larger trace distance for Clifford $\mathcal{CP}$-projection.

\subsection{Non-unitary reduced process}

We further study two examples of a non-unitary reduced process on a subset of the qubits. In Figs.~\ref{fig:reduced_trace_distance_1qubit} and \ref{fig:reduced_trace_distance_2qubits} are shown the average over all 1 and 2 qubit processes. Analogous to the previous section we show a reduced dataset in solid lines, and the fully post-processed data, here $\mathcal{CP}$-projection, in dot-dashed lines. In general, we expect the full $2n$-Clifford measurements, simulation shown in green, to grow exponentially with system size. As we have fixed the size of the process, we see Pauli measurements do not scale with respect to system size, and there is a weak scaling for the 2-qubit Clifford measurement scheme.

For all system sizes $n$, a classical shadow routine produces comparable or better  trace distance compared with MLE, with an application of $\mathcal{CP}$-projection after evolution. In addition, despite having limited separability compared to the Pauli matrices, the performance of the Clifford shadows are competitive with the Pauli shadows. 

\begin{figure}[h!]
    \centering
    \subfloat[]{\includegraphics[width=0.8\columnwidth]{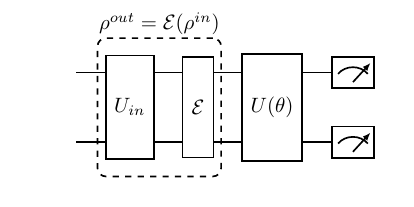}   
    \label{fig:overlap_circuit}
    }\\
    \subfloat[]{\includegraphics[width=0.95\columnwidth]{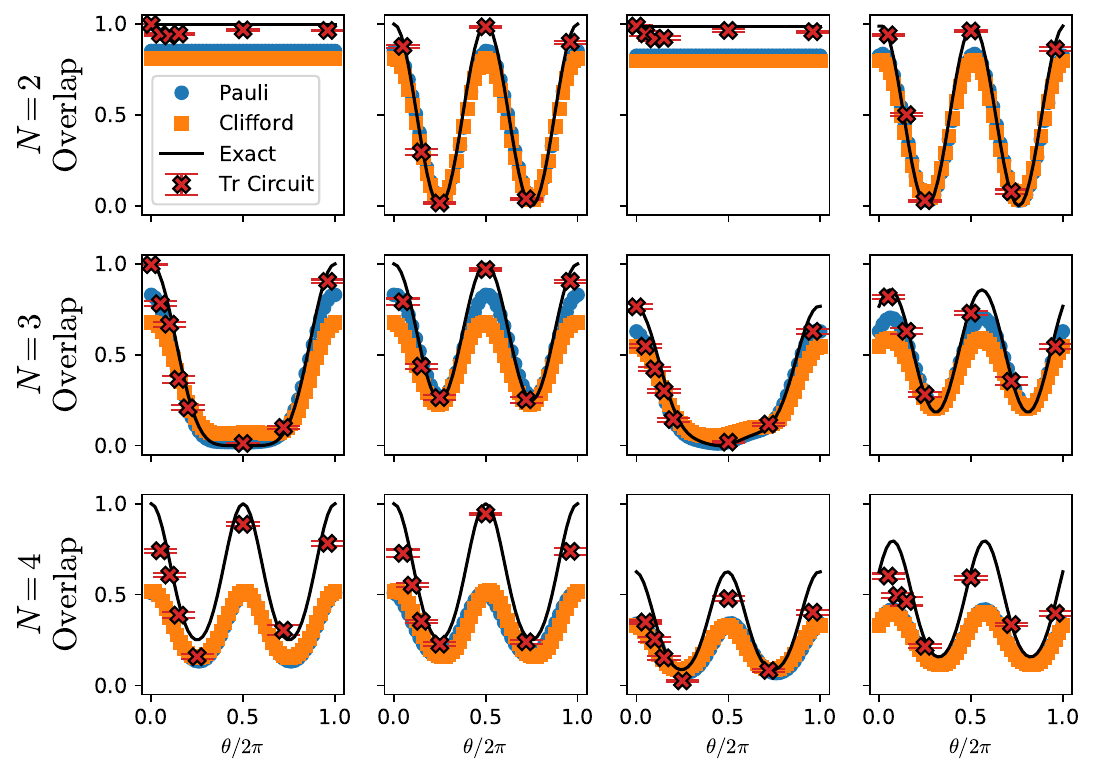}\label{fig:overlaps_all} } \\
    \subfloat[]{\includegraphics[width=0.95\columnwidth]{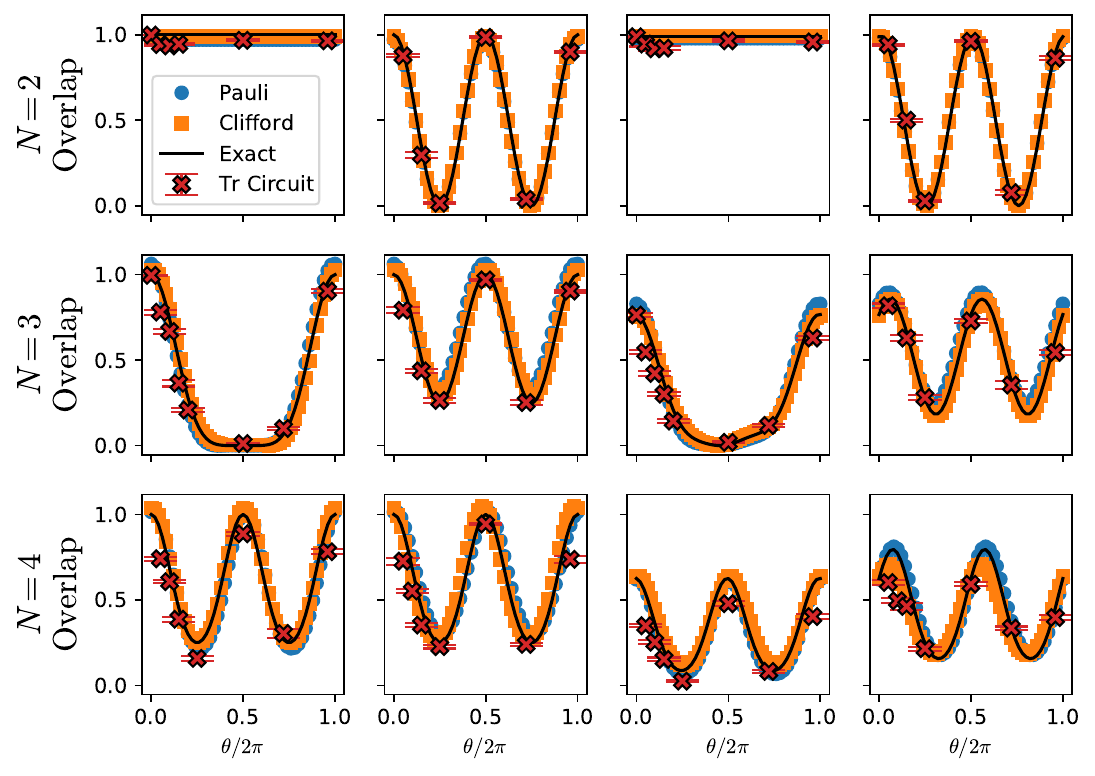}\label{fig:overlaps_all_scale} }
   
    \caption{ (a) Trace circuit measurement of the overlap  $\textup{tr}[\Lambda'^O (\rho^{in\,T}_i \otimes \sigma_i)]$ where $\sigma = U(\theta)^\dagger|0\rangle\langle 0|U(\theta)$. We measure this circuit on IonQ to compare with a prediction from ShadowQPT.
    (b) Comparison of the overlap $w_i^O = \textup{tr}[(\rho^{in\, T}_i \otimes \sigma_i) \Lambda^O]$ for Pauli and Clifford data on systems of size $n=2,3,4$ for pairs of density matrices $[\rho^{in},\sigma]$ formed from the state $[O_\rho^{in} \ket{0},O_\sigma\ket{0}]$ for operators $[I,R_y(\theta)],[I,R_x(\theta)],$ $[\otimes_j R_x(\phi_j),R_y(\theta)],[\otimes_j R_x(\phi_j),R_y(\theta)]$ (each column respectively) over 50 different angles. For comparison we include trace circuit measurements performed on the IonQ using 1000 repetitions per point; error bars are generated by using 10 batches of data.
    (c) Overlap predictions as in (b) but purifying $\Lambda^O$. }
\end{figure}

\subsection{Predicting Overlaps}

We also can compute overlaps using the shadow reconstructions formalism of Thm.~\ref{thm:overlap}.
For the choice of $\sigma_i$, we use circuits parameterized by $R_x(\theta)$ or $R_y(\theta)$ as input to a GHZ circuit; see Appendix for more details on the choices of states.  We compare to both the underlying unitary process and overlaps computed on the IonQ by constructing the trace circuit shown in Fig.~\ref{fig:overlap_circuit}; each circuit is measured with 1000 repetitions.  Note that this latter approach requires additional quantum measurements for each target overlap whereas the shadow tomography uses the same ShadowQPT data for the overlap of any output state.

In Fig.~\ref{fig:overlaps_all} we show the overlap prediction $w_i^O=\textup{tr}[(\rho^{in\, T}_i \otimes \sigma_i)\Lambda^{O}]$ to the shadow reconstruction data $\Lambda^O$ without post-processing. For $n=2,3$, we are able to reasonably reconstruct the correct overlap, but for $n=4$ there is not enough data to obtain the correct pure state behavior. Note that in each case the correct qualitative behavior is evident, despite predicting much smaller overlap values.

Then in Fig.~\ref{fig:overlaps_all_scale} we use the the most efficient post-processing approach to purify the Choi matrix. Purified Pauli and Clifford overlap predictions are now nearly identical with the expected noise-free values, even for the case of $n=4$.

\begin{figure}[t]
   \centering
   \subfloat[]{\includegraphics[width=\columnwidth]{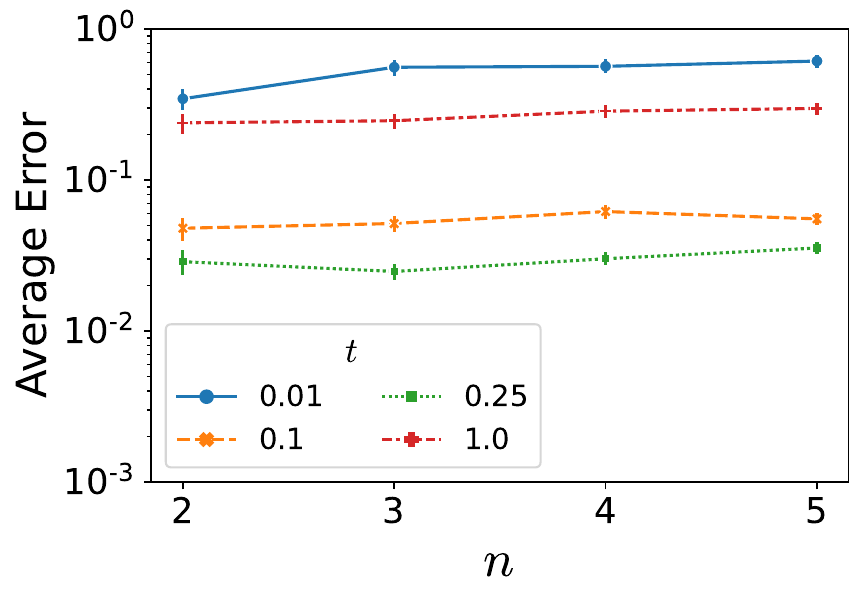}}\\
   \subfloat[]{
    \includegraphics[width=\columnwidth]{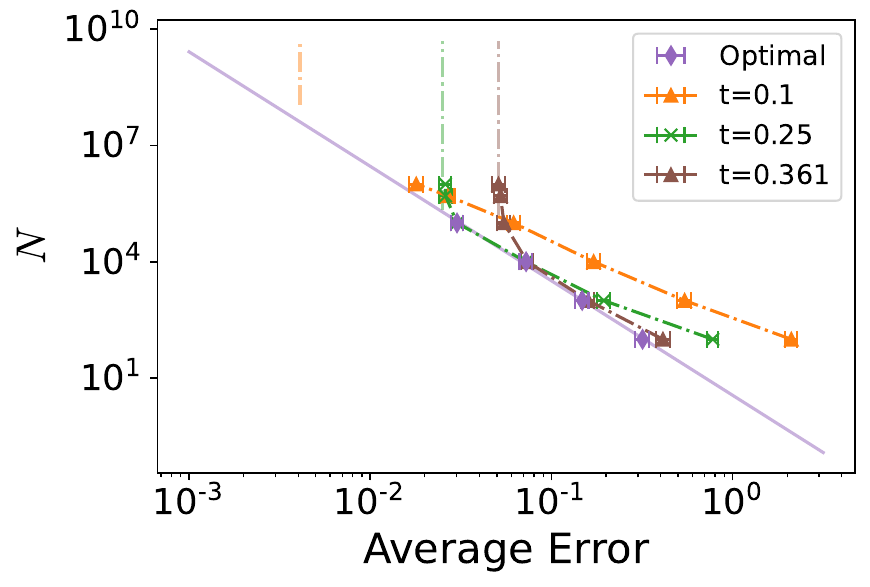} }
    \caption{ 
     (a) Hamiltonian learning simulation results for a 1D transverse field Ising model with $n$ sites and random couplings between $[-1,1]$. We average over 10 disorder realizations and use $N=100000$ random Pauli measurements with no additional post processing. The average error is given by average absolute error $\langle |\tilde{c}_i-c_i|\rangle$ to the original Hamiltonian coupling $c_i$.
    (b) Dependence on number of Pauli measurements $N$ vs. average error of Hamiltonian learning both for fixed $t$ and the optimal $t$ (shown in purple).  Notice at fixed $t$, as the error decreases it approaches the minimum error at that $t$, shown by the vertical dot-dashed line, coming from the systematic error caused by the linear approximation to the time evolution operator. } 
    
    \label{fig:ham_learning_sim}
\end{figure}

\section{Numerical Simulation of ShadowQPT Hamiltonian Learning}
Here we present simulation results with Hamiltonian learning. We study a 1D transverse field Ising model $H = \sum_i J_{i} X_i X_{i+1} + h_i Z_i $ with uniform random couplings $J_i, h_i \in [-1,1)$ over 10 disorder realizations. When applying ShadowQPT, we apply the unitary $e^{-itH}$ and sample $N$ random Pauli measurements with no additional post-processing.
For the $J_i$ terms we use $p_i = Z_i$, $q_i=[h_i,p_i]=-Y_iX_{i+1}$ and for the $h_i$ terms $p_i=X_i$, $q_i=[h_i,p_i]=Y_i$.   

Using the methods described in Section~\ref{sec:HamLearnAlg}, we plot the behavior of the average error in Fig.~\ref{fig:ham_learning_sim}.  In Figure~\ref{fig:ham_learning_sim}a, we show  the average absolute error $\langle |\tilde{c}_i - c_i| \rangle$ between the ShadowQPT learned couplings $\tilde{c_i}$ and the original couplings $c_i$ ($c_i  \equiv \{J,h\}$) with a simulation of $N=100000$.  
Given a fixed $t$, the average error scales nearly independently of system size $n$ (and should  be bounded by $O(\log(n))$ scaling at large $n$).
For small $t$, while the first order expansion of the time evolution operator is a good approximation,
one needs many measurements to 
determine the couplings beyond statistical noise as the channel is very close to the identity.   For large $t$, the higher order terms generated by the exponential become important and the learned couplings $\tilde{c}_i \approx c_i^{renorm}(t)=i\mathrm{tr}(\rho_{\Lambda_t}(p_i)q_i)/t \not\approx c_i$. At intermediate $t$, we find a favorable regime where there are enough measurements to learn the random couplings to an average error of $10^{-1}$.  

In Fig.~\ref{fig:ham_learning_sim}b (see Appendix for further details) we more carefully analyze the errors in this process, particularly as we change the number of measurements and time.   There are two sources of error in Hamiltonian learning using ShadowQPT.  A systematic error occurs due to the linear approximation of the time evolution operator and is given by $\epsilon_s(t) \equiv \langle | c_i - c_i^{renorm}(t)|\rangle$ and scales as $O(t^2)$.  Notice that this error does not depend on the number of measurements.  There is also a statistical error coming from the stochastic nature of the ShadowQPT process which, scales as  $\epsilon \propto 1/ \sqrt{N}$ (see $t=0.1$ in Fig.~\ref{fig:ham_learning_sim}b) and empirically $\epsilon \propto 1/t$ (see Appendix Fig.~\ref{fig:app_ham_learning_slice}).  For a fixed t, we can therefore make the error better by increasing the number of measurements until we reach an error of $\epsilon_s(t)$.  The optimal error, found when the systematic error is approximately the same as the statistical error, is shown in purple diamonds in Fig.~\ref{fig:ham_learning_sim}b and scales as $O(1/\epsilon^{3})$. Despite the increase in scaling, the number of measurements required are overall fewer than from the $O(1/\epsilon^2)$ results of using a fixed time. As we also note earlier in Sec.~\ref{sec:HamLearnAlg}, the systematic error can be further reduced by higher Taylor expansion and classical processing to achieve better overall scaling. Our study provides a prescription for choosing the number of measurements needed for a given value of $\epsilon$ with weak dependence of $n$.

\section{Conclusion}

In this work, we have developed classical shadows algorithms for quantum process tomography. We discuss the power of classical ShadowQPT in Theorem.~\ref{thm:k-qubit},~\ref{thm:overlap}.  Interestingly, ShadowQPT can be applied to low-weight inputs and outputs at a cost in quantum measurements which scales only logarithmically with $N$. 
To realize the ShadowQPT on near term quantum devices, we study both the effects of random Pauli measurements and Clifford measurements. In addition, we explore different post-processing techniques.
We further benchmark our methods and achieve good performance for unitary process, non-unitary reduced process, and overlap estimation on IonQ quantum hardware up to $n=4$ qubits. We find that our post-processing techniques, particularly purification, result in the ability to simultaneously compute very accurate overlaps for any state.   
Our method not only provides a theoretical foundation,
but also effectively applies to NISQ quantum device. %
Additionally, we show the equivalence between the ancilla based ShadowQPT and a two-sided scheme, and find both via simulations and measurements that the ancilla scheme under multiple repetitions achieves better performance. 

We additionally develop Hamiltonian learning using ShadowQPT, and discuss its scaling in Corollary~\ref{thm:ham_learn}. Then using numerical simulations we show that there is a logarithmic dependence on system size, as expected, and using intermediate time propagators can be advantageous for minimizing the total number of samples needed for a given error threshold.   

For future exploration, one can consider to integrate our methods with recently developed Hamiltonian driven shadow methods~\cite{hu2021classical,hu2021hamiltoniandriven} or full Clifford circuit decomposition \cite{huggins2021unbiasing}, which could allow for more flexible choices for unitaries as well as applications to different experimental platforms. 
ShadowQPT can further study the dynamics of observables and correlation functions, which is explored in a concurrent work of Ref.~\cite{kunjummen2021shadow} that also studied the use of classical shadows in QPT. 
We anticipate that ShadowQPT will be an important tool for future exploration and validation of quantum devices.

\begin{acknowledgements}
We acknowledge support from  the NSF Quantum Leap
Challenge Institute for Hybrid Quantum Architectures and
Networks (NSF Award No. 2016136). This work is supported by the National Science Foundation under Cooperative Agreement PHY-2019786 (The NSF AI Institute for Artificial Intelligence and Fundamental Interactions, http://iaifi.org/). This material is based upon work supported by the U.S. Department of Energy, Office of Science, National Quantum Information Science Research Centers, Co-design Center for Quantum Advantage (C2QA) under Contract No. DE-SC0012704.
A portion of this work was completed during the 2021 QHack hackathon. We gratefully acknowledge funding for quantum computing machine-time from AWS and Xanadu. R.L. and D.L. thank our two hackathon team members E. Chertkov and A. Khan, and acknowledge helpful discussions from Hongye Hu, Nicholas Rubin, Juan Carrasquilla and Yizhuang You, as well as insightful comments on the shadow schemes from Daniel K. Mark, Minh Tran, Tianci Zhou, and Soonwon Choi. The Flatiron Institute is a division of the Simons Foundation. 

\end{acknowledgements}

\bibliography{reference}

\begin{thebibliography}{54}%
\makeatletter
\providecommand \@ifxundefined [1]{%
 \@ifx{#1\undefined}
}%
\providecommand \@ifnum [1]{%
 \ifnum #1\expandafter \@firstoftwo
 \else \expandafter \@secondoftwo
 \fi
}%
\providecommand \@ifx [1]{%
 \ifx #1\expandafter \@firstoftwo
 \else \expandafter \@secondoftwo
 \fi
}%
\providecommand \natexlab [1]{#1}%
\providecommand \enquote  [1]{``#1''}%
\providecommand \bibnamefont  [1]{#1}%
\providecommand \bibfnamefont [1]{#1}%
\providecommand \citenamefont [1]{#1}%
\providecommand \href@noop [0]{\@secondoftwo}%
\providecommand \href [0]{\begingroup \@sanitize@url \@href}%
\providecommand \@href[1]{\@@startlink{#1}\@@href}%
\providecommand \@@href[1]{\endgroup#1\@@endlink}%
\providecommand \@sanitize@url [0]{\catcode `\\12\catcode `\$12\catcode
  `\&12\catcode `\#12\catcode `\^12\catcode `\_12\catcode `\%12\relax}%
\providecommand \@@startlink[1]{}%
\providecommand \@@endlink[0]{}%
\providecommand \url  [0]{\begingroup\@sanitize@url \@url }%
\providecommand \@url [1]{\endgroup\@href {#1}{\urlprefix }}%
\providecommand \urlprefix  [0]{URL }%
\providecommand \Eprint [0]{\href }%
\providecommand \doibase [0]{https://doi.org/}%
\providecommand \selectlanguage [0]{\@gobble}%
\providecommand \bibinfo  [0]{\@secondoftwo}%
\providecommand \bibfield  [0]{\@secondoftwo}%
\providecommand \translation [1]{[#1]}%
\providecommand \BibitemOpen [0]{}%
\providecommand \bibitemStop [0]{}%
\providecommand \bibitemNoStop [0]{.\EOS\space}%
\providecommand \EOS [0]{\spacefactor3000\relax}%
\providecommand \BibitemShut  [1]{\csname bibitem#1\endcsname}%
\let\auto@bib@innerbib\@empty
\bibitem [{\citenamefont {Bouwmeester}\ \emph {et~al.}(1997)\citenamefont
  {Bouwmeester}, \citenamefont {Pan}, \citenamefont {Mattle}, \citenamefont
  {Eibl}, \citenamefont {Weinfurter},\ and\ \citenamefont
  {Zeilinger}}]{Bouwmeester1997}%
  \BibitemOpen
  \bibfield  {author} {\bibinfo {author} {\bibfnamefont {D.}~\bibnamefont
  {Bouwmeester}}, \bibinfo {author} {\bibfnamefont {J.-W.}\ \bibnamefont
  {Pan}}, \bibinfo {author} {\bibfnamefont {K.}~\bibnamefont {Mattle}},
  \bibinfo {author} {\bibfnamefont {M.}~\bibnamefont {Eibl}}, \bibinfo {author}
  {\bibfnamefont {H.}~\bibnamefont {Weinfurter}},\ and\ \bibinfo {author}
  {\bibfnamefont {A.}~\bibnamefont {Zeilinger}},\ }\bibfield  {title} {\bibinfo
  {title} {Experimental quantum teleportation},\ }\href
  {https://doi.org/10.1038/37539} {\bibfield  {journal} {\bibinfo  {journal}
  {Nature}\ }\textbf {\bibinfo {volume} {390}},\ \bibinfo {pages} {575}
  (\bibinfo {year} {1997})}\BibitemShut {NoStop}%
\bibitem [{\citenamefont {Giovannetti}\ \emph {et~al.}(2011)\citenamefont
  {Giovannetti}, \citenamefont {Lloyd},\ and\ \citenamefont
  {Maccone}}]{Giovannetti_2011}%
  \BibitemOpen
  \bibfield  {author} {\bibinfo {author} {\bibfnamefont {V.}~\bibnamefont
  {Giovannetti}}, \bibinfo {author} {\bibfnamefont {S.}~\bibnamefont {Lloyd}},\
  and\ \bibinfo {author} {\bibfnamefont {L.}~\bibnamefont {Maccone}},\
  }\bibfield  {title} {\bibinfo {title} {Advances in quantum metrology},\
  }\href {https://doi.org/10.1038/nphoton.2011.35} {\bibfield  {journal}
  {\bibinfo  {journal} {Nature Photonics}\ }\textbf {\bibinfo {volume} {5}},\
  \bibinfo {pages} {222–229} (\bibinfo {year} {2011})}\BibitemShut {NoStop}%
\bibitem [{\citenamefont {Arute}\ \emph {et~al.}(2019)\citenamefont {Arute},
  \citenamefont {Arya}, \citenamefont {Babbush}, \citenamefont {Bacon},
  \citenamefont {Bardin}, \citenamefont {Barends}, \citenamefont {Biswas},
  \citenamefont {Boixo}, \citenamefont {Brandao}, \citenamefont {Buell},
  \citenamefont {Burkett}, \citenamefont {Chen}, \citenamefont {Chen},
  \citenamefont {Chiaro}, \citenamefont {Collins}, \citenamefont {Courtney},
  \citenamefont {Dunsworth}, \citenamefont {Farhi}, \citenamefont {Foxen},
  \citenamefont {Fowler}, \citenamefont {Gidney}, \citenamefont {Giustina},
  \citenamefont {Graff}, \citenamefont {Guerin}, \citenamefont {Habegger},
  \citenamefont {Harrigan}, \citenamefont {Hartmann}, \citenamefont {Ho},
  \citenamefont {Hoffmann}, \citenamefont {Huang}, \citenamefont {Humble},
  \citenamefont {Isakov}, \citenamefont {Jeffrey}, \citenamefont {Jiang},
  \citenamefont {Kafri}, \citenamefont {Kechedzhi}, \citenamefont {Kelly},
  \citenamefont {Klimov}, \citenamefont {Knysh}, \citenamefont {Korotkov},
  \citenamefont {Kostritsa}, \citenamefont {Landhuis}, \citenamefont
  {Lindmark}, \citenamefont {Lucero}, \citenamefont {Lyakh}, \citenamefont
  {Mandr{\`a}}, \citenamefont {McClean}, \citenamefont {McEwen}, \citenamefont
  {Megrant}, \citenamefont {Mi}, \citenamefont {Michielsen}, \citenamefont
  {Mohseni}, \citenamefont {Mutus}, \citenamefont {Naaman}, \citenamefont
  {Neeley}, \citenamefont {Neill}, \citenamefont {Niu}, \citenamefont {Ostby},
  \citenamefont {Petukhov}, \citenamefont {Platt}, \citenamefont {Quintana},
  \citenamefont {Rieffel}, \citenamefont {Roushan}, \citenamefont {Rubin},
  \citenamefont {Sank}, \citenamefont {Satzinger}, \citenamefont {Smelyanskiy},
  \citenamefont {Sung}, \citenamefont {Trevithick}, \citenamefont
  {Vainsencher}, \citenamefont {Villalonga}, \citenamefont {White},
  \citenamefont {Yao}, \citenamefont {Yeh}, \citenamefont {Zalcman},
  \citenamefont {Neven},\ and\ \citenamefont {Martinis}}]{Arute2019}%
  \BibitemOpen
  \bibfield  {author} {\bibinfo {author} {\bibfnamefont {F.}~\bibnamefont
  {Arute}}, \bibinfo {author} {\bibfnamefont {K.}~\bibnamefont {Arya}},
  \bibinfo {author} {\bibfnamefont {R.}~\bibnamefont {Babbush}}, \bibinfo
  {author} {\bibfnamefont {D.}~\bibnamefont {Bacon}}, \bibinfo {author}
  {\bibfnamefont {J.~C.}\ \bibnamefont {Bardin}}, \bibinfo {author}
  {\bibfnamefont {R.}~\bibnamefont {Barends}}, \bibinfo {author} {\bibfnamefont
  {R.}~\bibnamefont {Biswas}}, \bibinfo {author} {\bibfnamefont
  {S.}~\bibnamefont {Boixo}}, \bibinfo {author} {\bibfnamefont {F.~G. S.~L.}\
  \bibnamefont {Brandao}}, \bibinfo {author} {\bibfnamefont {D.~A.}\
  \bibnamefont {Buell}}, \bibinfo {author} {\bibfnamefont {B.}~\bibnamefont
  {Burkett}}, \bibinfo {author} {\bibfnamefont {Y.}~\bibnamefont {Chen}},
  \bibinfo {author} {\bibfnamefont {Z.}~\bibnamefont {Chen}}, \bibinfo {author}
  {\bibfnamefont {B.}~\bibnamefont {Chiaro}}, \bibinfo {author} {\bibfnamefont
  {R.}~\bibnamefont {Collins}}, \bibinfo {author} {\bibfnamefont
  {W.}~\bibnamefont {Courtney}}, \bibinfo {author} {\bibfnamefont
  {A.}~\bibnamefont {Dunsworth}}, \bibinfo {author} {\bibfnamefont
  {E.}~\bibnamefont {Farhi}}, \bibinfo {author} {\bibfnamefont
  {B.}~\bibnamefont {Foxen}}, \bibinfo {author} {\bibfnamefont
  {A.}~\bibnamefont {Fowler}}, \bibinfo {author} {\bibfnamefont
  {C.}~\bibnamefont {Gidney}}, \bibinfo {author} {\bibfnamefont
  {M.}~\bibnamefont {Giustina}}, \bibinfo {author} {\bibfnamefont
  {R.}~\bibnamefont {Graff}}, \bibinfo {author} {\bibfnamefont
  {K.}~\bibnamefont {Guerin}}, \bibinfo {author} {\bibfnamefont
  {S.}~\bibnamefont {Habegger}}, \bibinfo {author} {\bibfnamefont {M.~P.}\
  \bibnamefont {Harrigan}}, \bibinfo {author} {\bibfnamefont {M.~J.}\
  \bibnamefont {Hartmann}}, \bibinfo {author} {\bibfnamefont {A.}~\bibnamefont
  {Ho}}, \bibinfo {author} {\bibfnamefont {M.}~\bibnamefont {Hoffmann}},
  \bibinfo {author} {\bibfnamefont {T.}~\bibnamefont {Huang}}, \bibinfo
  {author} {\bibfnamefont {T.~S.}\ \bibnamefont {Humble}}, \bibinfo {author}
  {\bibfnamefont {S.~V.}\ \bibnamefont {Isakov}}, \bibinfo {author}
  {\bibfnamefont {E.}~\bibnamefont {Jeffrey}}, \bibinfo {author} {\bibfnamefont
  {Z.}~\bibnamefont {Jiang}}, \bibinfo {author} {\bibfnamefont
  {D.}~\bibnamefont {Kafri}}, \bibinfo {author} {\bibfnamefont
  {K.}~\bibnamefont {Kechedzhi}}, \bibinfo {author} {\bibfnamefont
  {J.}~\bibnamefont {Kelly}}, \bibinfo {author} {\bibfnamefont {P.~V.}\
  \bibnamefont {Klimov}}, \bibinfo {author} {\bibfnamefont {S.}~\bibnamefont
  {Knysh}}, \bibinfo {author} {\bibfnamefont {A.}~\bibnamefont {Korotkov}},
  \bibinfo {author} {\bibfnamefont {F.}~\bibnamefont {Kostritsa}}, \bibinfo
  {author} {\bibfnamefont {D.}~\bibnamefont {Landhuis}}, \bibinfo {author}
  {\bibfnamefont {M.}~\bibnamefont {Lindmark}}, \bibinfo {author}
  {\bibfnamefont {E.}~\bibnamefont {Lucero}}, \bibinfo {author} {\bibfnamefont
  {D.}~\bibnamefont {Lyakh}}, \bibinfo {author} {\bibfnamefont
  {S.}~\bibnamefont {Mandr{\`a}}}, \bibinfo {author} {\bibfnamefont {J.~R.}\
  \bibnamefont {McClean}}, \bibinfo {author} {\bibfnamefont {M.}~\bibnamefont
  {McEwen}}, \bibinfo {author} {\bibfnamefont {A.}~\bibnamefont {Megrant}},
  \bibinfo {author} {\bibfnamefont {X.}~\bibnamefont {Mi}}, \bibinfo {author}
  {\bibfnamefont {K.}~\bibnamefont {Michielsen}}, \bibinfo {author}
  {\bibfnamefont {M.}~\bibnamefont {Mohseni}}, \bibinfo {author} {\bibfnamefont
  {J.}~\bibnamefont {Mutus}}, \bibinfo {author} {\bibfnamefont
  {O.}~\bibnamefont {Naaman}}, \bibinfo {author} {\bibfnamefont
  {M.}~\bibnamefont {Neeley}}, \bibinfo {author} {\bibfnamefont
  {C.}~\bibnamefont {Neill}}, \bibinfo {author} {\bibfnamefont {M.~Y.}\
  \bibnamefont {Niu}}, \bibinfo {author} {\bibfnamefont {E.}~\bibnamefont
  {Ostby}}, \bibinfo {author} {\bibfnamefont {A.}~\bibnamefont {Petukhov}},
  \bibinfo {author} {\bibfnamefont {J.~C.}\ \bibnamefont {Platt}}, \bibinfo
  {author} {\bibfnamefont {C.}~\bibnamefont {Quintana}}, \bibinfo {author}
  {\bibfnamefont {E.~G.}\ \bibnamefont {Rieffel}}, \bibinfo {author}
  {\bibfnamefont {P.}~\bibnamefont {Roushan}}, \bibinfo {author} {\bibfnamefont
  {N.~C.}\ \bibnamefont {Rubin}}, \bibinfo {author} {\bibfnamefont
  {D.}~\bibnamefont {Sank}}, \bibinfo {author} {\bibfnamefont {K.~J.}\
  \bibnamefont {Satzinger}}, \bibinfo {author} {\bibfnamefont {V.}~\bibnamefont
  {Smelyanskiy}}, \bibinfo {author} {\bibfnamefont {K.~J.}\ \bibnamefont
  {Sung}}, \bibinfo {author} {\bibfnamefont {M.~D.}\ \bibnamefont
  {Trevithick}}, \bibinfo {author} {\bibfnamefont {A.}~\bibnamefont
  {Vainsencher}}, \bibinfo {author} {\bibfnamefont {B.}~\bibnamefont
  {Villalonga}}, \bibinfo {author} {\bibfnamefont {T.}~\bibnamefont {White}},
  \bibinfo {author} {\bibfnamefont {Z.~J.}\ \bibnamefont {Yao}}, \bibinfo
  {author} {\bibfnamefont {P.}~\bibnamefont {Yeh}}, \bibinfo {author}
  {\bibfnamefont {A.}~\bibnamefont {Zalcman}}, \bibinfo {author} {\bibfnamefont
  {H.}~\bibnamefont {Neven}},\ and\ \bibinfo {author} {\bibfnamefont {J.~M.}\
  \bibnamefont {Martinis}},\ }\bibfield  {title} {\bibinfo {title} {Quantum
  supremacy using a programmable superconducting processor},\ }\href
  {https://doi.org/10.1038/s41586-019-1666-5} {\bibfield  {journal} {\bibinfo
  {journal} {Nature}\ }\textbf {\bibinfo {volume} {574}},\ \bibinfo {pages}
  {505} (\bibinfo {year} {2019})}\BibitemShut {NoStop}%
\bibitem [{\citenamefont {Zhong}\ \emph {et~al.}(2020)\citenamefont {Zhong},
  \citenamefont {Wang}, \citenamefont {Deng}, \citenamefont {Chen},
  \citenamefont {Peng}, \citenamefont {Luo}, \citenamefont {Qin}, \citenamefont
  {Wu}, \citenamefont {Ding}, \citenamefont {Hu}, \citenamefont {Hu},
  \citenamefont {Yang}, \citenamefont {Zhang}, \citenamefont {Li},
  \citenamefont {Li}, \citenamefont {Jiang}, \citenamefont {Gan}, \citenamefont
  {Yang}, \citenamefont {You}, \citenamefont {Wang}, \citenamefont {Li},
  \citenamefont {Liu}, \citenamefont {Lu},\ and\ \citenamefont
  {Pan}}]{Zhong1460}%
  \BibitemOpen
  \bibfield  {author} {\bibinfo {author} {\bibfnamefont {H.-S.}\ \bibnamefont
  {Zhong}}, \bibinfo {author} {\bibfnamefont {H.}~\bibnamefont {Wang}},
  \bibinfo {author} {\bibfnamefont {Y.-H.}\ \bibnamefont {Deng}}, \bibinfo
  {author} {\bibfnamefont {M.-C.}\ \bibnamefont {Chen}}, \bibinfo {author}
  {\bibfnamefont {L.-C.}\ \bibnamefont {Peng}}, \bibinfo {author}
  {\bibfnamefont {Y.-H.}\ \bibnamefont {Luo}}, \bibinfo {author} {\bibfnamefont
  {J.}~\bibnamefont {Qin}}, \bibinfo {author} {\bibfnamefont {D.}~\bibnamefont
  {Wu}}, \bibinfo {author} {\bibfnamefont {X.}~\bibnamefont {Ding}}, \bibinfo
  {author} {\bibfnamefont {Y.}~\bibnamefont {Hu}}, \bibinfo {author}
  {\bibfnamefont {P.}~\bibnamefont {Hu}}, \bibinfo {author} {\bibfnamefont
  {X.-Y.}\ \bibnamefont {Yang}}, \bibinfo {author} {\bibfnamefont {W.-J.}\
  \bibnamefont {Zhang}}, \bibinfo {author} {\bibfnamefont {H.}~\bibnamefont
  {Li}}, \bibinfo {author} {\bibfnamefont {Y.}~\bibnamefont {Li}}, \bibinfo
  {author} {\bibfnamefont {X.}~\bibnamefont {Jiang}}, \bibinfo {author}
  {\bibfnamefont {L.}~\bibnamefont {Gan}}, \bibinfo {author} {\bibfnamefont
  {G.}~\bibnamefont {Yang}}, \bibinfo {author} {\bibfnamefont {L.}~\bibnamefont
  {You}}, \bibinfo {author} {\bibfnamefont {Z.}~\bibnamefont {Wang}}, \bibinfo
  {author} {\bibfnamefont {L.}~\bibnamefont {Li}}, \bibinfo {author}
  {\bibfnamefont {N.-L.}\ \bibnamefont {Liu}}, \bibinfo {author} {\bibfnamefont
  {C.-Y.}\ \bibnamefont {Lu}},\ and\ \bibinfo {author} {\bibfnamefont {J.-W.}\
  \bibnamefont {Pan}},\ }\bibfield  {title} {\bibinfo {title} {Quantum
  computational advantage using photons},\ }\href
  {https://doi.org/10.1126/science.abe8770} {\bibfield  {journal} {\bibinfo
  {journal} {Science}\ }\textbf {\bibinfo {volume} {370}},\ \bibinfo {pages}
  {1460} (\bibinfo {year} {2020})}\BibitemShut {NoStop}%
\bibitem [{\citenamefont {Ebadi}\ \emph {et~al.}(2021)\citenamefont {Ebadi},
  \citenamefont {Wang}, \citenamefont {Levine}, \citenamefont {Keesling},
  \citenamefont {Semeghini}, \citenamefont {Omran}, \citenamefont {Bluvstein},
  \citenamefont {Samajdar}, \citenamefont {Pichler}, \citenamefont {Ho},\ and\
  \citenamefont {et~al.}}]{Ebadi_2021}%
  \BibitemOpen
  \bibfield  {author} {\bibinfo {author} {\bibfnamefont {S.}~\bibnamefont
  {Ebadi}}, \bibinfo {author} {\bibfnamefont {T.~T.}\ \bibnamefont {Wang}},
  \bibinfo {author} {\bibfnamefont {H.}~\bibnamefont {Levine}}, \bibinfo
  {author} {\bibfnamefont {A.}~\bibnamefont {Keesling}}, \bibinfo {author}
  {\bibfnamefont {G.}~\bibnamefont {Semeghini}}, \bibinfo {author}
  {\bibfnamefont {A.}~\bibnamefont {Omran}}, \bibinfo {author} {\bibfnamefont
  {D.}~\bibnamefont {Bluvstein}}, \bibinfo {author} {\bibfnamefont
  {R.}~\bibnamefont {Samajdar}}, \bibinfo {author} {\bibfnamefont
  {H.}~\bibnamefont {Pichler}}, \bibinfo {author} {\bibfnamefont {W.~W.}\
  \bibnamefont {Ho}},\ and\ \bibinfo {author} {\bibnamefont {et~al.}},\
  }\bibfield  {title} {\bibinfo {title} {Quantum phases of matter on a 256-atom
  programmable quantum simulator},\ }\href
  {https://doi.org/10.1038/s41586-021-03582-4} {\bibfield  {journal} {\bibinfo
  {journal} {Nature}\ }\textbf {\bibinfo {volume} {595}},\ \bibinfo {pages}
  {227–232} (\bibinfo {year} {2021})}\BibitemShut {NoStop}%
\bibitem [{\citenamefont {Chuang}\ and\ \citenamefont
  {Nielsen}(1997)}]{linear}%
  \BibitemOpen
  \bibfield  {author} {\bibinfo {author} {\bibfnamefont {I.~L.}\ \bibnamefont
  {Chuang}}\ and\ \bibinfo {author} {\bibfnamefont {M.~A.}\ \bibnamefont
  {Nielsen}},\ }\bibfield  {title} {\bibinfo {title} {Prescription for
  experimental determination of the dynamics of a quantum black box},\ }\href
  {https://doi.org/10.1080/09500349708231894} {\bibfield  {journal} {\bibinfo
  {journal} {Journal of Modern Optics}\ }\textbf {\bibinfo {volume} {44}},\
  \bibinfo {pages} {2455} (\bibinfo {year} {1997})}\BibitemShut {NoStop}%
\bibitem [{\citenamefont {Poyatos}\ \emph {et~al.}(1997)\citenamefont
  {Poyatos}, \citenamefont {Cirac},\ and\ \citenamefont
  {Zoller}}]{PhysRevLett.78.390}%
  \BibitemOpen
  \bibfield  {author} {\bibinfo {author} {\bibfnamefont {J.~F.}\ \bibnamefont
  {Poyatos}}, \bibinfo {author} {\bibfnamefont {J.~I.}\ \bibnamefont {Cirac}},\
  and\ \bibinfo {author} {\bibfnamefont {P.}~\bibnamefont {Zoller}},\
  }\bibfield  {title} {\bibinfo {title} {Complete characterization of a quantum
  process: The two-bit quantum gate},\ }\href
  {https://doi.org/10.1103/PhysRevLett.78.390} {\bibfield  {journal} {\bibinfo
  {journal} {Phys. Rev. Lett.}\ }\textbf {\bibinfo {volume} {78}},\ \bibinfo
  {pages} {390} (\bibinfo {year} {1997})}\BibitemShut {NoStop}%
\bibitem [{\citenamefont {James}\ \emph {et~al.}(2001)\citenamefont {James},
  \citenamefont {Kwiat}, \citenamefont {Munro},\ and\ \citenamefont
  {White}}]{PhysRevA.64.052312}%
  \BibitemOpen
  \bibfield  {author} {\bibinfo {author} {\bibfnamefont {D.~F.~V.}\
  \bibnamefont {James}}, \bibinfo {author} {\bibfnamefont {P.~G.}\ \bibnamefont
  {Kwiat}}, \bibinfo {author} {\bibfnamefont {W.~J.}\ \bibnamefont {Munro}},\
  and\ \bibinfo {author} {\bibfnamefont {A.~G.}\ \bibnamefont {White}},\
  }\bibfield  {title} {\bibinfo {title} {Measurement of qubits},\ }\href
  {https://doi.org/10.1103/PhysRevA.64.052312} {\bibfield  {journal} {\bibinfo
  {journal} {Phys. Rev. A}\ }\textbf {\bibinfo {volume} {64}},\ \bibinfo
  {pages} {052312} (\bibinfo {year} {2001})}\BibitemShut {NoStop}%
\bibitem [{\citenamefont {Bouchard}\ \emph {et~al.}(2019)\citenamefont
  {Bouchard}, \citenamefont {Hufnagel}, \citenamefont {Koutný}, \citenamefont
  {Abbas}, \citenamefont {Sit}, \citenamefont {Heshami}, \citenamefont
  {Fickler},\ and\ \citenamefont {Karimi}}]{Bouchard_2019}%
  \BibitemOpen
  \bibfield  {author} {\bibinfo {author} {\bibfnamefont {F.}~\bibnamefont
  {Bouchard}}, \bibinfo {author} {\bibfnamefont {F.}~\bibnamefont {Hufnagel}},
  \bibinfo {author} {\bibfnamefont {D.}~\bibnamefont {Koutný}}, \bibinfo
  {author} {\bibfnamefont {A.}~\bibnamefont {Abbas}}, \bibinfo {author}
  {\bibfnamefont {A.}~\bibnamefont {Sit}}, \bibinfo {author} {\bibfnamefont
  {K.}~\bibnamefont {Heshami}}, \bibinfo {author} {\bibfnamefont
  {R.}~\bibnamefont {Fickler}},\ and\ \bibinfo {author} {\bibfnamefont
  {E.}~\bibnamefont {Karimi}},\ }\bibfield  {title} {\bibinfo {title} {Quantum
  process tomography of a high-dimensional quantum communication channel},\
  }\href {https://doi.org/10.22331/q-2019-05-06-138} {\bibfield  {journal}
  {\bibinfo  {journal} {Quantum}\ }\textbf {\bibinfo {volume} {3}},\ \bibinfo
  {pages} {138} (\bibinfo {year} {2019})}\BibitemShut {NoStop}%
\bibitem [{\citenamefont {Lvovsky}(2004)}]{Lvovsky_2004}%
  \BibitemOpen
  \bibfield  {author} {\bibinfo {author} {\bibfnamefont {A.~I.}\ \bibnamefont
  {Lvovsky}},\ }\bibfield  {title} {\bibinfo {title} {Iterative
  maximum-likelihood reconstruction in quantum homodyne tomography},\ }\href
  {https://doi.org/10.1088/1464-4266/6/6/014} {\bibfield  {journal} {\bibinfo
  {journal} {Journal of Optics B: Quantum and Semiclassical Optics}\ }\textbf
  {\bibinfo {volume} {6}},\ \bibinfo {pages} {S556–S559} (\bibinfo {year}
  {2004})}\BibitemShut {NoStop}%
\bibitem [{\citenamefont
  {Blume-Kohout}(2010{\natexlab{a}})}]{PhysRevLett.105.200504}%
  \BibitemOpen
  \bibfield  {author} {\bibinfo {author} {\bibfnamefont {R.}~\bibnamefont
  {Blume-Kohout}},\ }\bibfield  {title} {\bibinfo {title} {Hedged maximum
  likelihood quantum state estimation},\ }\href
  {https://doi.org/10.1103/PhysRevLett.105.200504} {\bibfield  {journal}
  {\bibinfo  {journal} {Phys. Rev. Lett.}\ }\textbf {\bibinfo {volume} {105}},\
  \bibinfo {pages} {200504} (\bibinfo {year} {2010}{\natexlab{a}})}\BibitemShut
  {NoStop}%
\bibitem [{\citenamefont {Smolin}\ \emph
  {et~al.}(2012{\natexlab{a}})\citenamefont {Smolin}, \citenamefont
  {Gambetta},\ and\ \citenamefont {Smith}}]{PhysRevLett.108.070502}%
  \BibitemOpen
  \bibfield  {author} {\bibinfo {author} {\bibfnamefont {J.~A.}\ \bibnamefont
  {Smolin}}, \bibinfo {author} {\bibfnamefont {J.~M.}\ \bibnamefont
  {Gambetta}},\ and\ \bibinfo {author} {\bibfnamefont {G.}~\bibnamefont
  {Smith}},\ }\bibfield  {title} {\bibinfo {title} {Efficient method for
  computing the maximum-likelihood quantum state from measurements with
  additive gaussian noise},\ }\href
  {https://doi.org/10.1103/PhysRevLett.108.070502} {\bibfield  {journal}
  {\bibinfo  {journal} {Phys. Rev. Lett.}\ }\textbf {\bibinfo {volume} {108}},\
  \bibinfo {pages} {070502} (\bibinfo {year} {2012}{\natexlab{a}})}\BibitemShut
  {NoStop}%
\bibitem [{\citenamefont {Granade}\ \emph {et~al.}(2017)\citenamefont
  {Granade}, \citenamefont {Ferrie},\ and\ \citenamefont
  {Flammia}}]{Granade_2017}%
  \BibitemOpen
  \bibfield  {author} {\bibinfo {author} {\bibfnamefont {C.}~\bibnamefont
  {Granade}}, \bibinfo {author} {\bibfnamefont {C.}~\bibnamefont {Ferrie}},\
  and\ \bibinfo {author} {\bibfnamefont {S.~T.}\ \bibnamefont {Flammia}},\
  }\bibfield  {title} {\bibinfo {title} {Practical adaptive quantum
  tomography},\ }\href {https://doi.org/10.1088/1367-2630/aa8fe6} {\bibfield
  {journal} {\bibinfo  {journal} {New Journal of Physics}\ }\textbf {\bibinfo
  {volume} {19}},\ \bibinfo {pages} {113017} (\bibinfo {year}
  {2017})}\BibitemShut {NoStop}%
\bibitem [{\citenamefont {Christandl}\ \emph {et~al.}(2009)\citenamefont
  {Christandl}, \citenamefont {K\"onig},\ and\ \citenamefont
  {Renner}}]{PhysRevLett.102.020504}%
  \BibitemOpen
  \bibfield  {author} {\bibinfo {author} {\bibfnamefont {M.}~\bibnamefont
  {Christandl}}, \bibinfo {author} {\bibfnamefont {R.}~\bibnamefont
  {K\"onig}},\ and\ \bibinfo {author} {\bibfnamefont {R.}~\bibnamefont
  {Renner}},\ }\bibfield  {title} {\bibinfo {title} {Postselection technique
  for quantum channels with applications to quantum cryptography},\ }\href
  {https://doi.org/10.1103/PhysRevLett.102.020504} {\bibfield  {journal}
  {\bibinfo  {journal} {Phys. Rev. Lett.}\ }\textbf {\bibinfo {volume} {102}},\
  \bibinfo {pages} {020504} (\bibinfo {year} {2009})}\BibitemShut {NoStop}%
\bibitem [{\citenamefont
  {Blume-Kohout}(2010{\natexlab{b}})}]{Blume_Kohout_2010}%
  \BibitemOpen
  \bibfield  {author} {\bibinfo {author} {\bibfnamefont {R.}~\bibnamefont
  {Blume-Kohout}},\ }\bibfield  {title} {\bibinfo {title} {Optimal, reliable
  estimation of quantum states},\ }\href
  {https://doi.org/10.1088/1367-2630/12/4/043034} {\bibfield  {journal}
  {\bibinfo  {journal} {New Journal of Physics}\ }\textbf {\bibinfo {volume}
  {12}},\ \bibinfo {pages} {043034} (\bibinfo {year}
  {2010}{\natexlab{b}})}\BibitemShut {NoStop}%
\bibitem [{\citenamefont {Rodionov}\ \emph {et~al.}(2014)\citenamefont
  {Rodionov}, \citenamefont {Veitia}, \citenamefont {Barends}, \citenamefont
  {Kelly}, \citenamefont {Sank}, \citenamefont {Wenner}, \citenamefont
  {Martinis}, \citenamefont {Kosut},\ and\ \citenamefont
  {Korotkov}}]{Rodionov_2014}%
  \BibitemOpen
  \bibfield  {author} {\bibinfo {author} {\bibfnamefont {A.~V.}\ \bibnamefont
  {Rodionov}}, \bibinfo {author} {\bibfnamefont {A.}~\bibnamefont {Veitia}},
  \bibinfo {author} {\bibfnamefont {R.}~\bibnamefont {Barends}}, \bibinfo
  {author} {\bibfnamefont {J.}~\bibnamefont {Kelly}}, \bibinfo {author}
  {\bibfnamefont {D.}~\bibnamefont {Sank}}, \bibinfo {author} {\bibfnamefont
  {J.}~\bibnamefont {Wenner}}, \bibinfo {author} {\bibfnamefont {J.~M.}\
  \bibnamefont {Martinis}}, \bibinfo {author} {\bibfnamefont {R.~L.}\
  \bibnamefont {Kosut}},\ and\ \bibinfo {author} {\bibfnamefont {A.~N.}\
  \bibnamefont {Korotkov}},\ }\bibfield  {title} {\bibinfo {title} {Compressed
  sensing quantum process tomography for superconducting quantum gates},\
  }\href {https://doi.org/10.1103/physrevb.90.144504} {\bibfield  {journal}
  {\bibinfo  {journal} {Physical Review B}\ }\textbf {\bibinfo {volume} {90}},\
  \bibinfo {pages} {144504} (\bibinfo {year} {2014})}\BibitemShut {NoStop}%
\bibitem [{\citenamefont {Torlai}\ \emph {et~al.}(2023)\citenamefont {Torlai},
  \citenamefont {Wood}, \citenamefont {Acharya}, \citenamefont {Carleo},
  \citenamefont {Carrasquilla},\ and\ \citenamefont
  {Aolita}}]{torlai2020quantum}%
  \BibitemOpen
  \bibfield  {author} {\bibinfo {author} {\bibfnamefont {G.}~\bibnamefont
  {Torlai}}, \bibinfo {author} {\bibfnamefont {C.~J.}\ \bibnamefont {Wood}},
  \bibinfo {author} {\bibfnamefont {A.}~\bibnamefont {Acharya}}, \bibinfo
  {author} {\bibfnamefont {G.}~\bibnamefont {Carleo}}, \bibinfo {author}
  {\bibfnamefont {J.}~\bibnamefont {Carrasquilla}},\ and\ \bibinfo {author}
  {\bibfnamefont {L.}~\bibnamefont {Aolita}},\ }\bibfield  {title} {\bibinfo
  {title} {Quantum process tomography with unsupervised learning and tensor
  networks},\ }\bibfield  {journal} {\bibinfo  {journal} {Nature
  Communications}\ }\textbf {\bibinfo {volume} {14}},\ \href
  {https://doi.org/10.1038/s41467-023-38332-9} {10.1038/s41467-023-38332-9}
  (\bibinfo {year} {2023})\BibitemShut {NoStop}%
\bibitem [{\citenamefont {Knee}\ \emph
  {et~al.}(2018{\natexlab{a}})\citenamefont {Knee}, \citenamefont {Bolduc},
  \citenamefont {Leach},\ and\ \citenamefont {Gauger}}]{Knee_2018}%
  \BibitemOpen
  \bibfield  {author} {\bibinfo {author} {\bibfnamefont {G.~C.}\ \bibnamefont
  {Knee}}, \bibinfo {author} {\bibfnamefont {E.}~\bibnamefont {Bolduc}},
  \bibinfo {author} {\bibfnamefont {J.}~\bibnamefont {Leach}},\ and\ \bibinfo
  {author} {\bibfnamefont {E.~M.}\ \bibnamefont {Gauger}},\ }\bibfield  {title}
  {\bibinfo {title} {Quantum process tomography via completely positive and
  trace-preserving projection},\ }\href
  {https://doi.org/10.1103/physreva.98.062336} {\bibfield  {journal} {\bibinfo
  {journal} {Physical Review A}\ }\textbf {\bibinfo {volume} {98}},\ \bibinfo
  {pages} {062336} (\bibinfo {year} {2018}{\natexlab{a}})}\BibitemShut
  {NoStop}%
\bibitem [{\citenamefont {Surawy-Stepney}\ \emph {et~al.}(2022)\citenamefont
  {Surawy-Stepney}, \citenamefont {Kahn}, \citenamefont {Kueng},\ and\
  \citenamefont {Guta}}]{surawystepney2021projected}%
  \BibitemOpen
  \bibfield  {author} {\bibinfo {author} {\bibfnamefont {T.}~\bibnamefont
  {Surawy-Stepney}}, \bibinfo {author} {\bibfnamefont {J.}~\bibnamefont
  {Kahn}}, \bibinfo {author} {\bibfnamefont {R.}~\bibnamefont {Kueng}},\ and\
  \bibinfo {author} {\bibfnamefont {M.}~\bibnamefont {Guta}},\ }\bibfield
  {title} {\bibinfo {title} {Projected least-squares quantum process
  tomography},\ }\href {https://doi.org/10.22331/q-2022-10-20-844} {\bibfield
  {journal} {\bibinfo  {journal} {Quantum}\ }\textbf {\bibinfo {volume} {6}},\
  \bibinfo {pages} {844} (\bibinfo {year} {2022})}\BibitemShut {NoStop}%
\bibitem [{\citenamefont {Zhang}\ \emph
  {et~al.}(2021{\natexlab{a}})\citenamefont {Zhang}, \citenamefont {Huo},\ and\
  \citenamefont {Li}}]{zhang2021quantum}%
  \BibitemOpen
  \bibfield  {author} {\bibinfo {author} {\bibfnamefont {G.}~\bibnamefont
  {Zhang}}, \bibinfo {author} {\bibfnamefont {M.}~\bibnamefont {Huo}},\ and\
  \bibinfo {author} {\bibfnamefont {Y.}~\bibnamefont {Li}},\ }\href@noop {}
  {\bibinfo {title} {Quantum operation of fermionic systems and process
  tomography using majorana fermion gates}} (\bibinfo {year}
  {2021}{\natexlab{a}}),\ \Eprint {https://arxiv.org/abs/2102.00620}
  {arXiv:2102.00620 [quant-ph]} \BibitemShut {NoStop}%
\bibitem [{\citenamefont {White}\ \emph {et~al.}(2022)\citenamefont {White},
  \citenamefont {Pollock}, \citenamefont {Hollenberg}, \citenamefont {Modi},\
  and\ \citenamefont {Hill}}]{white2021nonmarkovian}%
  \BibitemOpen
  \bibfield  {author} {\bibinfo {author} {\bibfnamefont {G.~A.~L.}\
  \bibnamefont {White}}, \bibinfo {author} {\bibfnamefont {F.~A.}\ \bibnamefont
  {Pollock}}, \bibinfo {author} {\bibfnamefont {L.~C.~L.}\ \bibnamefont
  {Hollenberg}}, \bibinfo {author} {\bibfnamefont {K.}~\bibnamefont {Modi}},\
  and\ \bibinfo {author} {\bibfnamefont {C.~D.}\ \bibnamefont {Hill}},\
  }\bibfield  {title} {\bibinfo {title} {Non-markovian quantum process
  tomography},\ }\href {https://doi.org/10.1103/PRXQuantum.3.020344} {\bibfield
   {journal} {\bibinfo  {journal} {PRX Quantum}\ }\textbf {\bibinfo {volume}
  {3}},\ \bibinfo {pages} {020344} (\bibinfo {year} {2022})}\BibitemShut
  {NoStop}%
\bibitem [{\citenamefont {Govia}\ \emph {et~al.}(2020)\citenamefont {Govia},
  \citenamefont {Ribeill}, \citenamefont {Rist{\`e}}, \citenamefont {Ware},\
  and\ \citenamefont {Krovi}}]{Govia2020}%
  \BibitemOpen
  \bibfield  {author} {\bibinfo {author} {\bibfnamefont {L.~C.~G.}\
  \bibnamefont {Govia}}, \bibinfo {author} {\bibfnamefont {G.~J.}\ \bibnamefont
  {Ribeill}}, \bibinfo {author} {\bibfnamefont {D.}~\bibnamefont {Rist{\`e}}},
  \bibinfo {author} {\bibfnamefont {M.}~\bibnamefont {Ware}},\ and\ \bibinfo
  {author} {\bibfnamefont {H.}~\bibnamefont {Krovi}},\ }\bibfield  {title}
  {\bibinfo {title} {Bootstrapping quantum process tomography via a
  perturbative ansatz},\ }\href {https://doi.org/10.1038/s41467-020-14873-1}
  {\bibfield  {journal} {\bibinfo  {journal} {Nature Communications}\ }\textbf
  {\bibinfo {volume} {11}},\ \bibinfo {pages} {1084} (\bibinfo {year}
  {2020})}\BibitemShut {NoStop}%
\bibitem [{\citenamefont {Xue}\ \emph {et~al.}(2022)\citenamefont {Xue},
  \citenamefont {Liu}, \citenamefont {Wang}, \citenamefont {Zhu}, \citenamefont
  {Guo},\ and\ \citenamefont {Wu}}]{xue2021variational}%
  \BibitemOpen
  \bibfield  {author} {\bibinfo {author} {\bibfnamefont {S.}~\bibnamefont
  {Xue}}, \bibinfo {author} {\bibfnamefont {Y.}~\bibnamefont {Liu}}, \bibinfo
  {author} {\bibfnamefont {Y.}~\bibnamefont {Wang}}, \bibinfo {author}
  {\bibfnamefont {P.}~\bibnamefont {Zhu}}, \bibinfo {author} {\bibfnamefont
  {C.}~\bibnamefont {Guo}},\ and\ \bibinfo {author} {\bibfnamefont
  {J.}~\bibnamefont {Wu}},\ }\bibfield  {title} {\bibinfo {title} {Variational
  quantum process tomography of unitaries},\ }\href
  {https://doi.org/10.1103/PhysRevA.105.032427} {\bibfield  {journal} {\bibinfo
   {journal} {Phys. Rev. A}\ }\textbf {\bibinfo {volume} {105}},\ \bibinfo
  {pages} {032427} (\bibinfo {year} {2022})},\ \Eprint
  {https://arxiv.org/abs/2108.02351} {arXiv:2108.02351 [quant-ph]} \BibitemShut
  {NoStop}%
\bibitem [{\citenamefont {Onorati}\ \emph {et~al.}(2021)\citenamefont
  {Onorati}, \citenamefont {Kohler},\ and\ \citenamefont
  {Cubitt}}]{onorati2021fitting}%
  \BibitemOpen
  \bibfield  {author} {\bibinfo {author} {\bibfnamefont {E.}~\bibnamefont
  {Onorati}}, \bibinfo {author} {\bibfnamefont {T.}~\bibnamefont {Kohler}},\
  and\ \bibinfo {author} {\bibfnamefont {T.}~\bibnamefont {Cubitt}},\
  }\href@noop {} {\bibinfo {title} {Fitting quantum noise models to tomography
  data}} (\bibinfo {year} {2021}),\ \Eprint {https://arxiv.org/abs/2103.17243}
  {arXiv:2103.17243 [quant-ph]} \BibitemShut {NoStop}%
\bibitem [{\citenamefont {Jamiołkowski}(1972)}]{JAMIOLKOWSKI1972275}%
  \BibitemOpen
  \bibfield  {author} {\bibinfo {author} {\bibfnamefont {A.}~\bibnamefont
  {Jamiołkowski}},\ }\bibfield  {title} {\bibinfo {title} {Linear
  transformations which preserve trace and positive semidefiniteness of
  operators},\ }\href
  {https://doi.org/https://doi.org/10.1016/0034-4877(72)90011-0} {\bibfield
  {journal} {\bibinfo  {journal} {Reports on Mathematical Physics}\ }\textbf
  {\bibinfo {volume} {3}},\ \bibinfo {pages} {275} (\bibinfo {year}
  {1972})}\BibitemShut {NoStop}%
\bibitem [{\citenamefont {Choi}(1975)}]{CHOI1975285}%
  \BibitemOpen
  \bibfield  {author} {\bibinfo {author} {\bibfnamefont {M.-D.}\ \bibnamefont
  {Choi}},\ }\bibfield  {title} {\bibinfo {title} {Completely positive linear
  maps on complex matrices},\ }\href
  {https://doi.org/https://doi.org/10.1016/0024-3795(75)90075-0} {\bibfield
  {journal} {\bibinfo  {journal} {Linear Algebra and its Applications}\
  }\textbf {\bibinfo {volume} {10}},\ \bibinfo {pages} {285} (\bibinfo {year}
  {1975})}\BibitemShut {NoStop}%
\bibitem [{\citenamefont {Huang}\ \emph {et~al.}(2020)\citenamefont {Huang},
  \citenamefont {Kueng},\ and\ \citenamefont {Preskill}}]{Huang2020}%
  \BibitemOpen
  \bibfield  {author} {\bibinfo {author} {\bibfnamefont {H.-Y.}\ \bibnamefont
  {Huang}}, \bibinfo {author} {\bibfnamefont {R.}~\bibnamefont {Kueng}},\ and\
  \bibinfo {author} {\bibfnamefont {J.}~\bibnamefont {Preskill}},\ }\bibfield
  {title} {\bibinfo {title} {Predicting many properties of a quantum system
  from very few measurements},\ }\href
  {https://doi.org/10.1038/s41567-020-0932-7} {\bibfield  {journal} {\bibinfo
  {journal} {Nature Physics}\ }\textbf {\bibinfo {volume} {16}},\ \bibinfo
  {pages} {1050} (\bibinfo {year} {2020})}\BibitemShut {NoStop}%
\bibitem [{\citenamefont {Hadfield}(2021)}]{hadfield2021adaptive}%
  \BibitemOpen
  \bibfield  {author} {\bibinfo {author} {\bibfnamefont {C.}~\bibnamefont
  {Hadfield}},\ }\href@noop {} {\bibinfo {title} {Adaptive pauli shadows for
  energy estimation}} (\bibinfo {year} {2021}),\ \Eprint
  {https://arxiv.org/abs/2105.12207} {arXiv:2105.12207 [quant-ph]} \BibitemShut
  {NoStop}%
\bibitem [{\citenamefont {Chen}\ \emph {et~al.}(2021)\citenamefont {Chen},
  \citenamefont {Yu}, \citenamefont {Zeng},\ and\ \citenamefont
  {Flammia}}]{chen2021robust}%
  \BibitemOpen
  \bibfield  {author} {\bibinfo {author} {\bibfnamefont {S.}~\bibnamefont
  {Chen}}, \bibinfo {author} {\bibfnamefont {W.}~\bibnamefont {Yu}}, \bibinfo
  {author} {\bibfnamefont {P.}~\bibnamefont {Zeng}},\ and\ \bibinfo {author}
  {\bibfnamefont {S.~T.}\ \bibnamefont {Flammia}},\ }\bibfield  {title}
  {\bibinfo {title} {Robust shadow estimation},\ }\href
  {https://doi.org/10.1103/PRXQuantum.2.030348} {\bibfield  {journal} {\bibinfo
   {journal} {PRX Quantum}\ }\textbf {\bibinfo {volume} {2}},\ \bibinfo {pages}
  {030348} (\bibinfo {year} {2021})},\ \Eprint
  {https://arxiv.org/abs/2011.09636} {arXiv:2011.09636 [quant-ph]} \BibitemShut
  {NoStop}%
\bibitem [{\citenamefont {Koh}\ and\ \citenamefont
  {Grewal}(2022)}]{koh2020classical}%
  \BibitemOpen
  \bibfield  {author} {\bibinfo {author} {\bibfnamefont {D.~E.}\ \bibnamefont
  {Koh}}\ and\ \bibinfo {author} {\bibfnamefont {S.}~\bibnamefont {Grewal}},\
  }\bibfield  {title} {\bibinfo {title} {Classical {S}hadows {W}ith {N}oise},\
  }\href {https://doi.org/10.22331/q-2022-08-16-776} {\bibfield  {journal}
  {\bibinfo  {journal} {{Quantum}}\ }\textbf {\bibinfo {volume} {6}},\ \bibinfo
  {pages} {776} (\bibinfo {year} {2022})}\BibitemShut {NoStop}%
\bibitem [{\citenamefont {Hillmich}\ \emph {et~al.}(2021)\citenamefont
  {Hillmich}, \citenamefont {Hadfield}, \citenamefont {Raymond}, \citenamefont
  {Mezzacapo},\ and\ \citenamefont {Wille}}]{hillmich2021decision}%
  \BibitemOpen
  \bibfield  {author} {\bibinfo {author} {\bibfnamefont {S.}~\bibnamefont
  {Hillmich}}, \bibinfo {author} {\bibfnamefont {C.}~\bibnamefont {Hadfield}},
  \bibinfo {author} {\bibfnamefont {R.}~\bibnamefont {Raymond}}, \bibinfo
  {author} {\bibfnamefont {A.}~\bibnamefont {Mezzacapo}},\ and\ \bibinfo
  {author} {\bibfnamefont {R.}~\bibnamefont {Wille}},\ }\bibfield  {title}
  {\bibinfo {title} {Decision diagrams for quantum measurements with shallow
  circuits},\ }in\ \href {https://doi.org/10.1109/QCE52317.2021.00018} {\emph
  {\bibinfo {booktitle} {2021 IEEE International Conference on Quantum
  Computing and Engineering (QCE)}}}\ (\bibinfo {year} {2021})\ pp.\ \bibinfo
  {pages} {24--34},\ \Eprint {https://arxiv.org/abs/2105.06932}
  {arXiv:2105.06932 [quant-ph]} \BibitemShut {NoStop}%
\bibitem [{\citenamefont {Struchalin}\ \emph {et~al.}(2021)\citenamefont
  {Struchalin}, \citenamefont {Zagorovskii}, \citenamefont {Kovlakov},
  \citenamefont {Straupe},\ and\ \citenamefont
  {Kulik}}]{struchalin2021experimental}%
  \BibitemOpen
  \bibfield  {author} {\bibinfo {author} {\bibfnamefont {G.~I.}\ \bibnamefont
  {Struchalin}}, \bibinfo {author} {\bibfnamefont {Y.~A.}\ \bibnamefont
  {Zagorovskii}}, \bibinfo {author} {\bibfnamefont {E.~V.}\ \bibnamefont
  {Kovlakov}}, \bibinfo {author} {\bibfnamefont {S.~S.}\ \bibnamefont
  {Straupe}},\ and\ \bibinfo {author} {\bibfnamefont {S.~P.}\ \bibnamefont
  {Kulik}},\ }\bibfield  {title} {\bibinfo {title} {Experimental estimation of
  quantum state properties from classical shadows},\ }\href@noop {} {\bibfield
  {journal} {\bibinfo  {journal} {PRX Quantum}\ }\textbf {\bibinfo {volume}
  {2}},\ \bibinfo {pages} {010307} (\bibinfo {year} {2021})}\BibitemShut
  {NoStop}%
\bibitem [{\citenamefont {Zhang}\ \emph
  {et~al.}(2021{\natexlab{b}})\citenamefont {Zhang}, \citenamefont {Sun},
  \citenamefont {Fang}, \citenamefont {Zhang}, \citenamefont {Yuan},\ and\
  \citenamefont {Lu}}]{zhang2021experimental}%
  \BibitemOpen
  \bibfield  {author} {\bibinfo {author} {\bibfnamefont {T.}~\bibnamefont
  {Zhang}}, \bibinfo {author} {\bibfnamefont {J.}~\bibnamefont {Sun}}, \bibinfo
  {author} {\bibfnamefont {X.-X.}\ \bibnamefont {Fang}}, \bibinfo {author}
  {\bibfnamefont {X.-M.}\ \bibnamefont {Zhang}}, \bibinfo {author}
  {\bibfnamefont {X.}~\bibnamefont {Yuan}},\ and\ \bibinfo {author}
  {\bibfnamefont {H.}~\bibnamefont {Lu}},\ }\bibfield  {title} {\bibinfo
  {title} {Experimental quantum state measurement with classical shadows},\
  }\href {https://doi.org/10.1103/PhysRevLett.127.200501} {\bibfield  {journal}
  {\bibinfo  {journal} {Phys. Rev. Lett.}\ }\textbf {\bibinfo {volume} {127}},\
  \bibinfo {pages} {200501} (\bibinfo {year} {2021}{\natexlab{b}})}\BibitemShut
  {NoStop}%
\bibitem [{\citenamefont {Zhao}\ \emph {et~al.}(2021)\citenamefont {Zhao},
  \citenamefont {Rubin},\ and\ \citenamefont {Miyake}}]{zhao2020fermionic}%
  \BibitemOpen
  \bibfield  {author} {\bibinfo {author} {\bibfnamefont {A.}~\bibnamefont
  {Zhao}}, \bibinfo {author} {\bibfnamefont {N.~C.}\ \bibnamefont {Rubin}},\
  and\ \bibinfo {author} {\bibfnamefont {A.}~\bibnamefont {Miyake}},\
  }\bibfield  {title} {\bibinfo {title} {Fermionic partial tomography via
  classical shadows},\ }\href {https://doi.org/10.1103/PhysRevLett.127.110504}
  {\bibfield  {journal} {\bibinfo  {journal} {Phys. Rev. Lett.}\ }\textbf
  {\bibinfo {volume} {127}},\ \bibinfo {pages} {110504} (\bibinfo {year}
  {2021})}\BibitemShut {NoStop}%
\bibitem [{\citenamefont {Hu}\ and\ \citenamefont
  {You}(2022)}]{hu2021hamiltoniandriven}%
  \BibitemOpen
  \bibfield  {author} {\bibinfo {author} {\bibfnamefont {H.-Y.}\ \bibnamefont
  {Hu}}\ and\ \bibinfo {author} {\bibfnamefont {Y.-Z.}\ \bibnamefont {You}},\
  }\bibfield  {title} {\bibinfo {title} {Hamiltonian-driven shadow tomography
  of quantum states},\ }\href
  {https://doi.org/10.1103/PhysRevResearch.4.013054} {\bibfield  {journal}
  {\bibinfo  {journal} {Phys. Rev. Res.}\ }\textbf {\bibinfo {volume} {4}},\
  \bibinfo {pages} {013054} (\bibinfo {year} {2022})},\ \Eprint
  {https://arxiv.org/abs/2102.10132} {arXiv:2102.10132 [quant-ph]} \BibitemShut
  {NoStop}%
\bibitem [{\citenamefont {Acharya}\ \emph {et~al.}(2021)\citenamefont
  {Acharya}, \citenamefont {Saha},\ and\ \citenamefont
  {Sengupta}}]{acharya2021informationally}%
  \BibitemOpen
  \bibfield  {author} {\bibinfo {author} {\bibfnamefont {A.}~\bibnamefont
  {Acharya}}, \bibinfo {author} {\bibfnamefont {S.}~\bibnamefont {Saha}},\ and\
  \bibinfo {author} {\bibfnamefont {A.~M.}\ \bibnamefont {Sengupta}},\
  }\bibfield  {title} {\bibinfo {title} {Shadow tomography based on
  informationally complete positive operator-valued measure},\ }\href
  {https://doi.org/10.1103/PhysRevA.104.052418} {\bibfield  {journal} {\bibinfo
   {journal} {Phys. Rev. A}\ }\textbf {\bibinfo {volume} {104}},\ \bibinfo
  {pages} {052418} (\bibinfo {year} {2021})},\ \Eprint
  {https://arxiv.org/abs/2105.05992} {arXiv:2105.05992 [quant-ph]} \BibitemShut
  {NoStop}%
\bibitem [{\citenamefont {Rouzé}\ and\ \citenamefont
  {França}(2021)}]{rouze2021learning}%
  \BibitemOpen
  \bibfield  {author} {\bibinfo {author} {\bibfnamefont {C.}~\bibnamefont
  {Rouzé}}\ and\ \bibinfo {author} {\bibfnamefont {D.~S.}\ \bibnamefont
  {França}},\ }\href@noop {} {\bibinfo {title} {Learning quantum many-body
  systems from a few copies}} (\bibinfo {year} {2021}),\ \Eprint
  {https://arxiv.org/abs/2107.03333} {arXiv:2107.03333 [quant-ph]} \BibitemShut
  {NoStop}%
\bibitem [{\citenamefont {Hadfield}\ \emph {et~al.}(2022)\citenamefont
  {Hadfield}, \citenamefont {Bravyi}, \citenamefont {Raymond},\ and\
  \citenamefont {Mezzacapo}}]{hadfield2020measurements}%
  \BibitemOpen
  \bibfield  {author} {\bibinfo {author} {\bibfnamefont {C.}~\bibnamefont
  {Hadfield}}, \bibinfo {author} {\bibfnamefont {S.}~\bibnamefont {Bravyi}},
  \bibinfo {author} {\bibfnamefont {R.}~\bibnamefont {Raymond}},\ and\ \bibinfo
  {author} {\bibfnamefont {A.}~\bibnamefont {Mezzacapo}},\ }\bibfield  {title}
  {\bibinfo {title} {Measurements of quantum hamiltonians with locally-biased
  classical shadows},\ }\href {https://doi.org/10.1007/s00220-022-04343-8}
  {\bibfield  {journal} {\bibinfo  {journal} {Communications in Mathematical
  Physics}\ }\textbf {\bibinfo {volume} {391}},\ \bibinfo {pages} {951}
  (\bibinfo {year} {2022})},\ \Eprint {https://arxiv.org/abs/2006.15788}
  {arXiv:2006.15788 [quant-ph]} \BibitemShut {NoStop}%
\bibitem [{\citenamefont {Huggins}\ \emph {et~al.}(2022)\citenamefont
  {Huggins}, \citenamefont {O'Gorman}, \citenamefont {Rubin}, \citenamefont
  {Reichman}, \citenamefont {Babbush},\ and\ \citenamefont
  {Lee}}]{huggins2021unbiasing}%
  \BibitemOpen
  \bibfield  {author} {\bibinfo {author} {\bibfnamefont {W.~J.}\ \bibnamefont
  {Huggins}}, \bibinfo {author} {\bibfnamefont {B.~A.}\ \bibnamefont
  {O'Gorman}}, \bibinfo {author} {\bibfnamefont {N.~C.}\ \bibnamefont {Rubin}},
  \bibinfo {author} {\bibfnamefont {D.~R.}\ \bibnamefont {Reichman}}, \bibinfo
  {author} {\bibfnamefont {R.}~\bibnamefont {Babbush}},\ and\ \bibinfo {author}
  {\bibfnamefont {J.}~\bibnamefont {Lee}},\ }\bibfield  {title} {\bibinfo
  {title} {Unbiasing fermionic quantum monte carlo with a quantum computer},\
  }\href {https://doi.org/10.1038/s41586-021-04351-z} {\bibfield  {journal}
  {\bibinfo  {journal} {Nature}\ }\textbf {\bibinfo {volume} {603}},\ \bibinfo
  {pages} {416} (\bibinfo {year} {2022})},\ \Eprint
  {https://arxiv.org/abs/2106.16235} {arXiv:2106.16235 [quant-ph]} \BibitemShut
  {NoStop}%
\bibitem [{\citenamefont {Hu}\ \emph {et~al.}(2023)\citenamefont {Hu},
  \citenamefont {Choi},\ and\ \citenamefont {You}}]{hu2021classical}%
  \BibitemOpen
  \bibfield  {author} {\bibinfo {author} {\bibfnamefont {H.-Y.}\ \bibnamefont
  {Hu}}, \bibinfo {author} {\bibfnamefont {S.}~\bibnamefont {Choi}},\ and\
  \bibinfo {author} {\bibfnamefont {Y.-Z.}\ \bibnamefont {You}},\ }\bibfield
  {title} {\bibinfo {title} {Classical shadow tomography with locally scrambled
  quantum dynamics},\ }\href {https://doi.org/10.1103/PhysRevResearch.5.023027}
  {\bibfield  {journal} {\bibinfo  {journal} {Phys. Rev. Res.}\ }\textbf
  {\bibinfo {volume} {5}},\ \bibinfo {pages} {023027} (\bibinfo {year}
  {2023})},\ \Eprint {https://arxiv.org/abs/2107.04817} {arXiv:2107.04817
  [quant-ph]} \BibitemShut {NoStop}%
\bibitem [{\citenamefont {Huang}\ \emph {et~al.}(2022)\citenamefont {Huang},
  \citenamefont {Kueng}, \citenamefont {Torlai}, \citenamefont {Albert},\ and\
  \citenamefont {Preskill}}]{huang2021provably}%
  \BibitemOpen
  \bibfield  {author} {\bibinfo {author} {\bibfnamefont {H.-Y.}\ \bibnamefont
  {Huang}}, \bibinfo {author} {\bibfnamefont {R.}~\bibnamefont {Kueng}},
  \bibinfo {author} {\bibfnamefont {G.}~\bibnamefont {Torlai}}, \bibinfo
  {author} {\bibfnamefont {V.~V.}\ \bibnamefont {Albert}},\ and\ \bibinfo
  {author} {\bibfnamefont {J.}~\bibnamefont {Preskill}},\ }\bibfield  {title}
  {\bibinfo {title} {Provably efficient machine learning for quantum many-body
  problems},\ }\bibfield  {journal} {\bibinfo  {journal} {Science}\ }\textbf
  {\bibinfo {volume} {377}},\ \href {https://doi.org/10.1126/science.abk3333}
  {10.1126/science.abk3333} (\bibinfo {year} {2022})\BibitemShut {NoStop}%
\bibitem [{\citenamefont {Aaronson}\ and\ \citenamefont
  {Gottesman}(2004)}]{Aaronson_2004}%
  \BibitemOpen
  \bibfield  {author} {\bibinfo {author} {\bibfnamefont {S.}~\bibnamefont
  {Aaronson}}\ and\ \bibinfo {author} {\bibfnamefont {D.}~\bibnamefont
  {Gottesman}},\ }\bibfield  {title} {\bibinfo {title} {Improved simulation of
  stabilizer circuits},\ }\href {https://doi.org/10.1103/physreva.70.052328}
  {\bibfield  {journal} {\bibinfo  {journal} {Physical Review A}\ }\textbf
  {\bibinfo {volume} {70}},\ \bibinfo {pages} {052328} (\bibinfo {year}
  {2004})}\BibitemShut {NoStop}%
\bibitem [{\citenamefont {Li}\ \emph {et~al.}(2020)\citenamefont {Li},
  \citenamefont {Zou},\ and\ \citenamefont {Hsieh}}]{Li2020}%
  \BibitemOpen
  \bibfield  {author} {\bibinfo {author} {\bibfnamefont {Z.}~\bibnamefont
  {Li}}, \bibinfo {author} {\bibfnamefont {L.}~\bibnamefont {Zou}},\ and\
  \bibinfo {author} {\bibfnamefont {T.~H.}\ \bibnamefont {Hsieh}},\ }\bibfield
  {title} {\bibinfo {title} {Hamiltonian tomography via quantum quench},\
  }\href {https://doi.org/10.1103/physrevlett.124.160502} {\bibfield  {journal}
  {\bibinfo  {journal} {Physical Review Letters}\ }\textbf {\bibinfo {volume}
  {124}},\ \bibinfo {pages} {160502} (\bibinfo {year} {2020})}\BibitemShut
  {NoStop}%
\bibitem [{\citenamefont {Haah}\ \emph {et~al.}(2021)\citenamefont {Haah},
  \citenamefont {Kothari},\ and\ \citenamefont {Tang}}]{haah2021optimal}%
  \BibitemOpen
  \bibfield  {author} {\bibinfo {author} {\bibfnamefont {J.}~\bibnamefont
  {Haah}}, \bibinfo {author} {\bibfnamefont {R.}~\bibnamefont {Kothari}},\ and\
  \bibinfo {author} {\bibfnamefont {E.}~\bibnamefont {Tang}},\ }\bibfield
  {title} {\bibinfo {title} {Optimal learning of quantum hamiltonians from
  high-temperature gibbs states},\ }\href@noop {} {\bibfield  {journal}
  {\bibinfo  {journal} {arXiv preprint arXiv:2108.04842}\ } (\bibinfo {year}
  {2021})}\BibitemShut {NoStop}%
\bibitem [{\citenamefont {Greenberger}\ \emph {et~al.}(1989)\citenamefont
  {Greenberger}, \citenamefont {Horne},\ and\ \citenamefont {Zeilinger}}]{GHZ}%
  \BibitemOpen
  \bibfield  {author} {\bibinfo {author} {\bibfnamefont {D.~M.}\ \bibnamefont
  {Greenberger}}, \bibinfo {author} {\bibfnamefont {M.~A.}\ \bibnamefont
  {Horne}},\ and\ \bibinfo {author} {\bibfnamefont {A.}~\bibnamefont
  {Zeilinger}},\ }\bibfield  {title} {\bibinfo {title} {Going beyond bell’s
  theorem},\ }in\ \href@noop {} {\emph {\bibinfo {booktitle} {Bell’s theorem,
  quantum theory and conceptions of the universe}}}\ (\bibinfo  {publisher}
  {Springer},\ \bibinfo {year} {1989})\ pp.\ \bibinfo {pages}
  {69--72}\BibitemShut {NoStop}%
\bibitem [{\citenamefont {Cruz}\ \emph {et~al.}(2018)\citenamefont {Cruz},
  \citenamefont {Fournier}, \citenamefont {Gremion}, \citenamefont {Jeannerot},
  \citenamefont {Komagata}, \citenamefont {Tosic}, \citenamefont
  {Thiesbrummel}, \citenamefont {Chan}, \citenamefont {Macris}, \citenamefont
  {Dupertuis},\ and\ \citenamefont {Javerzac-Galy}}]{Cruz2018}%
  \BibitemOpen
  \bibfield  {author} {\bibinfo {author} {\bibfnamefont {D.}~\bibnamefont
  {Cruz}}, \bibinfo {author} {\bibfnamefont {R.}~\bibnamefont {Fournier}},
  \bibinfo {author} {\bibfnamefont {F.}~\bibnamefont {Gremion}}, \bibinfo
  {author} {\bibfnamefont {A.}~\bibnamefont {Jeannerot}}, \bibinfo {author}
  {\bibfnamefont {K.}~\bibnamefont {Komagata}}, \bibinfo {author}
  {\bibfnamefont {T.}~\bibnamefont {Tosic}}, \bibinfo {author} {\bibfnamefont
  {J.}~\bibnamefont {Thiesbrummel}}, \bibinfo {author} {\bibfnamefont {C.~L.}\
  \bibnamefont {Chan}}, \bibinfo {author} {\bibfnamefont {N.}~\bibnamefont
  {Macris}}, \bibinfo {author} {\bibfnamefont {M.~A.}\ \bibnamefont
  {Dupertuis}},\ and\ \bibinfo {author} {\bibfnamefont {C.}~\bibnamefont
  {Javerzac-Galy}},\ }\bibfield  {title} {\bibinfo {title} {{Efficient quantum
  algorithms for GHZ and W states, and implementation on the IBM quantum
  computer}},\ }\href {https://doi.org/10.1002/qute.201900015} {\bibfield
  {journal} {\bibinfo  {journal} {arXiv}\ ,\ \bibinfo {pages} {1}} (\bibinfo
  {year} {2018})},\ \Eprint {https://arxiv.org/abs/1807.05572}
  {arXiv:1807.05572} \BibitemShut {NoStop}%
\bibitem [{\citenamefont {Altepeter}\ \emph {et~al.}(2003)\citenamefont
  {Altepeter}, \citenamefont {Branning}, \citenamefont {Jeffrey}, \citenamefont
  {Wei}, \citenamefont {Kwiat}, \citenamefont {Thew}, \citenamefont {O'Brien},
  \citenamefont {Nielsen},\ and\ \citenamefont
  {White}}]{PhysRevLett.90.193601}%
  \BibitemOpen
  \bibfield  {author} {\bibinfo {author} {\bibfnamefont {J.~B.}\ \bibnamefont
  {Altepeter}}, \bibinfo {author} {\bibfnamefont {D.}~\bibnamefont {Branning}},
  \bibinfo {author} {\bibfnamefont {E.}~\bibnamefont {Jeffrey}}, \bibinfo
  {author} {\bibfnamefont {T.~C.}\ \bibnamefont {Wei}}, \bibinfo {author}
  {\bibfnamefont {P.~G.}\ \bibnamefont {Kwiat}}, \bibinfo {author}
  {\bibfnamefont {R.~T.}\ \bibnamefont {Thew}}, \bibinfo {author}
  {\bibfnamefont {J.~L.}\ \bibnamefont {O'Brien}}, \bibinfo {author}
  {\bibfnamefont {M.~A.}\ \bibnamefont {Nielsen}},\ and\ \bibinfo {author}
  {\bibfnamefont {A.~G.}\ \bibnamefont {White}},\ }\bibfield  {title} {\bibinfo
  {title} {Ancilla-assisted quantum process tomography},\ }\href
  {https://doi.org/10.1103/PhysRevLett.90.193601} {\bibfield  {journal}
  {\bibinfo  {journal} {Phys. Rev. Lett.}\ }\textbf {\bibinfo {volume} {90}},\
  \bibinfo {pages} {193601} (\bibinfo {year} {2003})}\BibitemShut {NoStop}%
\bibitem [{\citenamefont {ANIS}\ \emph {et~al.}(2021)\citenamefont {ANIS} \emph
  {et~al.}}]{Qiskit}%
  \BibitemOpen
  \bibfield  {author} {\bibinfo {author} {\bibfnamefont {M.~S.}\ \bibnamefont
  {ANIS}} \emph {et~al.},\ }\href {https://doi.org/10.5281/zenodo.2573505}
  {\bibinfo {title} {Qiskit: An open-source framework for quantum computing}}
  (\bibinfo {year} {2021})\BibitemShut {NoStop}%
\bibitem [{\citenamefont {Bravyi}\ and\ \citenamefont
  {Maslov}(2021)}]{BravyiClifford}%
  \BibitemOpen
  \bibfield  {author} {\bibinfo {author} {\bibfnamefont {S.}~\bibnamefont
  {Bravyi}}\ and\ \bibinfo {author} {\bibfnamefont {D.}~\bibnamefont
  {Maslov}},\ }\bibfield  {title} {\bibinfo {title} {Hadamard-free circuits
  expose the structure of the clifford group},\ }\href
  {https://doi.org/10.1109/TIT.2021.3081415} {\bibfield  {journal} {\bibinfo
  {journal} {IEEE Transactions on Information Theory}\ }\textbf {\bibinfo
  {volume} {67}},\ \bibinfo {pages} {4546} (\bibinfo {year}
  {2021})}\BibitemShut {NoStop}%
\bibitem [{\citenamefont {Smolin}\ \emph
  {et~al.}(2012{\natexlab{b}})\citenamefont {Smolin}, \citenamefont
  {Gambetta},\ and\ \citenamefont {Smith}}]{Smolin2012}%
  \BibitemOpen
  \bibfield  {author} {\bibinfo {author} {\bibfnamefont {J.~A.}\ \bibnamefont
  {Smolin}}, \bibinfo {author} {\bibfnamefont {J.~M.}\ \bibnamefont
  {Gambetta}},\ and\ \bibinfo {author} {\bibfnamefont {G.}~\bibnamefont
  {Smith}},\ }\bibfield  {title} {\bibinfo {title} {Efficient method for
  computing the maximum-likelihood quantum state from measurements with
  additive gaussian noise},\ }\href
  {https://doi.org/10.1103/PhysRevLett.108.070502} {\bibfield  {journal}
  {\bibinfo  {journal} {Phys. Rev. Lett.}\ }\textbf {\bibinfo {volume} {108}},\
  \bibinfo {pages} {070502} (\bibinfo {year} {2012}{\natexlab{b}})}\BibitemShut
  {NoStop}%
\bibitem [{\citenamefont {Knee}\ \emph
  {et~al.}(2018{\natexlab{b}})\citenamefont {Knee}, \citenamefont {Bolduc},
  \citenamefont {Leach},\ and\ \citenamefont {Gauger}}]{knee2018quantum}%
  \BibitemOpen
  \bibfield  {author} {\bibinfo {author} {\bibfnamefont {G.~C.}\ \bibnamefont
  {Knee}}, \bibinfo {author} {\bibfnamefont {E.}~\bibnamefont {Bolduc}},
  \bibinfo {author} {\bibfnamefont {J.}~\bibnamefont {Leach}},\ and\ \bibinfo
  {author} {\bibfnamefont {E.~M.}\ \bibnamefont {Gauger}},\ }\bibfield  {title}
  {\bibinfo {title} {Quantum process tomography via completely positive and
  trace-preserving projection},\ }\href@noop {} {\bibfield  {journal} {\bibinfo
   {journal} {Physical Review A}\ }\textbf {\bibinfo {volume} {98}},\ \bibinfo
  {pages} {062336} (\bibinfo {year} {2018}{\natexlab{b}})}\BibitemShut
  {NoStop}%
\bibitem [{\citenamefont {Wright}\ \emph {et~al.}(2019)\citenamefont {Wright},
  \citenamefont {Beck}, \citenamefont {Debnath}, \citenamefont {Amini},
  \citenamefont {Nam}, \citenamefont {Grzesiak}, \citenamefont {Chen},
  \citenamefont {Pisenti}, \citenamefont {Chmielewski}, \citenamefont {Collins}
  \emph {et~al.}}]{wright2019benchmarking}%
  \BibitemOpen
  \bibfield  {author} {\bibinfo {author} {\bibfnamefont {K.}~\bibnamefont
  {Wright}}, \bibinfo {author} {\bibfnamefont {K.}~\bibnamefont {Beck}},
  \bibinfo {author} {\bibfnamefont {S.}~\bibnamefont {Debnath}}, \bibinfo
  {author} {\bibfnamefont {J.}~\bibnamefont {Amini}}, \bibinfo {author}
  {\bibfnamefont {Y.}~\bibnamefont {Nam}}, \bibinfo {author} {\bibfnamefont
  {N.}~\bibnamefont {Grzesiak}}, \bibinfo {author} {\bibfnamefont {J.-S.}\
  \bibnamefont {Chen}}, \bibinfo {author} {\bibfnamefont {N.}~\bibnamefont
  {Pisenti}}, \bibinfo {author} {\bibfnamefont {M.}~\bibnamefont
  {Chmielewski}}, \bibinfo {author} {\bibfnamefont {C.}~\bibnamefont
  {Collins}}, \emph {et~al.},\ }\bibfield  {title} {\bibinfo {title}
  {Benchmarking an 11-qubit quantum computer},\ }\href@noop {} {\bibfield
  {journal} {\bibinfo  {journal} {Nature communications}\ }\textbf {\bibinfo
  {volume} {10}},\ \bibinfo {pages} {1} (\bibinfo {year} {2019})}\BibitemShut
  {NoStop}%
\bibitem [{\citenamefont {Kunjummen}\ \emph {et~al.}(2023)\citenamefont
  {Kunjummen}, \citenamefont {Tran}, \citenamefont {Carney},\ and\
  \citenamefont {Taylor}}]{kunjummen2021shadow}%
  \BibitemOpen
  \bibfield  {author} {\bibinfo {author} {\bibfnamefont {J.}~\bibnamefont
  {Kunjummen}}, \bibinfo {author} {\bibfnamefont {M.~C.}\ \bibnamefont {Tran}},
  \bibinfo {author} {\bibfnamefont {D.}~\bibnamefont {Carney}},\ and\ \bibinfo
  {author} {\bibfnamefont {J.~M.}\ \bibnamefont {Taylor}},\ }\bibfield  {title}
  {\bibinfo {title} {Shadow process tomography of quantum channels},\ }\href
  {https://doi.org/10.1103/PhysRevA.107.042403} {\bibfield  {journal} {\bibinfo
   {journal} {Phys. Rev. A}\ }\textbf {\bibinfo {volume} {107}},\ \bibinfo
  {pages} {042403} (\bibinfo {year} {2023})}\BibitemShut {NoStop}%
\bibitem [{\citenamefont {Cincio}\ \emph {et~al.}(2018)\citenamefont {Cincio},
  \citenamefont {Subaşi}, \citenamefont {Sornborger},\ and\ \citenamefont
  {Coles}}]{Cincio2018}%
  \BibitemOpen
  \bibfield  {author} {\bibinfo {author} {\bibfnamefont {L.}~\bibnamefont
  {Cincio}}, \bibinfo {author} {\bibfnamefont {Y.}~\bibnamefont {Subaşi}},
  \bibinfo {author} {\bibfnamefont {A.~T.}\ \bibnamefont {Sornborger}},\ and\
  \bibinfo {author} {\bibfnamefont {P.~J.}\ \bibnamefont {Coles}},\ }\bibfield
  {title} {\bibinfo {title} {{Learning the quantum algorithm for state
  overlap}},\ }\href {https://doi.org/10.1088/1367-2630/aae94a} {\bibfield
  {journal} {\bibinfo  {journal} {New Journal of Physics}\ }\textbf {\bibinfo
  {volume} {20}},\ \bibinfo {pages} {1} (\bibinfo {year} {2018})},\ \Eprint
  {https://arxiv.org/abs/1803.04114} {arXiv:1803.04114} \BibitemShut {NoStop}%
\end{thebibliography}%

\clearpage

\onecolumngrid
\begin{center}
  \noindent\textbf{Appendix}
  \bigskip
    
  \noindent\textbf{\large{Classical Shadows of Quantum Process Tomography on Near-term Quantum Computers}}
    
\end{center}

\appendix

\renewcommand\thefigure{A\arabic{figure}}  
\renewcommand\thetable{A\arabic{table}}  
\setcounter{figure}{0}  
\setcounter{table}{0}

\section{Proofs for Shadow Process Tomography}\label{app:proof}

We first review the formal version of a classical shadow quantum state tomography theorem in Ref.~\cite{Huang2020}.

\begin{theorem}\label{thm:shadow-formal}
Fix a measurement primitive $U$, a collection $\{O_i\}_{i=1,\dots,M}$ of $2^n \times 2^n$ Hermitian matrices and accuracy parameters $\epsilon,\delta \in (0,1)$. Set $K= 2 \textup{log}(2M/\delta)$ and $N=\frac{34}{\epsilon^2} \textup{max}_i ||O_i - \frac{\textup{tr}(O_i)}{2}\mathds{I}||^2_{shadow}$, where $||.||_{shadow}$ is the shadow norm defined with the inverse channel $\mathcal{M}^{-1}$ induced by the measurement primitive as follows,

\begin{equation}
    ||O||_{shadow}=max_{\sigma} (\mathbb{E}_{U \sim \mathcal{U}} \sum_{b \in \{0,1\}^{n}} \bra{b}U\sigma U^{\dagger}\ket{b}\bra{b}U\mathcal{M}^{-1}(O)U^{\dagger}\ket{b}^{2})^{1/2}
\end{equation}

Then, a collection of $NK$ independent classical shadows allow for accurately predicting all features via median of means prediction $\hat{o_i}$:

\begin{equation}
    |\hat{o_i}(N,K) - \textup{tr}(O_i \rho)| \leq \epsilon 
\end{equation}
$\textup{for all}$ $1 \leq i \leq M$  with probability at least $1-\delta$.

In particular, for random global Clifford measurement primitive, $||O - \frac{\textup{tr}(O)}{2}\mathds{I}||^2_{shadow} \leq 3 \textup{tr}(O^2)$. For random single qubit Clifford measurement primitive, $||O - \frac{\textup{tr}(O)}{2}\mathds{I}||^2_{shadow} \leq 4^k ||O||_{\infty}^2$ where $k$  is the locality of the operator $O$ and $||.||_{\infty}$ is the spectral norm. Furthermore, if $O$ is a single $k$-local Pauli observable, $||O - \frac{\textup{tr}(O)}{2}\mathds{I}||^2_{shadow} \leq 3^k $.
\end{theorem}

\quad

We now combine the above theorem and the Jamio\l{}kowski process-state isomorphism~\cite{JAMIOLKOWSKI1972275} to reach the following two theorems for shadow quantum process tomography.

\begin{theorem}
For a $n$-qubit quantum process $\rho_{\Lambda}$ and $\epsilon,\delta \in (0,1)$, given a set of density matrice pairs  $\{(\rho^{in}_1,\sigma_1), \dots,(\rho^{in}_{M},\sigma_M)\}$, the number of measurements $N$ that suffices to predict the overlaps $\textup{tr}(\rho_{\Lambda}(\rho^{in\,T}_i \otimes \sigma_i))$ for any $i$ up to error $\epsilon$ with probability $1-\delta$ is

\begin{equation}
    N = \frac{68}{\epsilon^2} \textup{log}(2M/\delta) \textup{max}_{i} ||O_{i} - \frac{\textup{tr}(O_{i})}{2}\mathds{I}||^2_{shadow}
\end{equation}
where $O_{i}=\rho^{in\,T}_i \otimes \sigma_i$. For random global Clifford measurement,  $N = \frac{68}{\epsilon^2} \textup{log}(2M/\delta) \textup{max}_{i} \textup{tr}(O_{i}^2)$. If $\rho^{in}_i$ and $\sigma_i$ are all pure states, then $\textup{tr}(O_{i}^2)=1$. For random single qubit Clifford measurement, $N = \frac{68}{\epsilon^2} \textup{log}(2M/\delta) \textup{max}_{i} 4^{k_{i}} ||O_{i}||_{\infty}^2$, where $O_{i}$ acts nontrivially on $k_i$-qubits. If $\rho^{in}_i$ and $\sigma_i$ are all pure states, then $||O_{i}||_{\infty}^2=1$.

\end{theorem}

\begin{proof}
Under the the Jamio\l{}kowski process-state isomorphism~\cite{JAMIOLKOWSKI1972275},

\begin{equation}
    \textup{tr}(\rho_{\Lambda} (\rho^{in\,T}_i \otimes \sigma_i))=\textup{tr}(\rho_{\Lambda} (O_{i})),
\end{equation}

The result follows directly from applying Theorem.~\ref{thm:shadow-formal} with $M$ measurements.
For random global Clifford measurement and pure states $\rho^{in}_i$ and $\sigma_i$, $O_{i}=\rho^{in\,T}_i \otimes \sigma_i$ is also a pure state of rank 1 so that $\textup{tr}(O_{i}^2)=1$. It implies that $||O_i - \frac{\textup{tr}(O_i)}{2}\mathds{I}||^2_{shadow} \leq 3 \textup{tr}(O_i^2) = 3$. For random single qubit Clifford measurement and pure states $\rho^{in}_i$ and $\sigma_i$, $O_{i}=\rho^{in\,T}_i \otimes \sigma_i$ is also a pure state so that $||O_{i}||_{\infty}^2=1$. 
\qed
\end{proof}

\begin{theorem}
For a $n$-qubit process Choi matrix $\rho$ and $\epsilon,\delta \in (0,1)$, the following statements hold:

(i) the number of measurements $N$ that suffices to simultaneously predict any reduced $k$-qubit process Choi matrix $\rho_k$ up to error $\epsilon$ in the Frobenius norm and probability $1-\delta$ is

\begin{equation}\label{eq:A2_N}
    N = \frac{68}{\epsilon^2} 4^k \textup{log}(2 (8n)^{2k} / \delta)max_{i} ||O_i^{(2k)}- \frac{\textup{tr}(O_i^{(2k)})}{2}\mathds{I}||_{shadow}^2
\end{equation}
where $O_i^{(2k)}$ is any $2k$-qubit Pauli observable.

In particular, for random global Clifford measurement, $N = \frac{204}{\epsilon^2} 4^{n+k} \textup{log}(2 (8n)^{2k} /\delta)$ and for Pauli-6 POVM measurement, $N = \frac{68}{\epsilon^2}36^k \textup{log}(2 (8n)^{2k} /\delta) max_i||O_i^{(2k)}||_{\infty}^2$.

\quad

(ii) the number of Pauli-6 POVM measurements $N$ that suffices to simultaneously predict any reduced $k$-qubit process Choi matrix $\rho_k$ up to error $\epsilon$ in the trace norm and probability $1-\delta$ is

\begin{equation}
    N = \frac{8}{3}\frac{144^k}{\epsilon^2} \textup{log}((24n)^{2k}/\delta)
\end{equation}
\end{theorem}

\begin{proof}
(i) According to the Jamiołkowski process-state isomorphism~\cite{JAMIOLKOWSKI1972275}, a quantum process on a $n$-qubit system can be represented by a density matrix $\rho$ on a $2n$-qubit system. Any reduced $k$-qubit process matrix $\rho_k$ could be expressed as 

\begin{equation}
    \rho_k = \frac{1}{4^k} \sum^{16^k} \alpha_i O_i^{(2k)} 
\end{equation}
where $\alpha_i = \textup{tr}(\rho_k O_i^{(2k)}) = \textup{tr}(\rho O_i^{(2k)} \otimes \mathds{I}_{2n-2k})$.

Denote the estimated $k$-qubit reduced density matrix from the shadow tomography as $\hat{\rho}_{k}$ and $ \hat{\rho}_k = \frac{1}{4^k} \sum^{16^k} \hat{\alpha}_i O_i^{(k)} $, where $\hat{\alpha}_i$ is the shadow estimation of observable $O_i^{2k}$. It follows that 

\begin{equation}
    \langle \hat{\rho}_k - \rho_k, \hat{\rho}_k - \rho_k \rangle_{F} = \langle \frac{1}{4^k} \sum^{16^k} (\hat{\alpha}_j - \alpha) O_j^{(2k)},\frac{1}{4^k} \sum^{16^k} (\hat{\alpha}_i - \alpha) O_i^{(2k)} \rangle_{F} = \frac{1}{4^k} \sum^{16^k} (\hat{\alpha}_i-\alpha_i)^2
\end{equation}
where $\langle \cdot, \cdot \rangle_{F}$ is the Frobenius inner product.

According to Theorem.~\ref{thm:shadow-formal}, with $N$ given in Eq.~\ref{eq:A2_N}, we have 

\begin{equation}
    |\hat{\alpha}_i - \alpha_i|=|\hat{\alpha}_i-\textup{tr}(\rho O_i^{(2k)} \otimes \mathds{I}_{2n-2k})| \leq  \frac{\epsilon}{2^k} \quad \forall \quad i
\end{equation}

It implies that $\langle \hat{\rho}_k - \rho_k, \hat{\rho}_k - \rho_k \rangle_{F} \leq \frac{16^k}{4^k} \frac{\epsilon^2}{4^k}=\epsilon^2$, so that $||\hat{\rho}_k-\rho_k||_{F} \leq \epsilon$ where $||\cdot||_F$ is the Frobenius norm. For random global Clifford measurement, the result follows from $||O - \frac{\textup{tr}(O)}{2}\mathds{I}||^2_{shadow} \leq 3 \textup{tr}(O^2)=3 \times 4^n$ for $O=O_i^{(2k)}$. For Pauli-6 POVM measurement, the result follows from $||O - \frac{\textup{tr}(O)}{2}\mathds{I}||^2_{shadow} \leq 9^k $ for $O=O_i^{(2k)}$~\cite{acharya2021informationally}. 

\quad

(ii) For a given $\epsilon,\delta \in (0,1)$, the Pauli-6 POVM measurement of size $N=\frac{8}{3 \epsilon^2} 12^r (\textup{log} (p^r 12^r / \delta) )$ is sufficient to predict all reduced density matrix on subsystem size $r \leq p$ on a $p$-qubit system with trace norm error $\epsilon$ and probability at least $1-\delta$ according to Lemma 1 in Ref.~\cite{huang2021provably}. With the Jamio\l{}kowski process-state isomorphism~\cite{JAMIOLKOWSKI1972275}, a $k$-qubit process Choi matrix is equivalent to a $2k$-qubit density matrix on a $2n$-qubit system. The result follows from choosing $r=2k$ and $p=2n$, which implies that with probability $1-\delta$, $N = \frac{8}{3}\frac{144^k}{\epsilon^2} \textup{log}((24n)^{2k}/\delta)$ number of Pauli-6 POVM measurements is sufficient to simultaneously predict any reduced $k$-qubit process Choi matrix $\rho_k$ up to error $\epsilon$ in the trace norm, i.e. any estimator $\hat{\rho}_k$ from the shadow measurement is closer to the exact $\rho_k$ by $||\hat{\rho}_k - \rho_k||_1 < \epsilon$. \qed
\end{proof}

\section{Equivalence of Two-sided Scheme with the Top-bottom Factorized Ancilla-based Scheme}\label{app:equivalence}

We would like to show that the two-sided scheme is equivalent to the top-bottom factorized ancilla-based scheme. Precisely, consider a top-bottom factorized ancilla-based scheme with unitaries $\{U_j\}$ acting on the top half and unitaries $\{U_i\}$ acting on the bottom half, where $U_j$, $U_i$ are sampled uniformly from a set $G$. Up to permutation, the reconstruction formula for the Choi matrix is the following:

\begin{equation}
        \rho_{\Lambda}^{a} = \mathbb{E}_{U_j \sim G} \mathbb{E}_{U_i \sim G} \mathbb{E}_{b_j, b_i \sim P(i,j)} \  \mathcal{M}^{-1}_n(U_i^{\dagger}\ket{b_i}\bra{b_i}U_i) \otimes \mathcal{M}^{-1}_n(U_j^{\dagger}\ket{b_j}\bra{b_j}U_j)
\end{equation}
where $b_j,b_i$ are measurement outcome with respect to $U_j, U_i$ in the computational basis, $P(i,j)=\textup{tr}(\rho_{\Lambda}(U_i^{\dagger}\ket{b_i}\bra{b_i}U_i \otimes U_j^{\dagger}\ket{b_j}\bra{b_j}U_j))$. 
$\mathcal{M}^{-1}_n$ is the inverse channel defined by $\mathcal{M}^{-1}_n(X) = (2^n+1)X-\mathds{I}$ for $n$-qubit Clifford unitaries or $\mathcal{M}^{-1}_n(X) = \otimes_i M_k^{-1}(X)$ for Pauli ($k=1$) for smaller $k$ qubit Cliffords. 

Meanwhile, consider a two-sided scheme where the input state is $\ket{0}$, followed by $\{U_i^{\dagger}\}$, process channel $\process$, $\{U_j\}$ and finally measurement in the computational basis shown in Fig.~\ref{fig:two_sided}. $U_j$, $U_i$ are also sampled uniformly from group $G$. The reconstruction formula for the Choi matrix is given as follows:

\begin{equation}
        \rho_{\Lambda}^{t} = \mathbb{E}_{U_j \sim G} \mathbb{E}_{U_i \sim G} \mathbb{E}_{b_j \sim P(j|[U_i^{\dagger}\ket{0}])} \ \mathcal{M}^{-1}_n((U_i^{\dagger}\ket{0}\bra{0}U_i)^T) \otimes \mathcal{M}^{-1}_n(U_j^{\dagger}\ket{b_j}\bra{b_j}U_j)
\end{equation} 
where $P(j|[U_i^{\dagger}\ket{0}])=\textup{tr}(\process(U_i^{\dagger}\ket{0}\bra{0}U_i) U_j^{\dagger}\ket{b_j}\bra{b_j}U_j)=2^n \textup{tr}(\rho_{\Lambda}((U_i^{\dagger}\ket{0}\bra{0}U_i)^T\otimes U_j^{\dagger}\ket{b_j}\bra{b_j}U_j))$. 
The last line comes from the Choi matrix identification. For any set $G$ that is right invariant with respect to any Pauli-X matrix such that $GX = G$, we have

\begin{equation}
        \rho_{\Lambda}^{t} = \mathbb{E}_{U_j \sim G} \mathbb{E}_{U_i \sim G} \mathbb{E}_{b_i \sim \frac{1}{2^n}} \mathbb{E}_{b_j \sim P(j|[U_i^{\dagger}\ket{b_i}])} \, \mathcal{M}^{-1}_n( (U_i^{\dagger}\ket{b_i}\bra{b_i}U_i)^T) \otimes \mathcal{M}^{-1}_n( U_j^{\dagger}\ket{b_j}\bra{b_j}U_j)
\end{equation}

This is true because $\ket{b_i} = \otimes^{\{\ell\}} X \ket{0}$ where $\{\ell\}$ are the locations where the measurement outcome is 1 and hence for any $U_i^{\dagger} \ket{0}$, there exists $U'_i$ from $G$ such that $U'_i \ket{b_i} = U_i \ket{0}$. Notice that this condition is true for random global Clifford measurement, random single qubit Clifford measurement, as well as Pauli-6 measurement. It follows that,

\begin{equation}
        \rho_{\Lambda}^{t} = \mathbb{E}_{U_j \sim G} \mathbb{E}_{U_i \sim G} \mathbb{E}_{b_j, b_i \sim P^{t}(i,j)} \  \mathcal{M}^{-1}_n( (U_i^{\dagger}\ket{b_i}\bra{b_i}U_i)^T) \otimes \mathcal{M}^{-1}_n( U_j^{\dagger}\ket{b_j}\bra{b_j}U_j) 
\end{equation}
where $P^{t}(i,j)=\frac{1}{2^n}P(j|[U_i^{\dagger}\ket{b_i}])=\textup{tr}(\rho_{\Lambda}((U_i^{\dagger}\ket{b_i}\bra{b_i}U_i)^T \otimes U_j^{\dagger}\ket{b_j}\bra{b_j}U_j))$. For Pauli-6 POVM, we have $(U_i^{\dagger}\ket{b_i}\bra{b_i}U_i)^T=U_i^{\dagger}\ket{b_i}\bra{b_i}U_i$, which implies $\rho^{t}_{\Lambda}=\rho^{a}_{\Lambda}$. For Clifford group $G$ and any fixed $\ket{b}$, we have $\{U_i^{\dagger}\ket{b}\bra{b}U_i|U_i \in G\}=\{(U_i^{\dagger}\ket{b}\bra{b}U_i)^T |U_i \in G\}$, which also implies $\rho^{t}_{\Lambda}=\rho^{a}_{\Lambda}$. 

We remark that the above argument holds for a single repetition measurement. For multiple repetition measurement with respect to one $U_i$, the equivalence breaks because multiple $\ket{b_i}$ are generated with respect to one $U_i$ in the ancilla-based scheme while the two-sided scheme always has input $\ket{0}$ fixed for one $U_i$.

\section{Details on Shadow Measurements}\label{app:shots}

We present the number of random unitary circuits used for both Clifford and Pauli measurements. For those with unitary entries with *, all possible unitaries have been used.

\begin{table}[h!]
    \centering
    \begin{tabular}{|c|c|c|c|}
    \hline 
    Qubits & Type & Unitaries & Repetitions/Unitary (Total)\tabularnewline
    \hline 
    \hline 
    \multirow{2}{*}{2} & Pauli & 81{*} & $\approx 632$ (51200)\tabularnewline
    \cline{2-4} \cline{3-4} \cline{4-4} 
     & Clifford & 1024 & 50 (51200)\tabularnewline
    \hline 
    \multirow{2}{*}{3} & Pauli & 729{*} & $\approx 70$ (51200)\tabularnewline
    \cline{2-4} \cline{3-4} \cline{4-4} 
     & Clifford & 1024 & 50 (51200)\tabularnewline
    \hline 
    \multirow{2}{*}{4} & Pauli & 1024 & 50 (51200)\tabularnewline
    \cline{2-4} \cline{3-4} \cline{4-4} 
     & Clifford & 1024 & 50 (51200)\tabularnewline
    \hline 
    \end{tabular}
\end{table}

\section{Comparison between Median-of-Means and Mean}
\begin{figure}[ht]
    \centering
    \includegraphics[width=0.5\columnwidth]{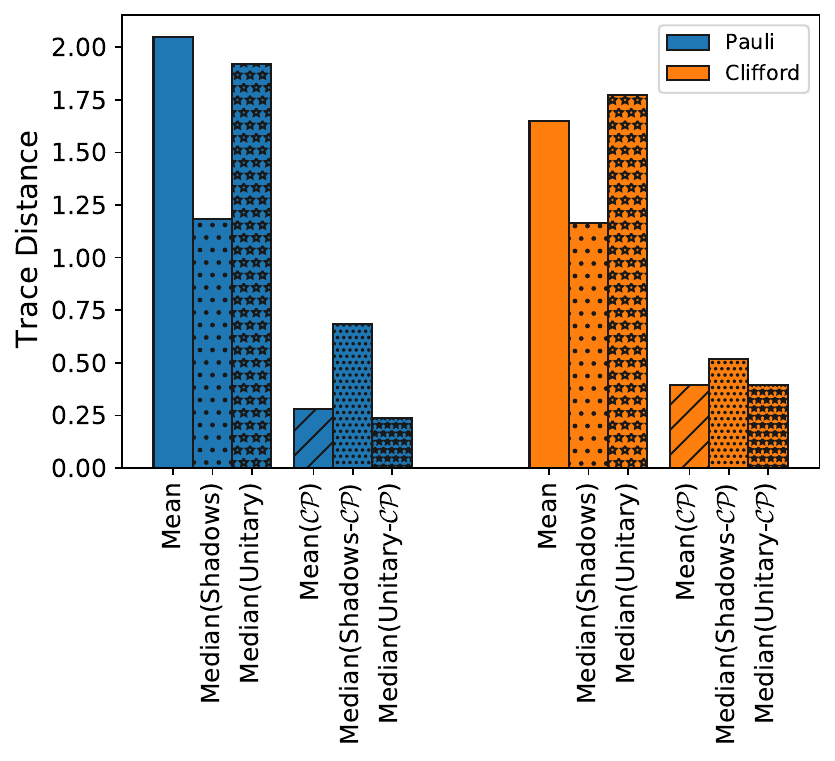}
    \caption{Normalized trace distance $T(\Lambda,\Lambda^O)$ between the unitary Choi matrix $\Lambda$ and a Pauli/Clifford shadow reconstructed Choi matrix $\Lambda^O$ for the $n=3$ GHZ process. The construction is either via mean, median of means via randomizing and batching individual shadows (with many derived from the same unitary), or median of means for each unique unitary (which has 50/$\approx$70 repetitions each). We include the effects of $\mathcal{CP}$ projection on the resulting Choi matrices for comparison.  }
    \label{fig:median_vs_mean}
\end{figure}

In Ref.~\cite{Huang2020}, the authors suggest a median-of-means procedure for the classical shadows, or more precisely taking the median of $K$ averaged shadows. Then in Ref.~\cite{struchalin2021experimental} they show that most of their results are unaffected by using the direct mean over a median-of-means. 

To understand the role of the procedure in ShadowQPT, we first distinguish between two median-of-means procedures. The first a median on the `shadow' level: take each classical shadow, many of which are from the same unitary, and perform the median-of-means procedure using this data set. An alternative procedure, is to first average the repetitions corresponding to each unitary, then obtain the median by using sets of $K$ averaged unitaries.

In Fig.~\ref{fig:median_vs_mean} is shown the results of using the mean and the two median-of-mean procedures for $n=3$ data, with a batch size of $K=23$. While using a median-of-means at the shadow level does show a slight improvement for both Pauli and Clifford measurements  without projection, after projection we find that there is little difference between the mean and median-of-means results. Given the weak dependence on the differences after post-processing and the potential introduction of a tunable variable ($K$), we chose to utilize the mean for our computing results of the classical shadows.

\section{3D Choi matrix plots}
 \begin{figure*}
     \centering
     \subfloat[]{\includegraphics[width=\columnwidth]{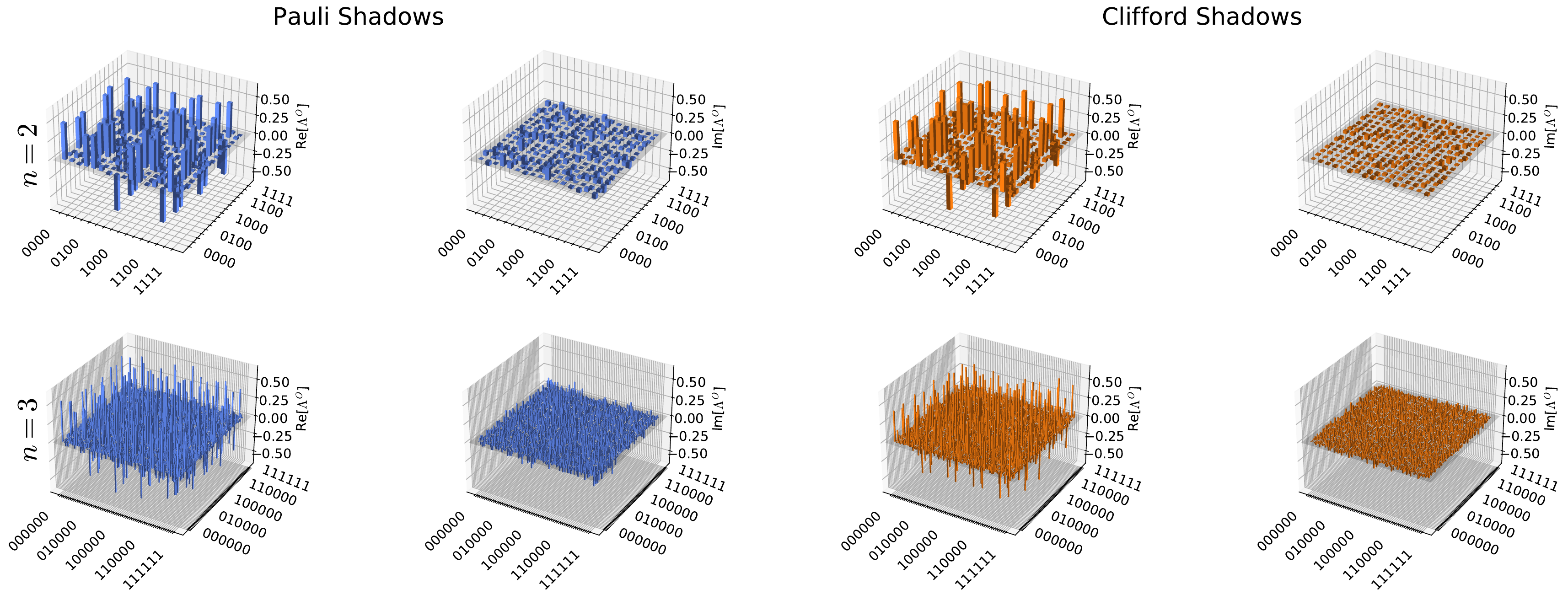}\label{fig:city}}
     \\
     \subfloat[]{\includegraphics[width=\columnwidth]{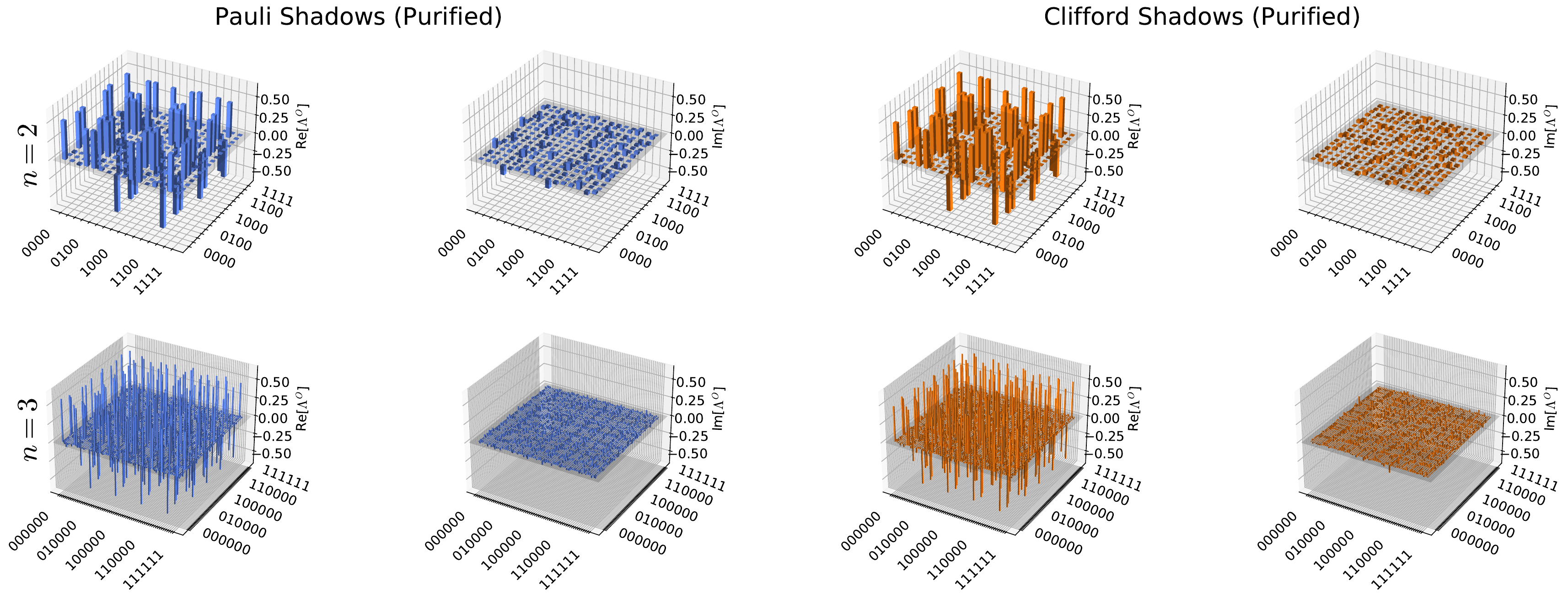}\label{fig:city_pure}}
       
     \caption{(a) Choi matrix visualization using Pauli (left two) and Clifford (right two) shadows, separated into real and imaginary components. 
     (b) Purification of matrices in (a). }
     \label{fig:dual_city_plots}
 \end{figure*}

In Fig.~\ref{fig:dual_city_plots} we present a 3D representation of $n=2,3$ in Fig.~\ref{fig:process_visualization}. The axes are shown enumerated in binary for the Choi matrices $4^n$ index values.

\section{Post-processing dependence}
\begin{figure}
    \centering
    \includegraphics[width=0.5\columnwidth]{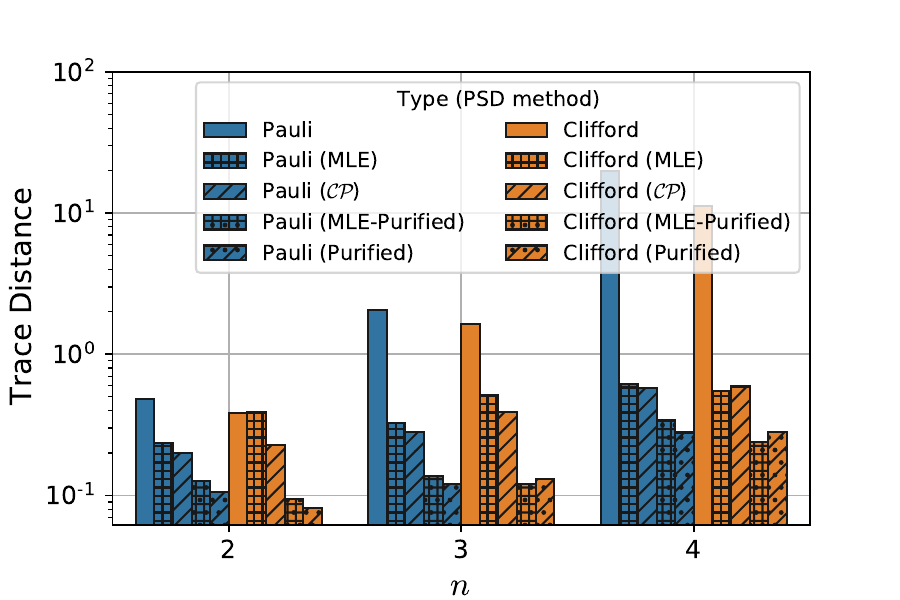}
    \caption{Normalized trace distance $T(\Lambda,\Lambda^O)$ between the unitary Choi matrix $\Lambda$ and a Pauli/Clifford shadow reconstructed Choi matrix $\Lambda^O$. Included are two projections of the $\Lambda^O$ into the space of positive semidefinite (PSD) matrices, the iterative maximum-likelihood estimator (MLE) method of Ref.~\cite{Lvovsky_2004} and the eigenvalue rescaling method of Ref.~\cite{Smolin2012}, as well as the corresponding purification.  }
    \label{fig:full_trace_distance_bar}
\end{figure}

\begin{figure}
    \centering
    \subfloat[]{\includegraphics[width=0.5\columnwidth]{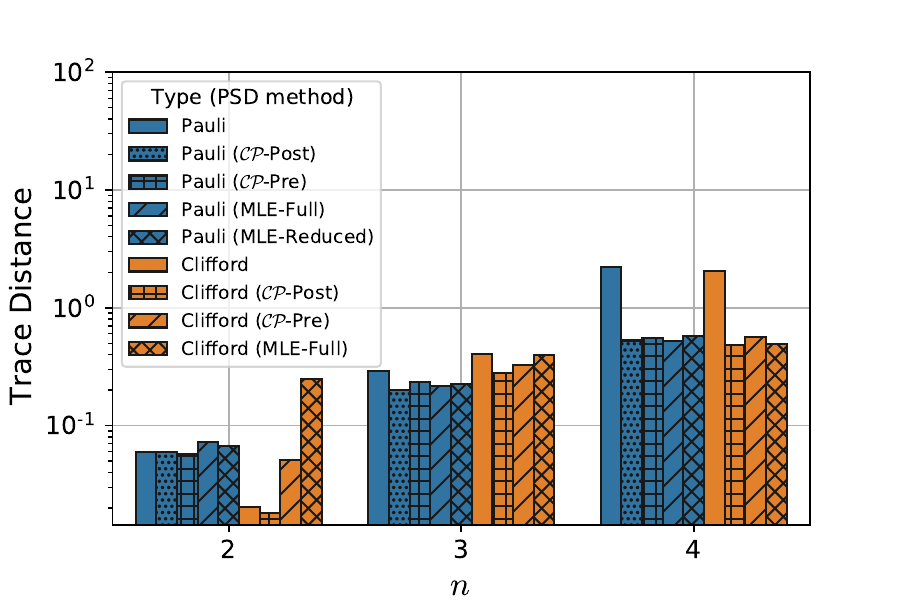} }
    \subfloat[]{\includegraphics[width=0.5\columnwidth]{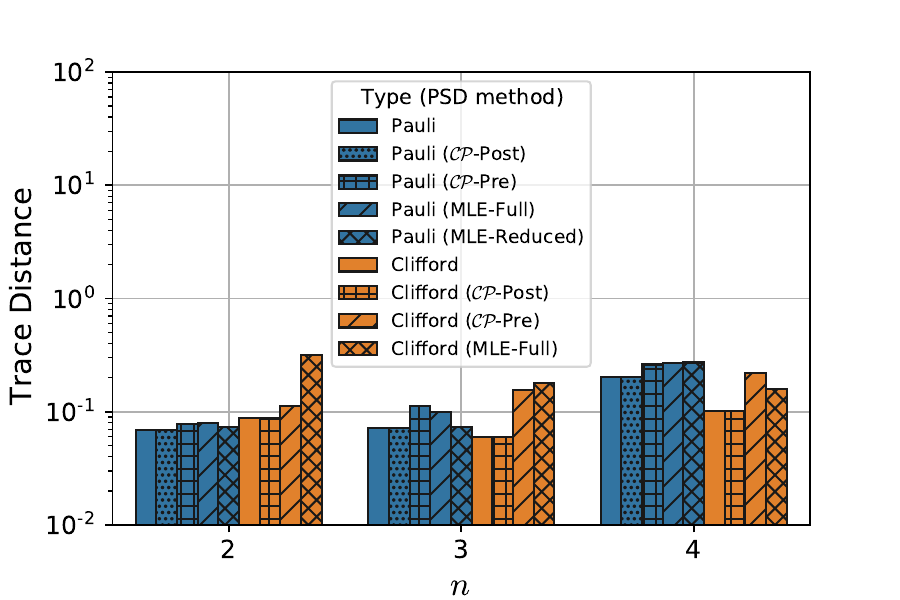}}
    \caption{ Normalized trace distance $T(\Lambda_B,\Lambda^O_B)$ of the subsystem of the first qubit  ($k=1$) (a) and (b) the remaining $n-1$ qubits ($k=n-1$)  between the unitary Choi matrix $\Lambda_B$ and a Pauli/Clifford shadow reconstructed Choi matrix $\Lambda^O_B$. }
    \label{fig:reduced_trace_distance_specific}
\end{figure}

\begin{figure}
    \centering
    \subfloat[]{\includegraphics[width=0.53\columnwidth]{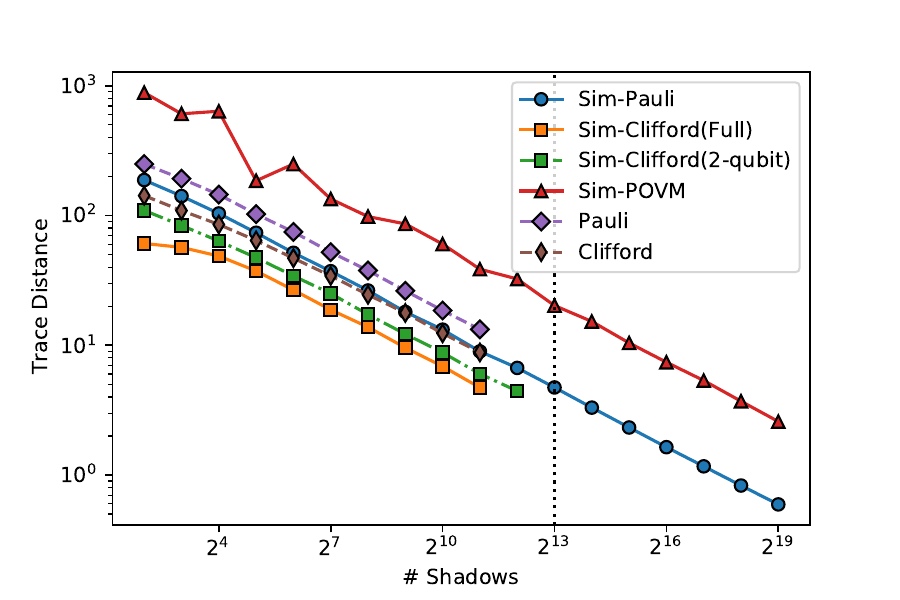} }
    \subfloat[]{\includegraphics[width=0.47\columnwidth]{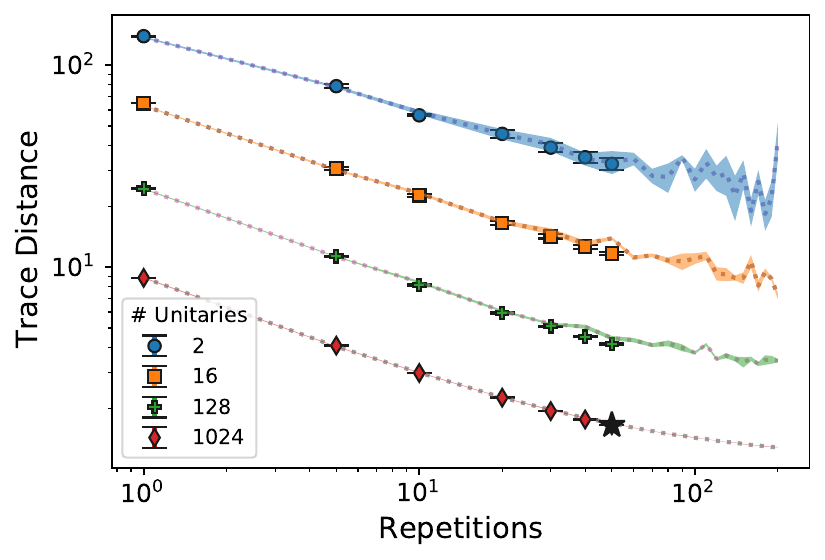}}
    \caption{Left (a) Normalized trace distance $T(\Lambda,\Lambda^O)$ between the unitary Choi matrix $\Lambda$ and various simulated reconstruction methods. Simulated data includes a prefix of ``Sim". The vertical line represents saturation of all Pauli strings.
    Right (b) Normalized trace distance $T(\Lambda,\Lambda^O)$ between the unitary Choi matrix $\Lambda$ and a Clifford shadow reconstruction from resampled data (Markers) on top of simulated data (dotted lines with shading). A star marks the full IonQ dataset; errorbars represent an average of 5 resampling trials. Resampled IonQ data consists of equal fixed and random Clifford orderings.  }
    \label{fig:trace_distance_vs_shots}
\end{figure}

In order to show the dependence on post-processing method, Fig.~\ref{fig:full_trace_distance_bar} shows the normalized trace distance between MLE reconstruction, $\mathcal{CP}$-projection, and purification. Overall the results are roughly consistent across system size $n$, in that purified shadows are competitive with purification of (Clifford) MLE results. Just considering $\mathcal{CP}$ projection and standard MLE, Pauli measurements have the minimal trace distance.

\section{Repetition dependence on trace distance}

Next we compare the dependence on the trace distance with the number of unitaries and classical shadows for $n=3$ in Fig.~\ref{fig:trace_distance_vs_shots}. 

Second, we resample the Clifford measurements from the IonQ alongside simulated data, computing the normalized trace distance in Fig.~\ref{fig:trace_distance_vs_shots}b. We find that the measured data tracks the simulation very closely for a number of chosen unitary Clifford measurements. Although we predict performance could be increased by increasing the number of repetitions/circuit (i.e. we had not reached instability in trace distance seen in the low unitary counts, the blue and orange lines of Fig.~\ref{fig:trace_distance_vs_shots}b), the number of individual circuits and repetitions/circuit chosen appears to be reasonable in the study of our method.  Although there is an instability in the regime where the number of repetitions is greater than the number of unitary, the trace distance remains improved compared to a single measurement.

\section{Distance using the Frobenius norm }

We provide additional plots of Figs.~\ref{fig:full_trace_distance},\ref{fig:two_sided_plots}, \ref{fig:reduced_trace_distance_1qubit}, and \ref{fig:reduced_trace_distance_2qubits} using a Frobenius norm distance shown in Fig.~\ref{fig:l2_norm_version} and \ref{fig:two_sided_plots_l2}. We find nearly the same qualitative behavior for both distance measures. This is unsurprising as, at least for the Pauli case, there is similar scaling bounds of $\log(n)$ for fixed $k$ (see Theorem.~\ref{thm:k-qubit}). 

\begin{figure}
    \centering
    \subfloat[]{\includegraphics[width=0.5\columnwidth]{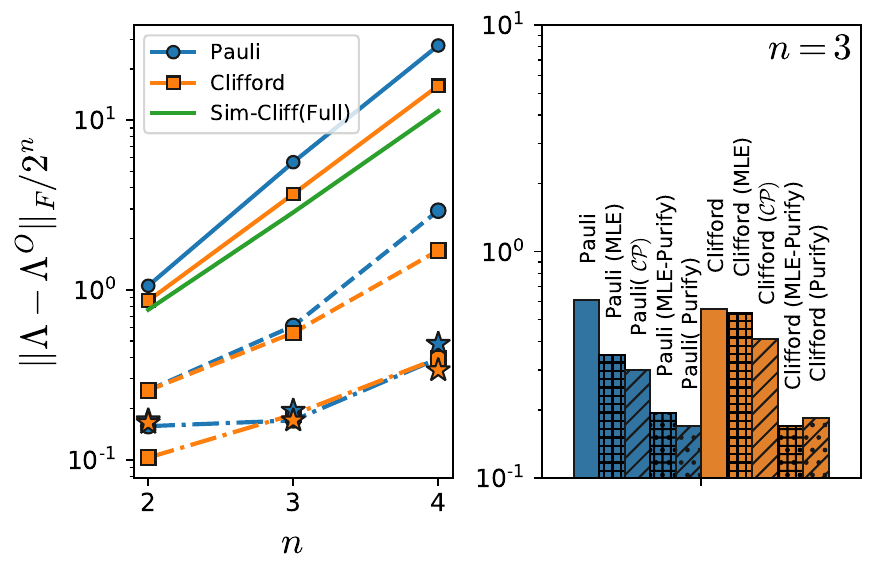} }\\
    \subfloat[]{\includegraphics[width=0.5\columnwidth]{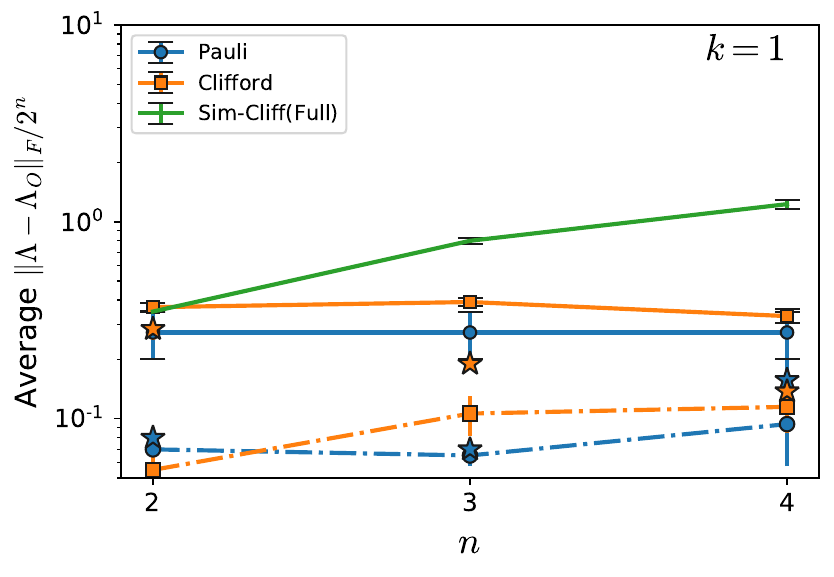}}
    \subfloat[]{\includegraphics[width=0.5\columnwidth]{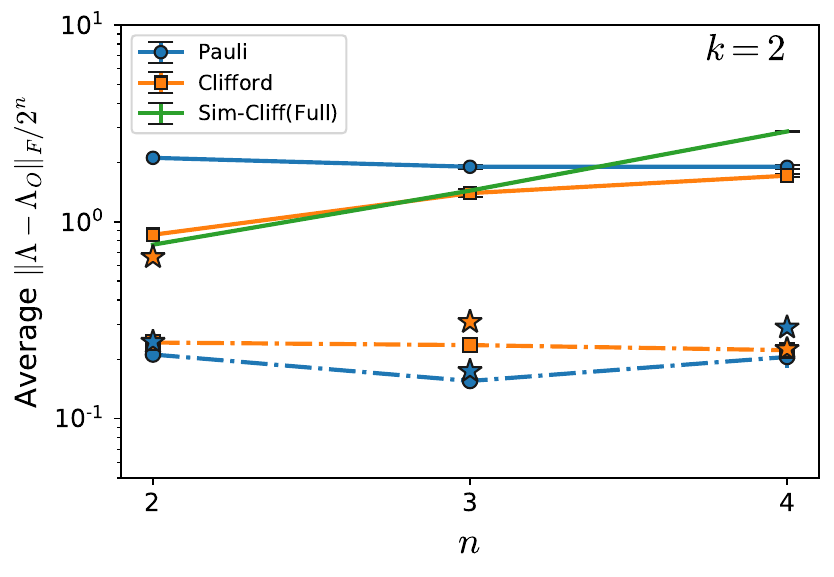}}
    \caption{(a) Left: Normalized Frobenius norm $\Vert \Lambda-\Lambda^O\Vert_F/2^n$ between the unitary Choi matrix $\Lambda$ and a Pauli/Clifford shadow reconstructed Choi matrix $\Lambda^O$. Solid lines represent single repetition shadows with 512 unitaries on IonQ for Pauli/Clifford measurements and in green, simulated $2n$-Clifford circuits; dotted lines includes all data/repetitions collected; dot-dashed lines are purified. The corresponding purified MLE results are shown in blue/orange stars for Pauli/Clifford measurements respectively.
    Right: Frobenius norm distance for various post-processing for $n=3$. The projection/post-processing method is shown in parenthesis above the bar.
    (b) $k=1$ qubit and (c) $k=2$ qubits reduced process reconstruction between the unitary Choi matrix $\Lambda_B$ and a Pauli/Clifford shadow reconstructed Choi matrix $\Lambda^O_B$ using the Frobenius norm. 
    Solid lines represent single repetition shadows with 512 unitaries (including simulated $2n$-Clifford circuits), dotted lines includes all data/repetitions, and dot-dashed lines are $\mathcal{CP}$-projected after the partial trace. The corresponding MLE results are shown in blue/orange stars for Pauli/Clifford measurements respectively. Pauli MLE is done directly in the reduced problem space of $k$ qubit reconstruction.
    }
    \label{fig:l2_norm_version}
\end{figure}
\begin{figure}
    \centering
    \includegraphics[width=\columnwidth]{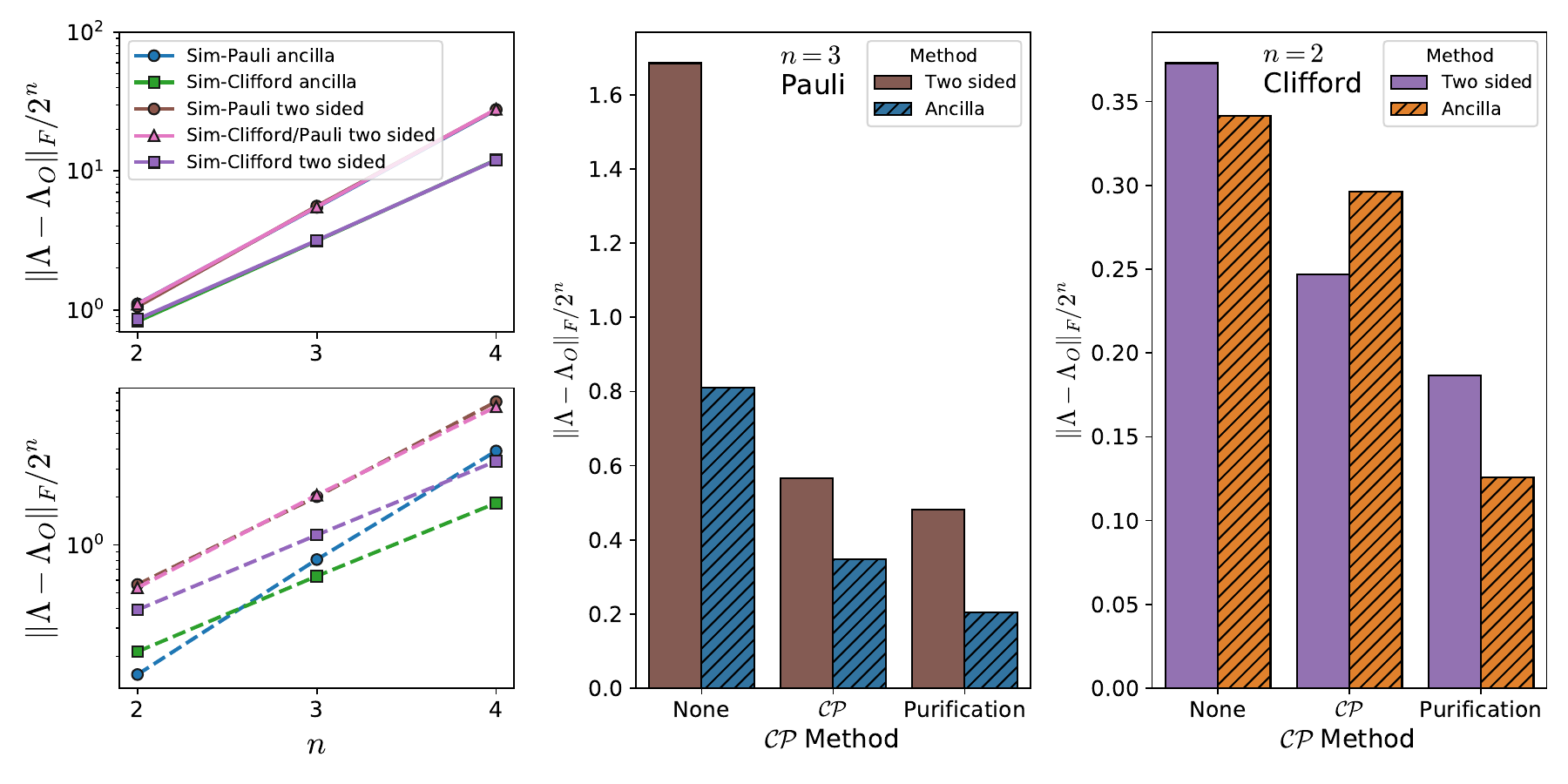}
    \caption{\textit{Left:} Simulated  Frobenius norm $\Vert \Lambda-\Lambda^O\Vert_F/2^n$ scaling for 512 unitaries with a single repetition (top) and 50 repetitions (bottom) between the unitary Choi matrix $\Lambda$ and a reconstruction via Pauli or Clifford ancilla based ShadowQPT or two sided ShadowQPT. `Sim-Clifford/Pauli' corresponds to single qubit Clifford ($k=1$) for $U_L$ and Pauli measurements on the right, which is nearly identical to Pauli two sided measurements. Ancilla results are nearly identical to the two sided results with a single repetition, but with multiple repetitions ancilla simulations (squares) have lower trace distance than their two sided counterparts (circles). \textit{Middle:} Normalized trace distance between (measured) Pauli two-sided ShadowQPT and ancilla ShadowQPT with 512 unitaries and 50 repetitions for a 3 qubit GHZ process. \textit{Right:}   Normalized Forbenius norm between (measured) Clifford two-sided ShadowQPT and ancilla ShadowQPT with 512 unitaries and 50 repetitions for a 2 qubit GHZ process. The ancilla Cliffords are decomposed into $k=2$ as in Fig.~\ref{fig:l2_norm_version} .}
    \label{fig:two_sided_plots_l2}
\end{figure}
\section{Details on Overlap Estimation}

To illustrate Theorem.~\ref{thm:overlap}, we compute $\textup{tr}[\Lambda^\prime(\rho^{in\,T}\otimes\sigma(\theta))]$ for several sets of $\rho$ and $\sigma(\theta)$. The input to the channel, $\rho^{in}$, can be represented by an initial state of which has three categories: $\ket{0}^{\otimes n}$, $\ket{+i}^{\otimes n}$, and $\otimes_j R_x(\phi_j)\ket{0}^{\otimes n}$ states\footnote{where $R_x(\phi_j)$ acts on the $j$-th qubit and $\phi_j=\{0.1717,0.1234,0.9876,0.888\}$}. The first two states are already measured within the Choi matrix, and are effectively measuring the reconstruction of the channel process. After running the circuit in Fig.~\ref{fig:overlap_circuit}, we measure $\ket{0}^{\otimes n}$ or $P(\ket{0}^{\otimes n})$ to obtain the value of the overlap. We only display $\ket{0}^{\otimes n}$ and $\otimes_j R_x(\phi_j)\ket{0}^{\otimes n}$ states in the main work.

We compute the average overlap error in Fig.~\ref{fig:overlaps_appendix} compared to the trace circuit measurements performed on the IonQ device. We use 51 pairs for each $\rho^{in}$ for the shadow reconstructed data  iterating uniformly $\theta\in[0,2\pi]$ and $\approx 5-8$ pairs for each $\rho^{in}$.
Error bars represent the standard error over the overlap dataset. 
We see that our reconstructed shadow data has similar trends to the trace circuit measurement, with a significant range of values.
With purification, we can significantly decrease the error to the noiseless case even compared to the trace circuit, particularly in the $n=4$ case.

\begin{figure}
    \centering
     \includegraphics[width=\columnwidth]{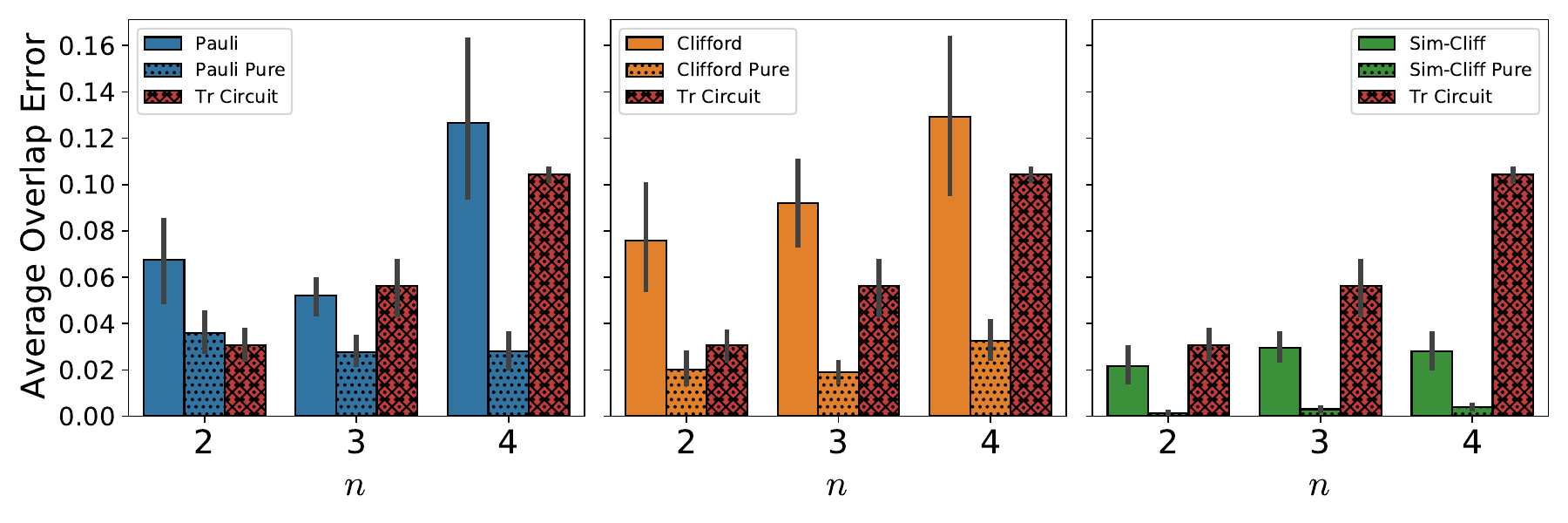}
   
    \caption{ Average overlap error $\langle \Delta\rangle = \langle|w_i^O-w_i|\rangle_i$ for Pauli reconstruction (left), Clifford (middle) and Simulated Clifford data (right) for systems of size $n=2,3,4$. We average over 3 different $\rho^{in}$ and 4 different $\sigma$ each with 51 different angles. For comparison we include trace circuit measurements (Tr circuit) performed on the IonQ using the data shown in Fig.~\ref{fig:overlaps_all}.  }
    
    \label{fig:overlaps_appendix}
\end{figure}
 
Purity of the channel and GHZ state can also be calculated. The Bell-Basis Algorithm (BBA) swap test \cite{Cincio2018} provides an algorithm to determine the overlap between two general quantum states. We compute the overlap of two GHZ states $\textup{tr}[\rho_{GHZ_n}^2]$ where $\rho_{GHZ_n}=\process(\ket{0}\bra{0}^{\otimes n})$ using the BBA algorithm, alongside computing the shadow reconstructed $\textup{tr}[\Lambda^2]/4^n$ in Fig.~\ref{fig:purity}. When measuring the purity, we find a much higher value than our shadow reconstructions. A noiseless quantum computer should produce unit purity for the full GHZ process, these large measurements further motivate our choice to use purification as a post-processing technique. 
Note that purity has a strong dependence on post-processing techniques.

\begin{figure}
    \centering
   
     \includegraphics[width=0.5\columnwidth]{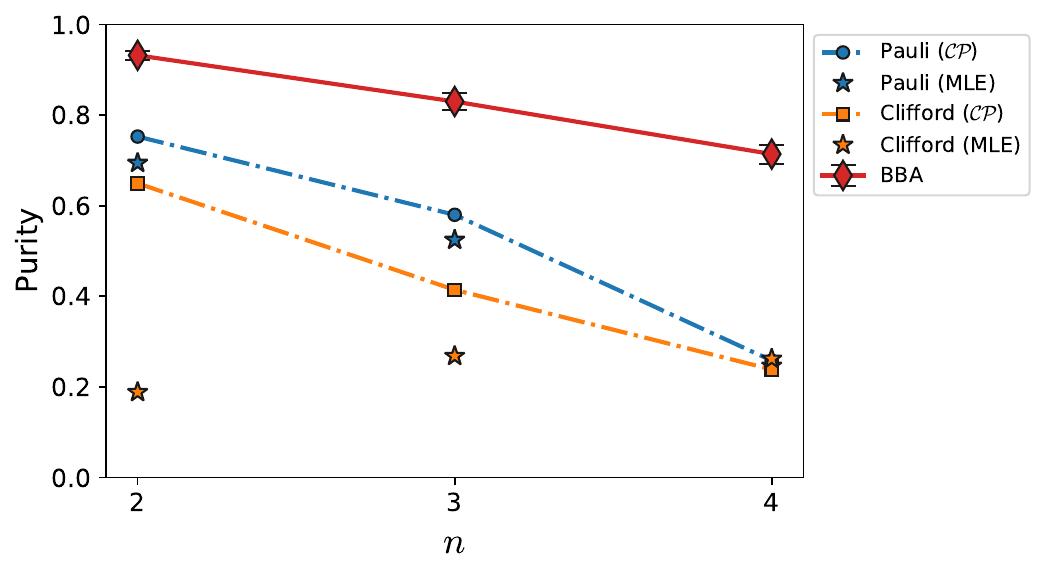}
   
    \caption{Normalized Purity measurement $\textup{tr}[\Lambda^2]/4^n$ of the  $\mathcal{CP}$-projected Pauli/Clifford measured Choi matrices as well as the MLE reconstructions shown as a star. We use the BBA algorithm run directly on the IonQ with 1000 repetitions shown in diamonds to compute $\textup{tr}[\rho_{GHZ_n}^2]$ where $\rho_{GHZ_n}=\process(\ket{0}\bra{0}^{\otimes n})$; error bars are done by batching the measurements into 10 sets.}
    
    \label{fig:purity}
\end{figure}

\section{Hamiltonian Learning Optimal Error Fits}

In Fig.~\ref{fig:app_ham_learning_slice} we show for a fixed number of Pauli measurements (shadows) $N$ the average error vs $t$ of the unitary operator $\exp(-iHt)$. These points have been simulated and averaged over the same 10 trials as in Fig.~\ref{fig:ham_learning_sim}. For a given value of $t$, because of the linear approximation to $\exp(-iHt)$, the error of which is shown in blue, grows as $O(t^2)$ and is the minimum achievable error for a given $t$. ShadowQPT reconstructed couplings $\tilde{c}_i$ are inherently measuring $c_i^{renorm}$ however, with an error as $O(1/t)$ shown in green. Thus as these two approach, the true absolute error between $\tilde{c}_i$ and $c_i$ is shown in orange, which mostly follows either the linear approximation error or renormalized error except near their crossing. The minimum error can then be found for a given $t$, shown in a dashed grey line. We additionally fit this line for $N=1000,1000000$ which are not shown. 

\begin{figure}
    \centering
   
     \includegraphics[width=\columnwidth]{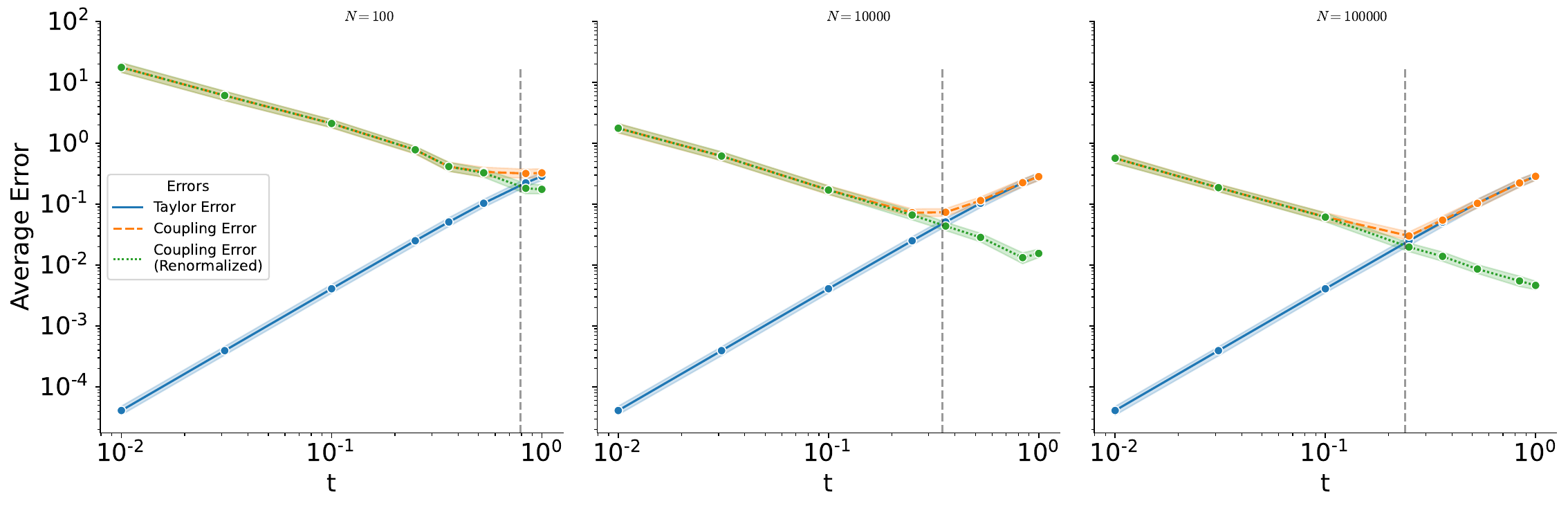}
   
    \caption{a) Hamiltonian learning simulation results for a 1D transverse field Ising model with $n$ sites and random couplings between $[-1,1]$. We average over 10 disorder realizations and use $N=100, 10000,$ and $100000$ random Pauli measurements respectively with no additional post processing. The average error is given by average absolute error $\langle |b_i-c_i|\rangle$ to the original Hamiltonian coupling $c_i$, where $b_i$ is either $c^{renorm}_i$ the renormalized couplings,  $\tilde{c}_i$ the ShadowQPT reconstructed couplings, or $\tilde{c}_i+c_i-c^{renorm}_i$ or the error between ShadowQPT and the renormalized couplings. A dashed line shows the optimal error where the linear approximation error is approximately the same as the renormalized error. }
    
    \label{fig:app_ham_learning_slice}
\end{figure}
\end{document}